\PassOptionsToPackage{dvipsnames}{xcolor}
\documentclass[acmsmall]{acmart}
\usepackage{graphicx} % Required for inserting images
\usepackage{listings}
\usepackage{xspace}
\usepackage[most]{tcolorbox}
\usepackage{fancyvrb}
\usepackage{siunitx}
\usepackage{caption}
\usepackage{enumitem}
\usepackage{amsthm}
\usepackage{multirow}
\usepackage{enumitem}
\usepackage{subcaption}
\usepackage{tikz}
\usepackage{pifont} 
\usepackage{wrapfig}
\usetikzlibrary{positioning, fit, shapes.geometric, arrows.meta, calc}

\usepackage{amsmath, amssymb, array}
\usepackage{comment} 
\usepackage[clock]{ifsym}
\usepackage{sav-science}
\usepackage{makecell}

\tcbuselibrary{skins,breakable}

\usepackage[utf8]{inputenc}
\usepackage[T1]{fontenc}
\usepackage{microtype}

\def\RM#1{{\textcolor{red}{\textbf{RM:} #1}}}
\def\MZ#1{{\textcolor{blue}{\textbf{MZ:} #1}}}
\def\UM#1{{\textcolor{green!80!black}{\textbf{UM:} #1}}}
\def\MO#1{{\textcolor{orange}{\textbf{MO:} #1}}}
\def\BA#1{{\textcolor{violet}{\textbf{BA: } #1}}}
\renewcommand{\BA}[1]{}
\renewcommand{\RM}[1]{}
\renewcommand{\MZ}[1]{}
\renewcommand{\UM}[1]{}
\renewcommand{\MO}[1]{}

\newenvironment{myparagraph}[1]{\vspace{0.1in}\noindent{\bf #1.}}{}

\def\set#1{{\{ #1 \}}}
\def\tuple#1{{\langle #1\rangle}}
\def\nats{\mathbb{N}}

\newcommand{\trust}{\textsc{Trust}\xspace}
\newcommand{\testor}{\textsc{TeStor}\xspace}
\newcommand{\dpor}{\textsc{DPOR}\xspace}

\newcommand{\jmc}{\textsc{JMC}\xspace}

\newcommand{\timeout}{$\VarClock$}
\newcommand{\oom}{\textsf{OOM}}

\newcommand{\opw}{\mathsf{W}}
\newcommand{\opr}{\mathsf{R}}
\newcommand{\op}{\mathsf{op}}
\newcommand{\opaq}{\mathsf{Aq}}
\newcommand{\oprl}{\mathsf{Rl}}

\newcommand{\Val}{\mathsf{Val}}
\newcommand{\Idx}{\mathsf{Idx}}
\newcommand{\Loc}{\mathsf{Loc}}

\newcommand{\Lab}{\mathsf{Lab}}
\newcommand{\Tid}{\mathsf{Tid}}
\newcommand{\Event}{\mathsf{Event}}

\newcommand{\val}{\mathsf{val}}
\newcommand{\idx}{\mathsf{idx}}
\newcommand{\loc}{\mathsf{loc}}
\newcommand{\reg}{\mathsf{reg}}
\newcommand{\lab}{\mathsf{lab}}
\newcommand{\tid}{\mathsf{tid}}

\newcommand{\floc}{\mathsf{loc}}

\newcommand{\flab}{\mathsf{lab}}
\newcommand{\fval}{\mathsf{val}}
\newcommand{\ftid}{\mathsf{tid}}

\newcommand{\egraph}{\mathsf{EG}}
\newcommand{\scg}{\mathsf{SCG}}
\newcommand{\lts}{\mathsf{LTS}}

\newcommand{\init}{\mathsf{init}}
\newcommand{\nextop}{\mathsf{next}}
\newcommand{\pickop}{\mathsf{pick}}
\newcommand{\visit}{\mathsf{visit}}
\newcommand{\exop}{\mathsf{exop}}
\newcommand{\irref}{\mathsf{irreflexive}}
\newcommand{\ismax}{\mathsf{isMax}}
\newcommand{\setmo}{\mathsf{setMo}}

\newcommand{\rf}[1]{{\textcolor{green!60!black}{\mathsf{rf}}_{#1}}}
\newcommand{\fr}[1]{{\textcolor{Aquamarine}{\mathsf{fr}}_{#1}}}

\newcommand{\mo}[1]{{\textcolor{orange}{\mathsf{mo}}_{#1}}}
\def\orangemo{\textcolor{orange}{\mathsf{mo}}}
\newcommand{\momax}[1]{{\orangemo^{\max}_{#1}}}

\newcommand{\po}[1]{{\textcolor{darkgray}{\mathsf{po}}_{#1}}}
\newcommand{\porf}[1]{\textcolor{darkgray}{\mathsf{po}\textcolor{green!60!black}{\mathsf{rf}}}_{#1}}
\newcommand{\seqcon}[1]{{\textcolor{blue}{\mathsf{sc}}_{#1}}}

\newcommand{\bwd}{{\mathsf{Bwd}}}
\newcommand{\inord}{\textcolor{purple}\leq}

\newtheorem{problem}{Problem}
\newcommand{\sharpP}{\#\mathsf{P}}

\newcommand{\mstep}[1][]{\ifthenelse{\equal{#1}{}}{\rightsquigarrow}{\overset{#1}{\rightsquigarrow}}}

\newcommand{\fwdtikz}{%
  \tikz[baseline=(char.base)]{
    \node[shape=circle, draw, fill=lime!50, inner sep=0.15pt] (char) {$\looparrowright$};
  }
}

\newcommand{\bwdtikz}{%
  \tikz[baseline=(char.base)]{
    \node[shape=circle, draw, fill=red!30, inner sep=0.15pt] (char) {$\looparrowleft$};
  }
}

\setcopyright{cc}
\setcctype{by}
\acmDOI{10.1145/3808291}
\acmYear{2026}
\acmJournal{PACMPL}
\acmVolume{10}
\acmNumber{PLDI}
\acmArticle{213}
\acmMonth{6}
\acmSubmissionID{pldi26main-p279-p}
\received{2025-11-14}
\received[accepted]{2026-04-03}

\begin{document}

\title{State Space Estimation for DPOR-Based Model Checkers}

\author{A. R. Balasubramanian}
\orcid{0000-0002-7258-5445}
\affiliation{%
  \institution{Max Planck Institute for Software Systems (MPI-SWS)}
  \city{Kaiserslautern}
  \country{Germany}
}
\email{bayikudi@mpi-sws.org}

\author{Mohammad Hossein Khoshechin Jorshari}
\orcid{0009-0005-7753-9017}
\affiliation{%
  \institution{Max Planck Institute for Software Systems (MPI-SWS)}
  \city{Kaiserslautern}
  \country{Germany}
}
\email{mkhoshechin@mpi-sws.org}

\author{Rupak Majumdar}
\orcid{0000-0003-2136-0542}
\affiliation{%
  \institution{Max Planck Institute for Software Systems (MPI-SWS)}
  \city{Kaiserslautern}
  \country{Germany}
}
\email{rupak@mpi-sws.org}

\author{Umang Mathur}
\orcid{0000-0002-7610-0660}
\affiliation{%
  \institution{National University of Singapore}
  \city{Singapore}
  \country{Singapore}
}
\email{umathur@nus.edu.sg}

\author{Minjian Zhang}
\orcid{0000-0002-5017-2228}
\affiliation{%
  \institution{University of Illinois Urbana-Champaign}
  \city{Champaign-Urbana}
  \country{USA}
}
\email{minjian2@illinois.edu}

\renewcommand{\shortauthors}{A. R. Balasubramanian, M. H. Khoshechin Jorshari, R. Majumdar, U. Mathur, and M. Zhang}
%\acmConference[PLDI '26]{Anonymous Submission}{2025}{Nowhere}

%!TEX root=./main.tex

\begin{abstract}
We study the estimation problem for concurrent programs: given a bounded program $P$, estimate the number of maximal Mazurkiewicz trace–equivalence classes induced by its interleavings. 
This quantity informs two practical questions for enumeration-based model checking: how long a model checking run is likely to take, 
and what fraction of the search space has been covered so far. 
We first show the counting problem is $\#\textsf{P}$-hard even for restricted programs and, unless $\textsf{P}=\textsf{NP}$, inapproximable within any subexponential factor in polynomial time.
Thus, we cannot expect efficient exact or randomized approximation algorithms. 
We give a Monte Carlo approach to find a polynomial-time unbiased estimator: 
we convert a stateless optimal \dpor algorithm into an unbiased estimator by viewing its exploration as a bounded-depth, bounded-width, tree 
whose leaves are the maximal Mazurkiewicz traces. 
A classical estimator by Knuth, when run on this tree, gives an unbiased estimation.
In order to control the variance of the estimation, we apply stochastic enumeration by maintaining a small population of partial paths per depth whose evolution is coupled. 
We have implemented our estimator in the \jmc model checker and evaluated it on shared-memory benchmarks.
We find that with modest budgets, our estimator yields stable 
estimates---typically within a 20\% band---within a few hundred trials, 
even when the state space has $10^{5}$--$10^{6}$ classes. 
We also show how the same machinery estimates model-checking cost by weighting all explored traces, not only the maximal ones. 
Our algorithms provide the first provable poly-time unbiased estimators for counting Mazurkiewicz traces.
%%---a problem of considerable importance when allocating model checking resources.
\end{abstract}

%% Finally, we give a theoretical $\tilde{O}(\sqrt{N})$ estimator with low variance (matching known lower bounds), though its constants preclude empirical use. 




\maketitle
%!TEX root=./main.tex

\section{Introduction}

Given a concurrent program, the \emph{estimation problem} asks: how many different behaviors (up to a certain size) does it have?
It is a natural, counting analogue of decision questions such as reachability.
In this work, we develop systematic solutions for the estimation problem.
As a first step, we must answer what we mean by a behavior, since 
different model checkers induce different notions of behaviors.
A naive interleaving-based model checker treats every thread schedule as a distinct execution.
In contrast, a checker based on partial-order reduction \cite{GodefroidThesis, spin, dpor, verisoft,pop-dpor,optimal-dpor,nidhugg,trust} 
collapses all schedules within the same equivalence class into a single behavior, 
thereby exploring a significantly smaller space.
Additional heuristics may reduce the space further, e.g., by defining coarser equivalences based on symmetry.

We focus on behaviors induced by Mazurkiewicz's trace equivalence, where two executions are considered equivalent
if one can be obtained from the other by repeated swapping of adjacent commutative operations.
Our choice is guided, on the one hand, by the success of this notion in model checking
\cite{spin, verisoft, dpor,optimal-dpor, trust}
and, on the other, by the recent emergence of many \emph{optimal} exploration algorithms that guarantee that each equivalence class is visited exactly 
once \cite{optimal-dpor, trust, must}.
In this setting, the estimation problem asks:
\begin{quote}
{\bf Exact Estimation problem:} Given a bounded program \(P\), how many Mazurkiewicz equivalence classes do its interleavings induce?
\end{quote}

For an optimal model checking algorithm, the estimation problem is directly related (but not identical) to the cost of model checking with partial order reduction.
Thus, solving the estimation problem can help answer two
high-level quantitative questions that arise when using a model checker:
\emph{``How long should we expect the model checker to run?''} and
\emph{``What fraction of the state space has the model checker covered so far?''}
Reliable answers to these questions allow users of model checking tools to predict verification effort (see, e.g., discussions in \cite{WangCGMK18,gator}).
%
% Reliable answers to these questions have the potential to significantly enhance
% the usability and adoption of model checkers \cite{gator,jpf}, by enabling software teams
% to allocate computational budgets more effectively, predict the remaining
% verification effort, and meaningfully compare competing verification tools.
% As model checkers become an integral part of the modern software development toolkit
% \cite{spin, cbmc,jbmc, p-lang, trust,optimal-dpor,nidhugg,condpor,kani,must,jpf,cpachecker,cdschecker,concuerror,rInspect},
% these resourcing questions become increasingly important.

We first ask whether the above estimation problem can be solved
tractably (i.e., in polynomial time in the size of the program).
The situation is bleak:
we show that the problem is \#P-hard even for highly restricted concurrent programs.
Informally, the \#P-hardness result implies that the task is at least as hard as
counting the satisfying assignments of a Boolean formula (\#SAT), a canonical
counting problem for which no efficient algorithm is believed to exist.
In fact, we show that unless \textsf{P} = \textsf{NP}, we cannot even \emph{approximate} the number of equivalence classes
up to a subexponential factor!
Thus, neither exact nor even moderately precise approximate
counting is feasible in general.

Although these hardness results rule out efficient exact or approximate
counting, they do not prevent us from obtaining statistical or \emph{Monte Carlo}
estimates of the size of the behavioral-space.
The statistical formulation of our problem asks for an \emph{unbiased estimator}:
\begin{quote}
{\bf Unbiased statistical estimation problem:} Given a bounded program \(P\), give a polynomial-time randomized procedure whose expected value is the number of Mazurkiewicz-equivalence classes.
\end{quote}

A natural starting point for deriving such an estimator is Knuth's classical
estimator \cite{knuth75}, which approximates the number of leaves in a
combinatorial search tree by performing a randomized exploratory
walk, recording the branching factors encountered along each walk, and
computing a weighted product; this yields an unbiased estimator of the total
number of leaves.

However, the space of Mazurkiewicz traces is a directed acyclic graph (DAG) but not a tree:
different schedules of events can merge into the same equivalence class.
On a DAG, Knuth's procedure is no longer an unbiased estimator.
Leonard Pitt generalized Knuth's procedure so it works for DAGs\cite{pitt}.
While Pitt's estimator is unbiased, its variance is too high to be of practical use: 
even when a program has exactly one maximal trace, the variance can be exponential.

The main main insight we employ in this work that the \trust stateless optimal \dpor algorithm \cite{trust}
can be converted into an unbiased estimator!
This follows because the \trust algorithm enjoys several crucial properties.
First, its optimality guarantee implies that the underlying state space explored by the algorithm defines a tree structure.
Second, since it is stateless, the next decisions depend solely on the current state and not on the history of the search.
Third, each exploration of the algorithm is bounded polynomially in the size of the program.

These observations together imply that the underlying state space traversed by \trust is a tree, whose depth and width are bounded by a quadratic function of the program size,
and whose leaves are exactly the maximal Mazurkiewicz traces.
Thus, estimating the leaves of this tree solves the estimation problem,
and we can apply the Knuth estimator to this (implicitly constructed) tree to obtain an unbiased estimator.
The algorithm is surprisingly simple: we run \trust, but instead of systematically exploring all successors, we randomly pick one successor at each stage.
At the end, we output the product of the number of choices we had along the path.
Note that we do not exhaustively explore all traces; this number can grow exponentially with program size.

While the estimator is unbiased, its variance can still be exponentially large,
reflecting the inherent hardness of approximation.
However, in many practical situations, a large variance often occurs not because of inherent computational complexity but because the underlying tree is ``skewed''.
In the \trust algorithm, some explored paths are linear in the size of the program while others can
be quadratic---a random walk procedure is much more likely to visit the ``shallow'' leaves than the ``deep'' ones.

In order to reduce variance, we use a technique from Monte Carlo estimation called \emph{stochastic enumeration} \cite{rubinstein, vaisman2017} or weighted ensemble \cite{huberkim}.
Instead of one random walk, as in Knuth's estimator, we maintain a small population of partial paths at each depth of the tree.
Whenever a path ends, the population is redistributed to reflect the structure.
The overall stochastic procedure ensures that no single highly unbalanced subtree dominates the
estimate.
This preserves the unbiasedness of estimation while substantially smoothing out the contribution
of deep but large subtrees.
As our evaluation shows, on a diverse set of bounded concurrent programs, the
resulting estimator attains stable estimates with a modest number of samples,
suggesting that this variance-reduction strategy is well suited to realistic
model-checking workloads.

We have implemented our estimator in the \dpor-based stateless Java model checker JMC and evaluated it on a suite of bounded concurrent programs, including concurrent data structures 
and synthetic benchmarks designed to exercise different concurrency patterns.

We show that the estimator based on \trust, combined with stochastic enumeration, consistently achieves stable estimates within a small number of trials---typically a few hundred---even 
for behavior spaces ranging from about \(10^{5}\) to \(10^{6}\)
executions, with relative error usually within a 20\% band.
On the other hand, each aspect of our estimation algorithm remains 
crucial: on these benchmarks, a DAG-based estimator or an estimator that does not perform stochastic enumeration fails to converge.

We note that the combination of \trust, tree-estimation, and stochastic enumeration exploits a number of subtle but important properties of the \trust algorithm.
First, we crucially leverage the fact that the underlying optimal \dpor algorithm is history independent.
In an optimal model checker like Nidhugg~\cite{nidhugg}, every execution graph is visited through a unique path, implicitly inducing a tree.
However, the task of sampling paths in such a tree (in the style of Knuth's estimator) can be non-trivial.
This is because, to decide whether a given execution graph is an extension of another, the model checker may have to consult the history of the exploration, 
 paying exponential cost in the worst case, undermining the poly-time sampling property.
Second, one could try to reduce variance using a number of heuristics that perform importance sampling in various ways.
However, in practice, the performance of importance sampling is very sensitive to good estimates of the relative sizes of different options.
These estimates are difficult to come by for practical benchmarks.
%% other than by computationally expensive nested Monte Carlo searches.

We also explore the theoretical question of whether there is a good approximate counter that uses \emph{subexponential} time.
We find that, surprisingly, there is an unbiased estimator with low variance that runs in time \(\tilde{O}(\sqrt{N})\), where \(N\) is the number of possible maximal traces (here $\tilde{O}$ hides polynomial factors)!
While theoretically optimal (there is an \(\Omega(\sqrt{N})\) lower bound \cite{Stockmeyer85}), the polynomial factors appearing in the algorithm are too large for empirical evaluation.

\myparagraph{Contributions}{
In summary, this paper makes the following contributions: 
\begin{itemize}[topsep=0pt, partopsep=0pt, itemsep=0.25em, parsep=0pt]
  \item  
  We formulate the estimation problem for Mazurkiewicz traces and show that it is \#P-hard 
  and
  inapproximable within any subexponential factor unless \textsf{P} = \textsf{NP}.

  \item 
  We show how the stateless optimal \dpor algorithm \trust can be modified to produce a poly-time unbiased estimation, and show how a population-based
  estimation strategy reduces the variance of the estimation (\cref{trust_sec}, \cref{sec:se}).
  
  \item We evaluate our estimation procedure on a wide variety of shared-memory concurrency benchmarks in the context of the \jmc model checker (\cref{sec:experiment}).
  The summary of our experiments is that our estimator provides accurate and robust estimates of the number of equivalence classes, as well as the cost of model checking,
  while using a small population budget and a modest number of samples.
\end{itemize}
}

%For space reasons, the detailed proofs be found in the full version of the paper~\cite{balasubramanian2026statespaceestimationdporbased}.

%\input{intro}
%!TEX root=./main.tex

\section{The State Estimation Problem}
\label{sec:problem}
\subsection{Programming Model and Partial Order Semantics}

We consider a multithreaded programming model, in which a set of \emph{threads} concurrently read and write \emph{shared locations}.
Each thread performs sequential computation, including reads and writes to shared locations and deterministic sequential control flow.
We refrain from providing a concrete syntax.
Instead, following \cite{trust}, we define a partial-order semantics using \emph{execution graphs}.
Intuitively, an execution graph consists of a collection of reads and writes by the threads to the shared locations,
ordered by program order for each thread, and
detailing for each read to a location, the unique write from which the read gets its value. 
Such a graph precisely captures one equivalence class under Mazurkiewicz trace equivalence of a program. 

Let us define execution graphs formally.
We assume a set of locations $\Loc$ and a set of values $\Val$.
An event, \( e \in  \Event\), is either the initialization event $\init$, or a thread event
\(\tuple{t, i, \ell}\) where $t \in \Tid$ is a thread identifier unique for each thread, 
$i \in \Idx \defeq \nats$ is an index inside each thread, and
$\ell \in \Lab$ is a label that is either a \emph{write label}
$\opw(l, v)$ or a \emph{read label} $\opr(l, v)$ for a location $l\in\Loc$ and value $v\in \Val$.
When applicable, the functions $\tid$, $\idx$, $\loc$, $\val$, and $\op$, return the thread identifier, index,
location, value of an event, and type of an event's label respectively. 
An event is called a read event if its label is a read label and an event is called
a write event if it is either $\init$ or it has a write label.
%We use $G.\opr \triangleq \{e \mid e \in G.E \wedge \op(e.\ell) = \opr\}$ to denote the set of all read events of $E$ and $G.\opw \triangleq \{e \mid e = \init \vee (e \in G.E \wedge \op(e.\ell) = \opw)\}$ to denotes the set of all write events. Also,  we use subscripts to further restrict those sets like $G.\opw_l \triangleq \{e \mid e = \init \vee (e \in G.\opw \wedge \loc(e) = l)\}$.

An  \emph{execution graph} $G = (E, \rf{}, \mo{})$ consists of: 
%\UM{Why do we need $\inord$?} \MO{Because of \trust algorithm and the semantics we defined based on it in section 4.} \BA{I think Umang is right. Introducing $\inord$ here makes everything a tree, because two different reads scheduled in two different ways give rise to two different insertion orders. It needs to be removed for $\mathcal{T}(P)$. I have done that and I will reintroduce it during TruST.} \MO{I see. Then we can simply say two graphs are equal up to the the $\inord$ component. This definition also work in Trust semantics. Otherwise we should define two version of execution graphs in the text.}
\begin{enumerate}
    \item  a set $E$ of events, containing $\init$, such that no two events have both the same thread identifier and the same index; We use $G.\opr$ to denote the set
    of all read events in $E$ and $G.\opw$ to denote the set of all write events in $E$.
    Also,  we use subscripts to further restrict those sets like $G.\opw_l$ is the set of all write events in $E$ whose location is $l$.
    %\item a total order $\inord$ on $E$, called \emph{insertion order}, represents the order in which events were incrementally added to the graph.
    %\item a strict partial order $\po{}\subseteq E\times E$, called \emph{program order}, that totally orders events from the same thread.  
    %We use $\po{}^{\max}(t)$ to denote the maximal event of thread $t$ w.r.t.\ $\po{}$.
    \item a partial order $\rf{} \subseteq G.\opw \times G.\opr$, called the \emph{reads-from} relation, that relates each read event to the 
    (unique) write event to the same location from which the read gets its value. Formally,
    \begin{enumerate}
        \item $\forall e \in G.\opr : \exists e' \in G.\opw_{\loc(e)} : (e', e) \in \rf{} \wedge \val(e) = \val(e').$ 
        %(Every read event $e$ is related to at least one write event $e'$ at the same location). \BA{Added the description in English.}
        \item $\forall (e_1, e_2) \in \rf{} : \op(e_2.\ell) = \opr \wedge (e_1 = \init \vee (\op(e_1.\ell) = \opw \wedge \loc(e_1) = \loc(e_2))$. 
        %(Every related pair $(e_1,e_2)$ satisfies that $e_2$ is a read event and $e_1$ is either $\init$ or a write event to the same location as $e_2$). \BA{Added the description in English.}
        \item $(e_1,e_2)\in \rf{} \wedge (e_3,e_2)\in \rf{} \implies e_1=e_3$. 
        %(Every read event is related to at most write event). \BA{Added the description in English.}
    \end{enumerate}
    \item a strict partial order $\mo{} \subseteq \bigcup_{l \in \Loc}  G.\opw_{l} \times G.\opw_{l}$, called the \emph{modification order}, which is a total order on 
    the write events on the same location. %We use $\mo{}^{\max}(l)$ to denote the latest write on the memory location $l$. 
    Formally, 
    \begin{enumerate}
        \item $\forall (e_1, e_2) \in \mo{} : e_1 \neq e_2 \wedge e_1,e_2 \in G.\opw \wedge (\loc(e_1) = \loc(e_2) \vee e_1 = \init).$ 
        \item $\forall l \in \Loc : \forall e_1,e_2 \in G.\opw_{l} : e_1 \neq e_2 \implies (e_1,e_2)  \in \mo{} \text{ or } (e_2,e_1) \in \mo{} \text{, but not both.}$
    \end{enumerate}
\end{enumerate}
We define a \emph{program order} $\po{}$ relation as follows:
$$
\po{} \triangleq \{(\init, e) \mid e \in E \setminus \init\} \cup \{(e_1, e_2) \mid e_1,e_2 \in E \wedge \tid(e_1) = \tid(e_2) \wedge \idx(e_1) < \idx(e_2)\}
$$
An execution graph is \emph{sequentially consistent} if $\seqcon{} \triangleq (\po{} \cup \rf{} \cup \mo{} \cup \fr{})^+$ is irreflexive, where $\fr{} \triangleq (\rf{}^{-1} ; \mo{}).$
The semantics of $P$ is the set of its sequentially consistent execution graphs.
We write $G|_{F}$ for the restriction of the graph $G$ to some set $F$ of its nodes and $G\setminus F$ for $G|_{E\setminus F}$.
%We write $G' = G\setminus{F}$ for the graph $G'$ obtained by removing a set of events $F$. 
%Formally,
%\begin{enumerate}
%    \item $G'\!.E = G.E \setminus{F}$.
%    \item $G'\!.R = G.R \setminus{\{(e_1,e_2) \mid (e_1,e_2) \in G.R \wedge (e_1 \in F \vee e_2 \in F)\}}$ for $R \in \{\rf{}, \mo{}\}$.
%\end{enumerate}
%We also use $G' = G|_{F}$ to denote the restriction of an execution graph $G$ to a set of events $F$. Formally, $G|_{F} = G \setminus{F'}$ where $F = G.E \setminus{F'}$.

\subsection{A Transition System on Execution Graphs}

%\UM{This section, as such, is not a `real' preliminary. It is in fact only required to phrase the problem in a manner that is amenable to Pitt's algorithm. In fact, our other solution (namely TESTOR) doesnt rely on this way of modeling the problem. Given this, %I think this can move to Section 3. }
%\RM{how will you define $C(P)$ without the transition system?}
%\UM{True, it may become more verbose and also repetitive otherwise. How about we define the transition system, but refrain from pointing out that it is a DAG right away? In Sec 3 we can point out that it is a DAG}

The operational semantics of the programming language defines a transition system on execution graphs.
Instead of a detailed syntax, we assume that there is a function $\nextop_P(G)$ which takes
an execution graph and returns a 
set of possible next events from the program $P$: these are the next operations
of the unblocked threads, after the events in $G$ have transpired.
If $\nextop_P(G)$ is empty, then there are no remaining operations and the program has terminated.

We define a labeled transition system ($\lts$) $\mathcal{T}(P)$ on the set of sequentially consistent execution graphs.  
First, for any execution graph $G$ and location $l$, we define $G.\momax{l}$ as the unique maximal event $e \in G.\opw_l$ in the $G.\mo{l}$ ordering.
Intuitively, $G.\momax{l}$ is the last write to the location $l$ in $G$. 
The labeled transition system $\mathcal{T}(P) = (S, S_0, \xrightarrow{})$ is defined as:
\begin{enumerate}
    \item $S$ is the set of sequentially consistent execution graphs.
    \item $S_0$ is the execution graph $G_{\init} = (\set{\init}, \emptyset, \emptyset)$ consisting of the singleton node $\init$.
    \item $\xrightarrow{} = \bigcup_{t \in \Tid}\{\}\xrightarrow{t}\}$, where $\xrightarrow{t}$ denotes the local transition relation of thread $t \in \Tid$.  
\end{enumerate}
There is an edge $G_1 \xrightarrow{t} G_2$ if
there is an $e\in \nextop_P(G_1)$ with $t = \tid(e)$, and
\begin{enumerate} 
\item $G_2.E = G_1.E \cup \set{e}$.
%\item $G_2.\inord = G_1.\inord \cup ~ \{(e', e) \mid e' \in G_1.E\}.$
%\item $G_2.\po{} = G_1.\po{} \cup \set{(\po{}^{\max}(t), e)} \cup \set{(e', e)\mid \po{}(e', \po{}^{\max}(t)}$.
\item $G_2.\rf{} = \begin{cases} 
G_1.\rf{} & \op(e.\ell) = \opw \\ 
G_1.\rf{} \cup \{(G_1.\momax{\loc(e)}, e)\} & \text{otherwise} 
\end{cases}$
\item $G_2.\mo{}{} = \begin{cases} 
G_1.\mo{}{} & \op(e.\ell) = \opr \\ 
G_1.\mo{} \cup \{(e', e) \mid e' \in G_1.W_{\loc(e)}\} & \text{otherwise} 
\end{cases}$
%\item If $\op(\ell) = \opw$, then $G_2.\rf{}=G_1.\rf{}$. If $\ell = \opr(l, v)$, then there exists an event $\hat{e} \in G_1.E$ such that  $\hat{e} = \init$ or $\hat{e} = \tuple{\cdot, \cdot, \opw(l, v)}$ such that $G_2.\rf{}=G_1.\rf{}\cup\{(\hat{e},e)\}$.
%\item If $\op(\ell) = \opr$ then  $G_2.\mo{}=G_1.\mo{}$.  If $\op(\ell)=\opw$, then  there exists an event $\hat{e} \in G_1.E$ such that $\hat{e} = \init$ or $\hat{e} = \tuple{\cdot, \cdot, \opw(l, v')}$ and $G_2.\mo{}=G_1.\mo{}\cup\{(\hat{e},e)\} \cup \set{(e', e)\mid (e', \hat{e})\in G_1.\mo{}}$.
\end{enumerate}
An execution graph is \emph{reachable} if $G_{\init} \xrightarrow{\cdot}^* G$.
An execution graph is \emph{maximal} if it has no successors,
i.e., $\nextop_P(G) = \emptyset$.
We write $C(P)$ for the number of maximal reachable execution graphs in $\mathcal{T}(P)$.
%
% Since we assume every run of $P$ is bounded, the transition system
% $\mathcal{T}(P)$ is finite, acyclic, and every path 
% from $G_{\init}$ ends in a maximal execution graph.

Note that $\mathcal{T}(P)$ is a DAG but not necessarily a tree: two 
different paths in the exploration can reach the same execution graph.
For example, if two threads each read a variable, the two possible thread schedules both lead to the same execution graph.

\subsection{Counting Execution Graphs: Hardness}

When counting execution graphs, we assume that all runs of $P$ are bounded (by some number $n$).
To simplify notation, we will assume in the following that we have ``compiled'' the number $n$ 
into the program (e.g., by unrolling loops an appropriate number of times) and 
therefore, every run of the program is already bounded by the size of the program. 
Thus, we will omit $n$ in our problem formulation and results, and assume $n = |P|$ 
bounds the number of events in any execution of $P$.
Our goal is to count the number $C(P)$ of maximal reachable execution graphs of $P$.
To this end, we define the following counting problem.

\begin{problem}[The Counting Problem]
    \label{prob1}
    Given as input a program $P$, output the number of maximal reachable execution graphs $C(P)$ in $\mathcal{T}(P)$.
\end{problem}

%To simplify notation, we will assume in the following that we have ``compiled'' the number
%$n$ into the program (e.g., by unrolling loops an appropriate number of times) and therefore, every execution of the program
%is already bounded by the size of the program.
%Thus, we will omit the input $n$ in our results and assume $n = |P|$.

We first show that the counting problem above is complete for the class $\sharpP$ (proof in \cref{theory:A}):
\begin{theorem}\label{thm:counting-problem-hardness}
   The Counting Problem is $\sharpP$-complete. 
   $\sharpP$-hardness holds even if either the number of threads or the number of shared locations is constant.
\end{theorem}

%\RM{check the sentence below:}
%In fact, the problem is $\sharpP$-hard already when either the number of threads or the number of locations (but not both) is fixed. \MZ{This is not correct. We showed that %if threads and locations are fixed but local register is not then the problem is also $\sharpP$-hard.
%We proved that when everything is constant it has a polynomial time algorithm because there is log-space transducer. }

Recall that $\sharpP$ is the complexity class associated with counting the number
of accepting computations of a non-deterministic Turing machine (see, e.g., \cite{aroraBarak}). 
In particular, the problem of counting the number of satisfying assignments of a SAT formula is 
$\sharpP$-complete. 
Hence, $\sharpP$-complete problems are at least as hard as NP-complete problems, which 
gives strong evidence against the existence of any polynomial-time
algorithms for Problem~\ref{prob1}. 
In light of this, we can relax the requirement of the Counting Problem 
so that we only need an approximate count, rather than an exact count.

\begin{problem}[The Approximate Counting Problem]
    \label{prob2}
    A randomized algorithm $A$ is an $(f(n), \rho)$-approximate counter if, for every program $P$ of size $n$, we have
    \[
    \operatorname{Pr}[C(P)/f(n) \leq A(P, \rho) \leq C(P)f(n)] \geq 1 - \rho
    \]
%    Given as input a program $P$ of size $n$, output a number $N$ such that $N/2^{n^{1-\epsilon}} \le C(P) \le N \times 2^{n^{1-\epsilon}}$
%    for some constant $\epsilon > 0$ independent of the input.
\end{problem}

Note that since $C(P)$ is always at most an exponential function in the size of $P$,
the least non-trivial approximation that we can ask for is an approximation
with a sub-exponential ratio. 
However, even for this very weak approximation ratio, we show that it is highly unlikely that a (randomized) polynomial-time algorithm exists, by means of the following theorem
(proof in \cref{theory:B}):

%\begin{theorem}
 %Let $d>0$ be any constant, and let $C$ be the function that computes Problem \ref{prob1}. There is no polynomial algorithm which always output an approximation $\mathcal{A}$ for $C(P)$ such that
  % $\mathcal{A}/2^{n^{d}}\leq  C(P)\leq \mathcal{A}*2^{n^{d}}$ for every concurrent program $P$ unless P=NP.
%\end{theorem}
\begin{theorem}\label{thm:approximate-counting-problem-hardness}
 Unless RP = NP, there is no randomized polynomial-time $(2^{n^{1-\epsilon}}, \rho)$-approximate 
 counter for $\epsilon > 0$ independent of the input.
\end{theorem}

Here RP is the class of problems which admit randomized polynomial-time algorithms;
(see \cite{aroraBarak}). Hence, this theorem puts hard limits on the Approximate Counting Problem. 
Furthermore, unless P = NP, a similar argument shows that there is no deterministic polynomial-time
algorithm that gives a sub-exponential approximation to $C(P)$. 
%(We defer the proofs of both of these theorems to ~\cref{theory}).

In light of these hardness results, we now present some polynomial-time randomized algorithms that offer an \emph{unbiased estimate} of $C(P)$, i.e., the \emph{expected value} returned by them will be $C(P)$.

% In the rest of the paper, we begin by presenting two Monte Carlo estimation methods that are provably unbiased, offering principled solutions to the estimation problem. We then establish the hardness of exact counting to provide a clearer theoretical understanding of the problem’s complexity. Finally, we evaluate the practical effectiveness of our proposed estimators on a set of concurrent programs by comparing them against exact counts obtained via a model checker.

%\begin{problem}[Approximate Counting with Additive Errors]
%    \label{prob3}
%\RM{TODO FIX}
%    Given as input a program $P$ of input size $n$, output a number $N$ such that $| C(P) - N| \le poly(n, 1/\epsilon)$
%    for some constant $\epsilon > 0$ independent of the input.
%\end{problem}

%!TEX root=./main.tex

\section{Poly-Time Unbiased Estimation Algorithms} 
\label{sec:pestor}
Let $P$ be a program and $\mathcal{T}(P)$ its transition system.
A randomized algorithm $A$ that maps programs $P$ to natural numbers 
is an \emph{unbiased estimator} for $C(P)$ if
the expected value of $A$ on the input $P$, over all of its internal random choices, is $C(P)$.
It is poly-time if it runs in time polynomial in $|P|$.

In this section, we present two poly-time unbiased estimators for $C(P)$.
Both of them have the nice property that they are simple to describe and work directly on $\mathcal{T}(P)$, which is quite natural and provides a convenient model of the space of Mazurkiewicz trace equivalence classes. However, as we shall see, both of them have high variance even on simple examples. 
In the next section, we will see how we can reduce the high variance, by building upon ideas from this section.

%Note that due to the hardness result of \Cref{thm:approximate-counting-problem-hardness}, it follows that the variance of any poly-time unbiased estimator will be large (unless RP = NP). 

\subsection{Knuth Estimator}
We now present our first estimator, which works \textit{as long as} $\mathcal{T}(P)$ \textit{is a tree}. It is based on Monte Carlo estimation.
The algorithm, given below, was described by Knuth \cite{knuth75}.

\begin{verse}
\label{algo:K}
\textbf{Algorithm K} \cite{knuth75}. 
Given a tree $T$, run the following randomized procedure.
Starting at the root,
for each node $n_i$ with $d_i$ children, choose one of its $d_i$ children uniformly at random and move there, until reaching a leaf $n_k$.
Stop and output $d_1 \cdot d_2 \cdot \ldots \cdot d_{k-1}$ as an estimate of the number of leaves of $T$.
\end{verse}

\begin{proposition}{\cite{knuth75}}
If $\mathcal{T}(P)$ is a tree, \textbf{Algorithm K} is a poly-time unbiased estimator for $C(P)$.
\end{proposition}

We observe that each path $n_1 \ldots n_k$ of the tree from its root $n_1$ to a leaf $n_k$ is reached by the random walk with probability $p(n_1,\dots,n_k) = 1/(d_1d_2\ldots d_{k-1})$. Since the output of our algorithm for the path $n_1,\dots,n_k$
is exactly $1/p(n_1,\dots,n_k)$, it follows that the expected value of our output
is simply the number of all possible paths $n_1,\dots,n_k$,
which is exactly the number of leaves.
\begin{tcolorbox}[
    colback=SkyBlue!5,
    colframe=blue!20,
    boxrule=1pt,
    arc=2pt,
    width=\textwidth,
    breakable
]

\begin{example}
Consider the program
\begin{equation*}
  \tag{\textsc{r+w+w}}
  \label{ex:r+w+w}
  \inarr{ \textcolor{Red}{\readInst{b}{x}} }{ \textcolor{blue}{\writeInst{x}{1}} }{ \textcolor{Dandelion!60!black}{\writeInst{x}{2}} }
\end{equation*}
\cref{fig:r+w+w-t} shows the transition system $\mathcal{T}(\mathsf{R+W+W})$.
The estimate for any random path through the tree is $3\times 2\times 1 = 6$, and the expectation is also 6. 
\qed
\end{example}
\input{graphs/RWW-T}
\end{tcolorbox}

To show that \textbf{Algorithm K} runs in poly-time, it suffices to show that the number of children of a node can be computed in poly-time. This is possible, since  the number of children of each execution graph is bounded by the number of available threads at an execution graph.
Hence, \textbf{Algorithm K} is easy to implement in a random tester: it only requires that we intercept the scheduling events 
and pick one of the available threads at random. 

Even when $\mathcal{T}(P)$ is not a tree, we can run \textbf{Algorithm K} as follows: Sample a path along $\mathcal{T}(P)$ and return the product of the number of available choices at each step. 
However, since this strategy ignores the 
DAG structure of $\mathcal{T}(P)$, it can substantially overcount the number of maximal execution graphs, and is not an unbiased estimator.

\begin{tcolorbox}[
    colback=SkyBlue!5,
    colframe=blue!20,
    boxrule=1pt,
    arc=2pt,
    width=\textwidth,
    boxsep=5pt,
    left=5pt,
    right=5pt,
    top=5pt,
    bottom=5pt
]

\noindent
\begin{minipage}[b]{0.65\textwidth}
\begin{example}
Consider the program  
\begin{equation*}
  \tag{\textsc{r+r+r}}
  \label{ex:r+r+r}
  \inarr{ \textcolor{Red}{\readInst{a_1}{x}} }{ \textcolor{blue}{\readInst{a_2}{x}} }{ \textcolor{Dandelion!40!black}{\readInst{a_3}{x}} }
\end{equation*}

There is exactly one maximal execution: all reads read from $\init$.
The transition system $\mathcal{T}(\text{\textsc{r+r+r}})$ is a DAG with exactly one maximal execution.
\textbf{Algorithm K} estimates 6 executions, one for each ordering of the threads.
In general, if there were $n$ reads, \textbf{Algorithm K} would estimate $n!$ executions. \qed
\end{example}
\end{minipage}%
\hfill
\begin{minipage}[b]{0.40\textwidth}
\centering
\vspace{0pt}
\begin{tikzpicture}[
    scale=0.8,
    transform shape,
    >=stealth,
    node distance=2.2cm and 2.5cm,
    hypernode/.style={draw, rectangle, rounded corners=4pt, very thick, inner sep=3pt, font=\small\bfseries},
    edge label/.style={midway, very thick,
    fill=white, font=\footnotesize, inner sep=1pt},
    edge-t1/.style={->, very thick, Red},
    edge-t2/.style={->, very thick, blue},
    edge-t3/.style={->, very thick, Dandelion!60!black}
]

% Top layer: G0
\node[hypernode] (G0) at (0,0) {$G_0$};

% Second layer
\node[hypernode] (G1) at (-2,-1) {$G_1$};
\node[hypernode] (G2) at (0,-1) {$G_2$};
\node[hypernode] (G3) at (2,-1) {$G_3$};

% Third layer
\node[hypernode] (G4) at (-1.5,-2) {$G_4$};
\node[hypernode] (G5) at (0,-2) {$G_5$};
\node[hypernode] (G6) at (1.5,-2) {$G_6$};

% Last layer
\node[hypernode] (G7) at (0,-3) {$G_7$};

% Edges from G0
\draw[edge-t1] (G0) -- (G1);
\draw[edge-t2] (G0) -- (G2);
\draw[edge-t3] (G0) -- (G3);

% Edges from G1
\draw[edge-t2] (G1) -- (G4);
\draw[edge-t3] (G1) -- (G5);

% Edges from G2
\draw[edge-t1] (G2) -- (G4);
\draw[edge-t3] (G2) -- (G6);

% Edges from G3
\draw[edge-t1] (G3) -- (G5);
\draw[edge-t2] (G3) -- (G6);

% Edges to G7
\draw[edge-t3] (G4) -- (G7);
\draw[edge-t2] (G5) -- (G7);
\draw[edge-t1] (G6) -- (G7);

\end{tikzpicture}
\label{hypergraph-rrr}
\end{minipage}%
\end{tcolorbox}

\subsection{Pitt Estimation}

The above disadvantage of \textbf{Algorithm K} on DAGs can be rectified:
Pitt \cite{pitt} showed that \textbf{Algorithm K} can be adapted to estimate the number of sinks in a DAG. 
The key idea is to correct for overcounting at each step,  by dividing  the multiplicative weight with the  number of incoming edges to the current node.

\begin{verse}
\textbf{Algorithm P} \cite{pitt}.
Given a DAG $T$, run the following randomized procedure.
Starting at the root, for each node $n_i$ with $d_i$ children and $e_i$ incoming edges,
pick a child uniformly at random and move there, until reaching a node $n_k$ with no outgoing edges.
Stop and output $\frac{d_1}{e_1} \ldots \frac{d_k}{e_k}$.
In the algorithm, we assume $e_1 = 1$ for the root and $d_k = 1$ for the leaf.
\end{verse}

\begin{proposition}
For any program $P$, \textbf{Algorithm P} is a poly-time unbiased estimator for $C(P)$.
\end{proposition}

\begin{tcolorbox}[
    colback=SkyBlue!5,
    colframe=blue!20,
    boxrule=1pt,
    arc=2pt,
    width=\textwidth,
    breakable
]
\begin{example}
For \cref{ex:r+r+r}, \textbf{Algorithm P} estimates $(3\times 2\times 1 \times 1)/(1\times 1 \times 2\times 3) = 1$ for each path through the DAG.
Thus, the expectation is also 1 (and the variance is 0).
\qed
\end{example}
\end{tcolorbox}
To show that \textbf{Algorithm P} is poly-time, given an execution graph, we have to determine the number of its \emph{incoming} edges in $\mathcal{T}(P)$ in poly-time.
For sequentially consistent executions, we can do this as follows.
Let $G=(E,\rf{},\mo{})$ be an execution graph.
An event $e\in G.E$ is said to be sequentially maximal if it is maximal with respect to the $\seqcon{}$ relation. 
%An event $e\in G.E$ is said to be sequentially maximal, written \emph{$\seqcon{}^{\max}(e)$}, if it is maximal with respect to the $\seqcon{}$ relation. 
%We define the set of all $\seqcon{}^{\max}$ events of $G$ as $\seqcon{}^{\max}(G) = \{e \mid \forall e' \in G.E : (e,e') \not\in G.\seqcon{}\}$.

\begin{proposition}
\label{prop1}
Let $P$ be a program and let $G$ a reachable execution graph in $\mathcal{T}(P)$. 
The number of predecessors of $G$ is exactly the number of sequentially maximal events of $G$. 
\end{proposition}

The proof can be found in \cref{theory:C}.
Since it is easily seen that we can compute the number of sequentially maximal events of an execution graph in polynomial-time,
Proposition~\ref{prop1} implies that \textbf{Algorithm P} is poly-time.

%Proposition~\ref{prop1} implies that the number of predecessors of each node in the DAG $\mathcal{T}(P)$ 
%can be computed by computing $\seqcon{}^{\max}$ which can be done in polynomial time. Thus, \textbf{Algorithm P} is poly-time.

\begin{comment}
\begin{example}
Unfortunately, Algorithm P can have very high variance.
Consider the program
\begin{equation*}
  \tag{\textsc{rlru+rlru}}
  \label{ex:rlru+rlru}
  \inarr{ 
\readInst{a_1}{x} \\
\lockInst{l}\\
\readInst{b_1}{y}\\
\unlockInst{l}
}{ 
\readInst{a_2}{x} \\
\lockInst{l}\\
\readInst{b_2}{y}\\
\unlockInst{l}
}
\end{equation*}
It has exactly two execution graphs, depending on whether the first or the second thread wins the race on the lock.
However, 
\RM{Mohammad fill in the rest}
\qed
\end{example}

\input{RLRU-example}
\end{comment}
\begin{tcolorbox}[
    colback=SkyBlue!5,
    colframe=blue!20,
    boxrule=1pt,
    arc=2pt,
    width=\textwidth,
    boxsep=5pt,
    left=5pt,
    right=5pt,
    top=5pt,
    bottom=5pt
]

\noindent
\begin{minipage}[b]{0.62\textwidth}
\begin{example} \label{example:R+RR}
Unfortunately, \textbf{Algorithm P} can have very high variance even for a trivial program.
Consider this simple program with three events:
 \begin{equation*}
  \tag{\textsc{r+rr}}
  \label{ex:r+rr}
  \inarr{ \textcolor{Red}{\readInst{a_1}{y}} }{
\textcolor{blue}{\readInst{a_2}{x}} \\
\textcolor{blue}{\readInst{a_3}{x}}
}   
\end{equation*}
\cref{meta-graph-r+rr} shows the DAG $\mathcal{T}(\text{\textsc{r+rr}})$ of the program. Although $\mathcal{T}(\text{\textsc{r+rr}})$ has exactly one sink, there are three paths to reach the sink, and running \textbf{Algorithm P} on each path returns a different estimation.
\end{example}
\end{minipage}%
\hfill
\begin{minipage}[b]{0.40\textwidth}
\centering
\vspace{0pt}
\centering
\begin{tikzpicture}[
    scale=0.7,
    transform shape,
    >=stealth,
    node distance=2.2cm and 2.5cm,
    hypernode/.style={draw, rectangle, rounded corners=4pt, very thick, inner sep=3pt, font=\small\bfseries},
    event/.style={draw=none, font=\normalsize},
    rf/.style={->, thick, dashed, green!60!black},
    co/.style={->, thick, dotted, red!80!black},
    graph box/.style={draw, rectangle, rounded corners=4pt, thick, minimum width=1.4cm, minimum height=1cm, inner sep=0.1cm, minimum size=0cm},
    edge label/.style={midway, very thick, fill=white, font=\footnotesize, inner sep=1pt},
    edge-t1/.style={->, very thick, Red},
    edge-t2/.style={->, very thick, blue},
]

% Top layer: G0
\node[hypernode] (G0) at (0,0) {$G_0$};

\node[hypernode] (G1) at (-1,-1) {$G_1$};

\node[hypernode] (G2) at (1,-1) {$G_2$};

\node[hypernode] (G3) at (0,-2) {$G_3$};

\node[hypernode] (G4) at (2,-2) {$G_4$};

% G5 - Concrete graph
\begin{scope}[xshift=2cm, yshift=-3.2cm, scale=0.5, local bounding box=G5box]
  \node[event] (g5i) {$\init$};
  \node[event, below=0.7cm of g5i, xshift=-1.0cm] (g5r1) {$\opr(y)$};
  \node[event, below=0.7cm of g5i, xshift=1.0cm] (g5r2) {$\opr(x)$};
  \node[event, below=0.4cm of g5r2] (g5r3) {$\opr(x)$};
  \draw[->] (g5i) -- (g5r1);
  \draw[->] (g5i) -- (g5r2);
  \draw[->] (g5r2) -- (g5r3);
\end{scope}
\node[graph box, fit=(G5box)] (G5) {};

% Abstract hyperedges
\draw[edge-t1] (G0) -- (G1);
\draw[edge-t2] (G0) -- (G2);

\draw[edge-t2] (G1) -- (G3);
\draw[edge-t1] (G2) -- (G3);

% Edges from G2
\draw[edge-t2] (G2) -- (G4);

% Edges to G5 - Concrete graph
\draw[edge-t2] (G3) -- (G5);
\draw[edge-t1] (G4) -- (G5);
\end{tikzpicture}
\captionof{figure}{The transition system $\mathcal{T}(\ref{ex:r+rr})$}
\label{meta-graph-r+rr}
\end{minipage}%

For example, the leftmost path—where $t_1$ is scheduled first—returns $0.5$. The middle path—where $t_1$ is scheduled second—returns $1$. The rightmost path—where $t_1$ is scheduled last—returns $2$. Intuitively, the path in which all events of $t_2$ are scheduled before $t_1$ has the lowest probability of being chosen and therefore yields the largest estimate. 

We extend this to the case where $t_2$ performs $n$ reads from $x$ (with $n>2$). There is again exactly one sink, but $n{+}1$ root-to-sink paths, indexed by $k\in[0,n]$ as the number of $t_2$ events before $t_1$. The path with index $k$ is taken with probability $2^{-k-1}$ if $k\neq n$ and $2^{-n}$ if $k=n$. The estimate returned by \textbf{Algorithm P} is 
\[
r_0 = 2^{1-n}, \quad r_k = 2^{2k-n} \mbox{ for }0< k <n, \mbox{ and }r_n = 2^{n-1}
\]
The variance of \textbf{Algorithm P} on this program is $\operatorname{Var}(r)=\frac{9}{28}\,2^{n}+\frac{10}{7}\,2^{-2n}-1$, which is exponential in $n$. Since the variance is large in this very simple setting, scaling to more complex programs is not promising using this approach.
\qed
\end{tcolorbox}

\section{Unbiased Estimation from Optimal Dynamic Partial Order Reduction}
\label{trust_sec}
\subsection{The Tree of \trust} \label{trust_sec:tree}
Based on Example~\ref{example:R+RR}, one reason \textbf{Algorithm P} has high variance is that it still samples interleavings, even when the set of sinks is much smaller. To circumvent the high variance, we now define a different transition system on the space of execution graphs, that we call
$\mathcal{D}(P)$.
Our construction will ensure three properties:
(a) $\mathcal{D}(P)$ is a tree;
(b) the set of leaves of $\mathcal{D}(P)$ corresponds to the set of sink nodes of $\mathcal{T}(P)$, and so the number of leaves
of $\mathcal{D}(P)$ is precisely $C(P)$; and 
(c) the children of an internal node $G$ of $\mathcal{D}(P)$ can be constructed in polynomial time, solely by looking at $G$.

The construction of \(\mathcal{D}(P)\) over the space of execution graphs relies on a \emph{stateless optimal \dpor} technique \cite{optimal-dpor,trust}.
An optimal \dpor technique is an algorithm that traverses $\mathcal{T}(P)$, with the guarantee that each sink in the graph is visited exactly once.
It is stateless if it only uses the current path to find the successors of a node.

%Our insight is that the run of a stateless optimal \dpor algorithm would yield a tree $\mathcal{D}(P)$ with the above features.
%
% By "optimal," we mean that if \(V\) represents the set of all reachable execution graphs for program \(P\), then an optimal, sound, and complete \dpor-based technique will
% explore a subset of these graphs, denoted as \(V'\), where \(V' \subseteq V\), \(C(p) \subseteq V'\), and \(\forall G, G' \in V': G \not\equiv_{\eta} G'\), with \(\eta\) 
% representing any equivalence partitioning definition, such as Mazurkiewicz traces \cite{mazurkiewicz}, Shasha-Snir \cite{shasha-snir}, or reads-from equivalence 
% \cite{dcdpor}. In essence, optimal \dpor ensures that \(\mathcal{D}(P)\) is always structured as a tree.

We describe the construction of $\mathcal{D}(P)$ using \trust\cite{trust}, a stateless optimal \dpor algorithm. In order to describe this construction, we first extend the definition of an execution graph $G$ to also include a total order $\inord$ on the events of $G$, called the \emph{insertion order} of $G$. This insertion order will represent the order in which events were incrementally added to $G$ (starting from the initial event $\init$). We will represent each such execution graph $G$ as a tuple of the form $(E,\rf,\mo,\inord)$. Note that any such new execution graph $G$ corresponds to a unique execution graph $G_{\textsf{del} \inord} \in \mathcal{T}(P)$ obtained by deleting the insertion order from $G$.
In addition to the new execution graphs, we also define a new relation, called the \emph{happens-before} relation  ($\porf{}$) as $(\po{} \cup \rf{})^+$.
From now on, we will use the extended definition of execution graphs throughout.
All of the definitions and properties that we used for execution graphs in $\mathcal{T}(P)$, for example, sequential consistency, carry over to the new execution graphs in a straightforward manner. 

We now describe the \trust algorithm in a declarative formulation, describing how it always produces a tree $\mathcal{D}(P)$ over (extended) sequentially consistent execution graphs from the program $P$. 
We describe the construction step-wise, by first giving some intuition on how the \trust algorithm visits each sequentially consistent execution graph of $P$ and then explaining each step formally.

Initially the \trust algorithm begins at the initial graph $G_\init = (\{\init\},\emptyset,\emptyset,\emptyset)$.
Assume that for each sequentially consistent graph $G$, the set $\nextop_{P}(G)$ is ordered---this corresponds to a fixed scheduler. Now, suppose the \trust algorithm
is currently at some sequentially consistent graph $G$. In order to determine which graph $G'$ to visit next, it first picks the minimal element $e$ from the set $\nextop_{P}(G)$. Now, $e$ could either be a read event or a write event. In case $e$ is a read event, \trust extends the graph $G$ by adding $e$ to it. To do this, it needs to decide which write event precedes the event $e$ according to the $\rf{}$ order.
To this end, let $e'_1,e'_2,\dots,e'_k$ be the set of all write events in $G$ that write to the same location that $e$ reads from and writes the same value that $e$ reads. Corresponding to each $e'_i$,
\trust creates a new graph $G'_i$ by adding $(e'_i,e)$ to the $\rf{}$ relation. Furthermore, in $G'_i$, $e$ is added as the maximum element in the insertion order $\inord$, reflecting the fact that it is the newest element that has been added. After $e$ has been added, the graph $G'_i$ need not necessarily be sequentially consistent. Thus, \trust also uses an oracle (consistency checker) to check if $G'_i$ is sequentially consistent. If it indeed turns out to be consistent, \trust then recursively continues the exploration from $G'_i$. (In the implementation, \trust first constructs $G'_1$, checks if it is consistent and if so, recursively explores from $G'_1$, then it constructs $G'_2$, checks if it is consistent and if so, recursively explores from $G'_2$ and so on.) These steps are called \emph{forward revisits - write to read} and, corresponding to them, in $\mathcal{D}(P)$, we have edges $G \xrightarrow{\fwdtikz, e, e'_i} G'_i$ for each $G'_i$ that is sequentially consistent.

Suppose $e$ is a write event. \trust can now add $e$ in two possible ways. The first way is to add $e$ so that it does not affect any of the existing read events in $G$.  To do this, \trust simply needs to decide where exactly $e$ would fit into the $\mo{}$ order. To this end, let $e'_1,\dots,e'_k$ be the
set of all write events in $G$ that write to the same location as $e$. Corresponding to each $e'_i$,
\trust creates a new graph $G'_i$ by placing $e$ just after $e'_i$ in the $\mo{}$ order, i.e., $e'_i$ and every predecessor of $e'_i$ is a predecessor of $e$ and every successor of $e'_i$ is a successor of $e$. 
Furthermore, in $G'_i$, $e$ is added as the maximum element in the insertion order $\inord$.
The $\rf{}$ relation is unchanged, because we are adding $e$ in a way so that it does not affect any of the existing read events in $G$. Then \trust checks for sequential consistency of $G'_i$ and if it is found to be consistent, it recursively explores from $G'_i$. (Once again, an implementation of \trust first constructs $G'_1$, checks if it is consistent and if so, recursively explores from $G'_1$, then it constructs $G'_2$, checks if it is consistent and if so, recursively explores from $G'_2$ and so on.) 
These steps are called \emph{forward revisits - write to write} and, corresponding to them, in  $\mathcal{D}(P)$, 
we have edges $G \xrightarrow{\fwdtikz, e, e'_i} G'_i$ for each $G_i'$ that is sequentially consistent. 

The other possible way for \trust to add $e$ is to do so in a way that it affects some existing read event $r$ in $G$, which reads the same value from the same location that $e$ is writing to. To this end, first observe that, if $r$ and $e$ are related by the $G.\porf{}$ relation, i.e., $(r,e) \in G.\porf{}$, then $e$ cannot affect $r$ at all. This is because the $\porf{}$ relation ensures that the only way a thread can perform $e$ is to have done $r$ first. So, \trust will only consider those $r$ that satisfy $(r,e) \notin G.\porf{}$. 

Now, it might happen that considering all read events $r$ that satisfy $(r,e) \notin G.\porf{}$ could cause 
redundancies in exploration: the same trace can be visited more than once. 
To avoid this, \trust further considers only those read events $r$ that satisfy a condition called \emph{maximally revisitable}. To define this, we set up some notation. Given a read event $r$ (that reads the same value from the same location that $e$ writes to), we let $D_r$ be the set of all events $e' \neq r$ that were inserted after $r$ (i.e., $(r,e') \in G.\inord$) such that $(e',e) \notin G.\porf{}$.

Now, intuitively, we say $G$ is \emph{maximally revisitable} with respect to $r$ if every event in 
$D_r \cup \{r\}$ was added to $G$ in a consecutive sequence. Formally, let $d_1,d_2,\dots,d_k$ be the elements of $D_r$ arranged according to $G.\inord$. Then, we say that $G$ is maximally revisitable with respect to $r$ if there exist graphs $H_1,H_2,\dots,H_k,H_{k+1}$ in $\mathcal{T}(P)$ 
such that $H_{k+1} = G_{\textsf{del} \inord}$, and each $H_{i+1}$ is obtained from $H_i$ by executing $d_i$.

%For a read event $r$, let $D_r$ be the set of all events $e' \neq r$ that were inserted after $r$ (i.e., $(r,e') \in G.\inord$) such that $(e',e) \notin G.\porf{}$. For each $e' \in D_r$, let $previous_{e',e}$ be the set $\{e'' | e'' \inord e' \lor (e'',e) \in G.\porf{}\}$. For all $e' \in D_r$ (where $e'$ is a write event), if there exists any read event $r' \in previous_{e', e}$ such that $G.\rf{}(r') = e'$, then $e'$ is not maximally added.

%For each $e' \in D_r$, let $f$ be $e'$ if $f$ is a write event, otherwise, let $f$ be the unique write event such that $(f,e') \in G.\rf{}$. We say that $e'$ was maximally added if $f \in previous_{e',e}$ and there is no $e'' \in previous_{e',e}$ such that $(f,e'') \in G.\mo{}$. We say that $r$ is maximally revisitable if every element in $D_r \cup \{r\}$ was maximally added.

%\MO{Begin of alternative def for max}
%We say that a graph $G$ is \emph{maximally revisitable} w.r.t.\ a  set of events $D$ if every event in $D$ was added to $G$ in a consecutive sequence in $\mathcal{T}(P)$. Formally, let $d_1,d_2,\dots,d_k$ be the elements of $D$ arranged according to $G.\inord$. Then, there must exist graphs $H_1$, $H_2$, $\ldots$, $H_{k+1}$ in $\mathcal{T}(P)$ such that $H_{k+1} = G_{\textsf{del} \inord}$, and each $H_{i+1}$ is obtained from $H_i$ by executing $d_i$.
%\MO{End of alternative def for max}

Now \trust only considers all those read events $r$ such that $(r,e) \notin G.\porf{}$ and $r$ is maximally revisitable. For any such event $r$, it first removes the set $D_r$ completely from $G$.
Then, it adds $e$ as a write event that precedes the read event $r$ (i.e., adds $(e,r)$ to the $\rf{}$ relation). Finally, it then decides where to include $e$ in the $\mo{}$ ordering, by going over all possible write events $w$ that write to the same location as $e$ does and adding $e$ just after $w$ in the $\mo{}$ ordering. (This part is similar to the procedure described in forward revisit - write to write). Let the new graph obtained this way be $G'_{r,w}$. \trust then checks if $G'_{r,w}$ is consistent and if so, it recursively continues the exploration from $G'_{r,w}$.  Such a step by \trust is called \emph{backward revisit - write to read}. Corresponding to this, in $\mathcal{D}(P)$, we will have the edge $G \xrightarrow{\bwdtikz, e, r} G'_{r,w}$ if $G'_{r,w}$ is consistent.

This completes the definition of $\mathcal{D}(P)$. As mentioned before, it is a declarative formulation of \trust and so inherits its ``nice'' properties. In particular, the following proposition follows from~\cite{trust}, where it is shown that $\mathcal{D}(P)$ is indeed a tree.

\begin{proposition}\label{prop:trust-size}
$\mathcal{D}(P)$ is a tree.
The outdegree of any node is bounded by $O(|P|^2)$.
The number of leaves of $\mathcal{D}(P)$ is precisely $C(P)$.
The depth of $\mathcal{D}(P)$ is $O(|P|^2)$.
\end{proposition}

We now present a new estimation algorithm that combines the construction of $\mathcal{D}(P)$ with \textbf{Algorithm K}.
%Note that finding the children of a node is a polynomial time operation that only depends on the current node.

\begin{verse}
\textbf{Algorithm T}. 
Starting at the root $G_1 := G_\init$ of $\mathcal{D}(P)$, run the following
randomized procedure.
At each node $G_i$, pick one of its $d_i$ children
uniformly at random and move there, until reaching a node $G_k$ with
no outgoing edges.
Stop and output $d_1\cdot\ldots\cdot d_{k-1}$.
\end{verse}

\begin{proposition}
\textbf{Algorithm T} is a poly-time unbiased estimator for $C(P)$.
\end{proposition}

Unbiasedness follows from the argument for \textbf{Algorithm K}.
The algorithm is poly-time because each execution graph is polynomial in the size of the program,
the depth of the tree $\mathcal{D}(P)$ is quadratically 
bounded and a path in the tree can be constructed purely by performing polynomial time operations
on the execution graphs along the path.

\begin{tcolorbox}[
    colback=SkyBlue!5,
    colframe=blue!20,
    boxrule=1pt,
    arc=2pt,
    width=\textwidth,
    breakable
]

\begin{example}
\cref{fig:r+w+w-knuth} shows $\mathcal{D}(\ref{ex:r+w+w})$. 
Notice that the tree is \emph{different} from $\mathcal{T}(P)$,
even though both are trees.
When we run \textbf{Algorithm T}, it picks the leaves 
with probabilities $1/6$, $1/6$, $1/12$, $1/12$, $1/4$, and $1/4$,
respectively.
The expectation is 6, which is equal to $C(P)$.

For this example, \textbf{Algorithm T} has higher variance than running 
\textbf{Algorithm K} on $\mathcal{T}(P)$.
However, the advantage is that \textbf{Algorithm T} works on programs for which
$\mathcal{T}(P)$ is not a tree.
For example, \textbf{Algorithm T} on both the programs \ref{ex:r+r+r} and \ref{ex:r+rr} returns 1. For both of these programs, their respective trees, $\mathcal{D}$(\ref{ex:r+r+r}) and $\mathcal{D}$(\ref{ex:r+rr}) contain exactly one
path and so the variance is 0. 
\qed
\end{example}
\centering
\begin{tikzpicture}[
    >=stealth,
    node distance=1.4cm and 1.8cm,
    event/.style={draw=none, font=\normalsize},
    rf/.style={->, thick, dashed, green!60!black},
    co/.style={->, thick, dotted, orange},
    hyper edge/.style={->, very thick, blue!70!black},
    graph box/.style={draw, rectangle, rounded corners=4pt, very thick, minimum width=1.4cm, minimum height=1cm, inner sep=0.1cm, minimum size=0cm},
graph box yellow/.style={graph box, dashed},
graph box blue/.style={graph box},
    scale=0.31,
    every node/.style={transform shape},
    edge-t1/.style={->, very thick, lime!60!black},
    edge-t2/.style={->, very thick, red!70!black},
    circ-label-t1/.style={text=black, inner sep=1pt, font=\normalsize, scale=2.5},
circ-label-t2/.style={text=black, inner sep=1pt, font=\normalsize, scale=2.5}
]

\pgfdeclarelayer{background}
\pgfsetlayers{background,main}

%% Level 0
% G1
\begin{scope}[xshift=0, yshift=0, local bounding box=G1box]
  \node[event] (g1i) {init};
  \node[event, below=0.7cm of g1i, xshift=-1.5cm] (g1r1) {$\opr(x)$};
  \draw[->] (g1i) -- node[midway, right, font=\small] {} (g1r1);
  \draw[rf] (g1i) to[bend right=20] (g1r1);
\end{scope}
\begin{pgfonlayer}{background}
\node[graph box yellow, fit=(G1box)] (G1) {};
\end{pgfonlayer}

%% Level 1
% G2
\begin{scope}[xshift=-8.5cm, yshift=-3cm, local bounding box=G2box]
  \node[event] (g2i) {init};
  \node[event, below=0.7cm of g2i, xshift=-1.5cm] (g2r1) {$\opr(x)$};
  \node[event, below=0.7cm of g2i] (g2w2) {$\opw(x)$};
  \draw[->] (g2i) -- node[midway, right, font=\small] {}(g2r1);
  \draw[->] (g2i) -- node[midway, right, font=\small] {}(g2w2);
  \draw[rf] (g2i) to[bend right=20] (g2r1);
\end{scope}
\begin{pgfonlayer}{background}
\node[graph box yellow, fit=(G2box)] (G2) {};
\end{pgfonlayer}

% G8
\begin{scope}[xshift=10cm, yshift=-3cm, local bounding box=G8box]
  \node[event] (g8i) {init};
  \node[event, below=0.7cm of g8i, xshift=-1.5cm] (g8r1) {$\opr(x)$};
  \node[event, below=0.7cm of g8i] (g8w2) {$\opw(x)$};
  \draw[->] (g8i) -- node[midway, left, font=\small] {}(g8r1);
  \draw[->] (g8i) -- node[midway, right, font=\small] {}(g8w2);
  \draw[rf] (g8w2) to[] (g8r1);
\end{scope}
\begin{pgfonlayer}{background}
\node[graph box yellow, fit=(G8box)] (G8) {};
\end{pgfonlayer}

%% Level 2
% G3
\begin{scope}[xshift=-14.3cm, yshift=-7cm, local bounding box=G3box]
  \node[event] (g3i) {init};
  \node[event, below=0.7cm of g3i, xshift=-1.5cm] (g3r1) {$\opr(x)$};
  \node[event, below=0.7cm of g3i] (g3w2) {$\opw(x)$};
  \node[event, below=0.7cm of g3i, xshift=1.5cm] (g3w3) {$\opw(x)$};
  \draw[->] (g3i) -- node[midway, left, font=\small] {}(g3r1);
  \draw[->] (g3i) -- node[midway, right, font=\small] {}(g3w2);
  \draw[->] (g3i) -- node[midway, right, xshift=0.3cm, font=\small] {}(g3w3);
  \draw[rf] (g3i) to[bend right=20] (g3r1);
  \draw[co] (g3w2) -- (g3w3);
\end{scope}
\begin{pgfonlayer}{background}
\node[graph box blue, fit=(G3box)] (G3) {};
\end{pgfonlayer}

% G4
\begin{scope}[xshift=-9.2cm, yshift=-7cm, local bounding box=G4box]
  \node[event] (g4i) {init};
  \node[event, below=0.7cm of g4i, xshift=-1.5cm] (g4r1) {$\opr(x)$};
  \node[event, below=0.7cm of g4i] (g4w2) {$\opw(x)$};
  \node[event, below=0.7cm of g4i, xshift=1.5cm] (g4w3) {$\opw(x)$};
  \draw[->] (g4i) -- node[midway, left, font=\small] {}(g4r1);
  \draw[->] (g4i) -- node[midway, right, font=\small] {}(g4w2);
  \draw[->] (g4i) -- node[midway, right, xshift=0.3cm, font=\small] {}(g4w3);
  \draw[rf] (g4i) to[bend right=20] (g4r1);
  \draw[co] (g4w3) -- (g4w2);
\end{scope}
\begin{pgfonlayer}{background}
\node[graph box blue, fit=(G4box)] (G4) {};
\end{pgfonlayer}

% G5
\begin{scope}[xshift=-4cm, yshift=-7cm, local bounding box=G5box]
  \node[event] (g5i) {init};
  \node[event, below=0.7cm of g5i, xshift=-1.5cm] (g5r1) {$\opr(x)$};
  \node[event, below=0.7cm of g5i, xshift=1.5cm] (g5w3) {$\opw(x)$};
  \draw[->] (g5i) -- node[midway, left, font=\small] {}(g5r1);
  \draw[->] (g5i) -- node[midway, right, xshift=0.3cm, font=\small] {}(g5w3);
  \draw[rf] (g5w3) to[] (g5r1);
\end{scope}
\begin{pgfonlayer}{background}
\node[graph box yellow, fit=(G5box)] (G5) {};
\end{pgfonlayer}

%%% Level 3
% G6
\begin{scope}[xshift=-7cm, yshift=-11cm, local bounding box=G6box]
  \node[event] (g6i) {init};
  \node[event, below=0.7cm of g6i, xshift=-1.5cm] (g6r1) {$\opr(x)$};
  \node[event, below=0.7cm of g6i] (g6w2) {$\opw(x)$};
  \node[event, below=0.7cm of g6i, xshift=1.5cm] (g6w3) {$\opw(x)$};
  \draw[->] (g6i) -- node[midway, left, font=\small] {}(g6r1);
  \draw[->] (g6i) -- node[midway, right, font=\small] {}(g6w2);
  \draw[->] (g6i) -- node[midway, right, xshift=0.3cm, font=\small] {}(g6w3);
  \draw[rf] (g6w3) to[bend left=20] (g6r1);
  \draw[co] (g6w2) -- (g6w3);
\end{scope}
\begin{pgfonlayer}{background}
\node[graph box blue, fit=(G6box)] (G6) {};
\end{pgfonlayer}

% G7
\begin{scope}[xshift=-1cm, yshift=-11cm, local bounding box=G7box]
  \node[event] (g7i) {init};
  \node[event, below=0.7cm of g7i, xshift=-1.5cm] (g7r1) {$\opr(x)$};
  \node[event, below=0.7cm of g7i] (g7w2) {$\opw(x)$};
  \node[event, below=0.7cm of g7i, xshift=1.5cm] (g7w3) {$\opw(x)$};
  \draw[->] (g7i) -- node[midway, left, font=\small] {}(g7r1);
  \draw[->] (g7i) -- node[midway, right, font=\small] {}(g7w2);
  \draw[->] (g7i) -- node[midway, right, xshift=0.3cm, font=\small] {}(g7w3);
  \draw[rf] (g7w3) to[bend left=20] (g7r1);
  \draw[co] (g7w3) -- (g7w2);
\end{scope}
\begin{pgfonlayer}{background}
\node[graph box blue, fit=(G7box)] (G7) {};
\end{pgfonlayer}

% G9
\begin{scope}[xshift=5cm, yshift=-7cm, local bounding box=G9box]
  \node[event] (g9i) {init};
  \node[event, below=0.7cm of g9i, xshift=-1.5cm] (g9r1) {$\opr(x)$};
  \node[event, below=0.7cm of g9i] (g9w2) {$\opw(x)$};
  \node[event, below=0.7cm of g9i, xshift=1.5cm] (g9w3) {$\opw(x)$};
  \draw[->] (g9i) -- node[midway, left, font=\small] {}(g9r1);
  \draw[->] (g9i) -- node[midway, right, font=\small] {}(g9w2);
  \draw[->] (g9i) -- node[midway, right, xshift=0.3cm, font=\small] {}(g9w3);
  \draw[rf] (g9w2) to[] (g9r1);
  \draw[co] (g9w2) -- (g9w3);
\end{scope}
\begin{pgfonlayer}{background}
\node[graph box blue, fit=(G9box)] (G9) {};
\end{pgfonlayer}

% G10
\begin{scope}[xshift=13.5cm, yshift=-7cm, local bounding box=G10box]
  \node[event] (g10i) {init};
  \node[event, below=0.7cm of g10i, xshift=-1.5cm] (g10r1) {$\opr(x)$};
  \node[event, below=0.7cm of g10i] (g10w2) {$\opw(x)$};
  \node[event, below=0.7cm of g10i, xshift=1.5cm] (g10w3) {$\opw(x)$};
  \draw[->] (g10i) -- node[midway, left, font=\small] {}(g10r1);
  \draw[->] (g10i) -- node[midway, right, font=\small] {}(g10w2);
  \draw[->] (g10i) -- node[midway, right, xshift=0.3cm, font=\small] {}(g10w3);
  \draw[rf] (g10w2) to[] (g10r1);
  \draw[co] (g10w3) -- (g10w2);
\end{scope}
\begin{pgfonlayer}{background}
\node[graph box blue, fit=(G10box)] (G10) {};
\end{pgfonlayer}

%%%% Hyper Edges

\draw[edge-t1] (G1) -- node[above, yshift=8pt, circ-label-t1] {$\fwdtikz$} (G2);

\draw[edge-t2] (G1) -- node[above, yshift=8pt, circ-label-t2] {$\bwdtikz$} (G8);

\draw[edge-t1] (G2) -- node[left, yshift=20pt, circ-label-t1] {$\fwdtikz$} (G3);

\draw[edge-t1] (G2) -- node[right, xshift=8pt, circ-label-t1] {$\fwdtikz$} (G4);

\draw[edge-t2] (G2) -- node[right, yshift=20pt, circ-label-t2] {$\bwdtikz$} (G5);

\draw[edge-t1] (G8) -- node[left, yshift=20pt, circ-label-t1] {$\fwdtikz$} (G9);

\draw[edge-t1] (G8) -- node[right, yshift=20pt, circ-label-t1] {$\fwdtikz$} (G10);

\draw[edge-t1] (G5) -- node[left, yshift=5pt, xshift=-10pt, circ-label-t1] {$\fwdtikz$} (G6);

\draw[edge-t1] (G5) -- node[right, yshift=5pt, xshift=10pt, circ-label-t1] {$\fwdtikz$} (G7);

\end{tikzpicture}
\captionof{figure}{The transition system $\mathcal{D}(\ref{ex:r+w+w})$}
\label{fig:r+w+w-knuth}
\end{tcolorbox}

\subsubsection*{Some Remarks.} Before we proceed to further analyse \textbf{Algorithm T}, we make a few general remarks. First, we note that while there are different algorithms for optimal \dpor,  some algorithms maintain additional history-dependent state as the algorithm proceeds 
(e.g., the wakeup-trees of \cite{optimal-dpor}).
We cannot always turn those algorithms into an unbiased estimator because we require that a random sample can be constructed independently of the search history.
Maintaining a search history would require exponential time.
Also, if the additional state maintained by an optimal \dpor algorithm can grow exponentially large, sampling using such algorithms can lead to an exponential computational cost, in contrast to sampling using \trust. 

Finally, since \trust is an \emph{optimal} \dpor algorithm, the result returned by \textbf{Algorithm T} on a program $P$ can be used to track progress and forecast the remaining time that \emph{any} DPOR-based model checker must take on the program $P$. This is because, any DPOR-based model checker must explore at least as many complete execution graphs as \trust does. 

With regards to comparing \textbf{Algorithm K} and \textbf{Algorithm P} with \textbf{Algorithm T}, since \textbf{Algorithm T} uses \trust to construct the underlying tree on which it samples a path, it inherits many of its nice properties. For instance, consider any program $P$, where each thread performs only read operations. 
For such programs, it is clear that there is exactly one complete execution graph. 
From the description of the \trust algorithm above, it is clear that \trust, on the program $P$, only performs \emph{forward revisits - write to read} steps. Furthermore, it also follows by construction that $\mathcal{D}(P)$ in such a case will be simply a single path from the root. Hence, $\textbf{Algorithm T}$ will correctly report $C(P)$ to be 1, i.e., $\textbf{Algorithm T}$ will have zero variance in such cases.
On the other hand, as we have seen by the examples, \ref{ex:r+r+r} and \ref{ex:r+rr}, both \textbf{Algorithm K} and \textbf{Algorithm P} can have exponential variance even for such simple cases.
%and since \trust represents trace equivalence via partial orders, \textbf{Algorithm T} in general tends to incur low variance on programs with very few complete behaviors. This is in contrast to \textbf{Algorithm P}, as can be seen by the simple example $\mathcal{D}$(\ref{ex:r+rr}),
%where \textbf{Algorithm P} has exponential variance, whereas \textbf{Algorithm T} has zero variance.

Finally, for the sake of readability, we have described \textbf{Algorithm T} for sequentially consistent executions. However, it can also be made to work for any memory model that satisfies a set of requirements as specified by the \trust algorithm, including TSO, PSO, and RC11~\cite{trust}. Indeed, the properties of the \trust algorithm ensure that we can get a poly-time unbiased estimators for all these memory models as well.

%While we have described an estimator for sequentially consistent executions, the \trust algorithm
%works as long as the memory model satisfies a set of requirements and in every such case, we get a poly-time unbiased
%estimator.
%We focus on sequential consistency for readability.
%Also, while there are different algorithms for optimal \dpor, some algorithms maintain additional history-dependent state as the algorithm proceeds 
%(e.g., the wakeup-trees of \cite{optimal-dpor}).
%We cannot turn those algorithms into an estimator because we require that a random sample can be constructed independently of the search history.

\subsection{The Problem of Variance}

As mentioned before, \textbf{Algorithm T} runs in time polynomial in the size of the program and its output is a random variable whose
expected value is $C(P)$.
However, the variance of the output can be (exponentially) large and we may require exponentially many samples
to estimate $C(P)$.
%Indeed, a lower bound by Stockmeyer \cite{Stockmeyer85} shows that an estimate for trees with $N$ leaves
%may require $\Omega(\sqrt{N})$ samples.
One might expect that this might be the case because $C(P)$ can be exponential in the size of $P$. However, the variance can be high even when the tree has few (polynomially many) leaves: if the tree is ``skewed,''
the probability of traversing a long path may be exponentially small and this can push the variance up. We now demonstrate an example for this phenomenon.

%The variance of the estimator is given by \cite{knuth75}:
%\[
%\variance(T_v) = d\sum_{1\leq j\leq d} \variance(T_{v_j}) + \sum_{1\leq i < j \leq d} (cost(T_{v_i}) - cost(T_{v_j}))^2
%\]
%where $v_1, \ldots, v_d$ are the children of node $v$.
%Thus, if the tree is ``badly balanced'' between children, the variance can be high.

\begin{tcolorbox}[
    colback=SkyBlue!5,
    colframe=blue!20,
    boxrule=1pt,
    arc=2pt,
    width=\textwidth,
    boxsep=5pt,
    left=5pt,
    right=5pt,
    top=5pt,
    bottom=5pt
]

\noindent
\begin{minipage}[b]{0.4\textwidth}
\begin{example}[Hairbrush]
\label{ex:hairbrush}
Consider the following program
\begin{equation*}
  \tag{\textsc{r+nw}}
  \label{ex:r+nw}
  \inarr{ \readInst{a}{x} }{
\writeInst{x}{1} \\
\writeInst{x}{2} \\
\ldots \\
\writeInst{x}{n}
}
\end{equation*}
The tree $\mathcal{D}(\mathsf{R+NW})$ is ``skewed to the right.''
\textbf{Algorithm T} reaches the rightmost execution graph with probability $1/2^n$.
The expectation is $n+1$ but the variance is 
\[
\sum_{i=1}^{n}\frac{1}{2^i}2^{2i} + \frac{1}{2^n}2^{2n} = 2^{n+1} + 2^n - 2
\]
\qed
\end{example}
\end{minipage}%
\hfill
\begin{minipage}[b]{0.6\textwidth}
\centering
\vspace{0pt}
\begin{tikzpicture}[
    scale=0.7,
    transform shape,
    >=stealth,
    node distance=2.2cm and 2.5cm,
    hypernode/.style={draw, rectangle, rounded corners=4pt, very thick, inner sep=3pt, font=\small\bfseries},
    event/.style={draw=none, font=\normalsize},
    rf/.style={->, thick, dashed, green!60!black},
    co/.style={->, thick, dotted, red!80!black},
    graph box/.style={draw, rectangle, rounded corners=4pt, thick, minimum width=1.4cm, minimum height=1cm, inner sep=0.1cm, minimum size=0cm},
    edge label/.style={midway, very thick, fill=white, font=\footnotesize, inner sep=1pt},
    graph box yellow/.style={graph box, dashed},
    graph box blue/.style={graph box},
    every node/.style={transform shape},
    edge-t1/.style={->, very thick, lime!60!black},
    edge-t2/.style={->, very thick, red!70!black},
    circ-label-t1/.style={text=black, inner sep=1pt, font=\normalsize, scale=1.2},
    circ-label-t2/.style={text=black, inner sep=1pt, font=\normalsize, scale=1.2}
]

\node[hypernode] (G0) at (0,0) {$G_0$};

\node[hypernode] (G1) at (1.5,-1.5) {$G_1$};

\node[hypernode] (G2) at (3,-3) {$G_2$};

%G3
\begin{scope}[xshift=-1.9cm, yshift=-1.5cm, scale=0.5, local bounding box=G3box]
  \node[event] (g3i) {$\init$};
  \node[event, below=0.7cm of g3i, xshift=-1.6cm] (g3r1) {$\opr(x)$};
  \node[event, below=0.7cm of g3i] (g3w2) {$\opw_1(x)$};
  %\node[event, below=1.4cm of g3i] (g3w3) {$\opw(x)$};
  
  % Dotted continuation
  \node[event, below=1.1cm of g3i, font=\Large] (g3dots) {$\vdots$};
  
  %\node[event, below=0.3cm of g3dots] (g3wn) {$\opw(x)$};
  \draw[->] (g3i) -- node[midway, left] {} (g3r1);
  \draw[->] (g3i) -- node[midway, right, font=\small] {}(g3w2);
  %  \draw[->] (g3w2) -- (g3w3);
  %\draw[->] (g3w2) -- (g3dots);
  %\draw[->] (g3dots) -- (g3wn);
  \draw[rf] (g3i) to[bend right=20] (g3r1);
\end{scope}
\node[graph box, fit=(G3box)] (G3) {};

%G4
\begin{scope}[xshift=-0.4cm, yshift=-3cm, scale=0.5, local bounding box=G4box]
  \node[event] (g4i) {$\init$};
  \node[event, below=0.7cm of g4i, xshift=-1.6cm] (g4r1) {$\opr(x)$};
  \node[event, below=0.7cm of g4i] (g4w2) {$\opw_1(x)$};
  \node[event, below=1.8cm of g4i] (g4w3) {$\opw_2(x)$};
  
  % Dotted continuation
  \node[event, below=2.9cm of g4i, font=\Large] (g4dots) {$\vdots$};
  
  %\node[event, below=0.3cm of g3dots] (g3wn) {$\opw(x)$};
  \draw[->] (g4i) -- node[midway, left] {} (g4r1);
  \draw[->] (g4i) -- node[midway, right, font=\small] {}(g4w2);
  \draw[->] (g4w2) -- (g4w3);
  \draw[->] (g4w3) -- (g4dots);
  %\draw[->] (g3dots) -- (g3wn);
  \draw[rf] (g4w2) to (g4r1);
\end{scope}
\node[graph box, fit=(G4box)] (G4) {};

%G4
% \begin{scope}[xshift=-0.3cm, yshift=-3cm, scale=0.5, local bounding box=G4box]
%   \node[event] (g4i) {$\init$};
  
%   \node[event, below=0.7cm of g4i, xshift=-1.6cm] (g4r1) {$\opr(x)$};
  
%   \node[event, below=-0.2cm of g4i, font=\Large] (g4dots2) {$\vdots$};
  
%   \node[event, below=0.7cm of g4dots2] (g4w2) {$\opw_2(x)$};
%   %\node[event, below=1.4cm of g4i] (g4w3) {$\opw(x)$};
  
%   % Dotted continuation
%   \node[event, below=1.6cm of g4i, font=\Large] (g4dots) {$\vdots$};
  
%   %\node[event, below=0.3cm of g4dots] (g4wn) {$\opw(x)$};
%   \draw[->] (g4i) -- node[midway, left] {} (g4r1);
%   \draw[->] (g4dots2) -- node[midway, right, font=\small] {}(g4w2);
%   %  \draw[->] (g4w2) -- (g4w3);
%   %\draw[->] (g4w2) -- (g4dots);
%   %\draw[->] (g4dots) -- (g4wn);
%   \draw[rf] (g4w2) to[] (g4r1);
% \end{scope}
% \node[graph box, fit=(G4box)] (G4) {};

%Gm
\begin{scope}[xshift=3cm, yshift=-6cm, scale=0.5, local bounding box=Gmbox]
  \node[event] (gmi) {$\init$};
  
  \node[event, below=0.7cm of gmi, xshift=-1.6cm] (gmr1) {$\opr(x)$};
  hair-brush-rnw
  \node[event, below=-0.2cm of gmi, font=\Large] (gmdots2) {$\vdots$};
  
  \node[event, below=0.7cm of gmdots2] (gmw2) {$\opw_{n-1}(x)$};
  \node[event, below=1.8cm of gmdots2] (gmw3) {$\opw_{n}(x)$};
  %\node[event, below=1.4cm of gmi] (gmw3) {$\opw(x)$};
  
  % Dotted continuation
  %\node[event, below=1.6cm of gmi, font=\Large] (gmdots) {$\vdots$};
  
  %\node[event, below=0.3cm of gmdots] (gmwn) {$\opw(x)$};
  \draw[->] (gmi) -- node[midway, left] {} (gmr1);
  \draw[->] (gmdots2) -- node[midway, right, font=\small] {}(gmw2);
  \draw[->] (gmw2) -- (gmw3);
  %\draw[->] (gmw2) -- (gmdots);
  %\draw[->] (gmdots) -- (gmwn);
  \draw[rf] (gmw2) to[] (gmr1);
\end{scope}
\node[graph box, fit=(Gmbox)] (Gm) {};

%Gn
\begin{scope}[xshift=6.8cm, yshift=-6cm, scale=0.5, local bounding box=Gnbox]
  \node[event] (gni) {$\init$};
  
  \node[event, below=0.7cm of gni, xshift=-1.6cm] (gnr1) {$\opr(x)$};
  
  \node[event, below=-0.2cm of gni, font=\Large] (gndots2) {$\vdots$};
  
  \node[event, below=0.7cm of gndots2] (gnw2) {$\opw_{n-1}(x)$};
  \node[event, below=1.8cm of gndots2] (gnw3) {$\opw_{n}(x)$};
  %\node[event, below=1.4cm of gni] (gnw3) {$\opw(x)$};
  
  % Dotted continuation
  %\node[event, below=1.6cm of gni, font=\Large] (gndots) {$\vdots$};
  
  %\node[event, below=0.3cm of gndots] (gnwn) {$\opw(x)$};
  \draw[->] (gni) -- node[midway, left] {} (gnr1);
  \draw[->] (gndots2) -- node[midway, right, font=\small] {}(gnw2);
  \draw[->] (gnw2) -- (gnw3);
  %\draw[->] (gnw2) -- (gndots);
  %\draw[->] (gndots) -- (gnwn);
  \draw[rf] (gnw3) to[bend left=20] (gnr1);
\end{scope}
\node[graph box, fit=(Gnbox)] (Gn) {};

% Extension dots from G2 (left branch)
\node[font=\Huge] (g2dots2) at (1.5, -4.7) {$\vdots$};

% Extension dots from G2 (right branch)
\node[font=\Huge] (g2dots3) at (4.5,-4.5) {$\vdots$};

%% Hyper Edges

\draw[edge-t2] (G0) -- node[above, yshift=3pt, xshift=7pt, circ-label-t2] {$\bwdtikz$} (G1);

\draw[edge-t1] (G0) -- node[above, yshift=8pt, circ-label-t1] {$\fwdtikz$} (G3);

\draw[edge-t2] (G1) -- node[above, yshift=3pt, xshift=7pt, circ-label-t2] {$\bwdtikz$} (G2);

\draw[edge-t1] (G1) -- node[above, yshift=8pt, circ-label-t1] {$\fwdtikz$} (G4);

\draw[edge-t2] (G2) -- node[above, yshift=3pt, xshift=7pt, circ-label-t2] {$\bwdtikz$} (g2dots3);

\draw[edge-t1] (G2) -- node[above, yshift=8pt, circ-label-t1] {$\fwdtikz$} (g2dots2);

\draw[edge-t2] (g2dots3) -- node[above, yshift=3pt, xshift=7pt, circ-label-t2] {$\bwdtikz$} (Gn);

\draw[edge-t1] (g2dots3) -- node[above, yshift=8pt, circ-label-t1] {$\fwdtikz$} (Gm);

\end{tikzpicture}
\captionof{figure}{The transition system $\mathcal{D}(\ref{ex:r+nw})$}
\label{hair-brush-rnw}
\end{minipage}%
\end{tcolorbox}

A simple heuristic to get better estimates is to always schedule writes first.
%With this strategy, the hairbrush tree of the above example has exactly $n+1$ leaves.
This strategy does reduce the variance for this example.
However, consider a small variation
in which the second thread starts with a read of a different location: we are back to an exponential variance.

\subsection{Stochastic Enumeration}
\label{sec:se}
In order to reduce the variance of \textbf{Algorithm T}, we use \emph{stochastic enumeration} \cite{rubinstein,vaisman2017}.
Stochastic enumeration runs multiple copies of \textbf{Algorithm T} in parallel, but couples the evolution of the parallel branches.

\begin{verse}
    \textbf{Algorithm S}. Given the tree $\mathcal{D}(P)$ and a budget $B\geq 1$, start at the root $H_1 = \set{G_\init}$ and run the following
    randomized procedure. 
    Suppose the current set of nodes is $H_i$, and let $S(H_i) = \cup_{G\in H_i} \set{G' \mid G' \mbox{ is a child of }G \mbox{ in } \mathcal{D}(P)}$.

    \noindent If $S(H_i) = \emptyset$, let $c(H_1), c(H_2), \dots, c(H_{i-1})$ be the number of leaves in the sets $H_1,H_2,\dots,H_{i-1}$. Then stop and return 
    \[
    \frac{c(H_1)}{|H_1|} + \frac{|S(H_1)|}{|H_1|}\frac{c(H_2)}{|H_2|} + \ldots + \left( \prod_{1\leq j\leq i-1} \frac{|S(H_j)|}{|H_j|}\right)\frac{c(H_{i})}{|H_{i}|}
    \]
    If $0 < |S(H_i)| \leq B$, set $H_{i+1} = S(H_i)$, and otherwise, let $H_{i+1}$ be a subset of $S(H_i)$ of size $B$ picked uniformly at random among
    all subsets of size $B$.
    Continue with $H_{i+1}$.
\end{verse}

\begin{proposition}
    For a program $P$ and a budget $B$ in unary, \textbf{Algorithm S} is a poly-time unbiased estimator for $C(P)$.
\end{proposition}

The fact that it is an unbiased estimator follows from~\cite{rubinstein,vaisman2017}.
By exactly the same argument employed for \textbf{Algorithm T}, we can show that \textbf{Algorithm S} is poly-time in the size of $P$ and $B$.

We make a few simple remarks.
First, \textbf{Algorithm S} with a budget $B = 1$ is exactly \textbf{Algorithm T}.
Second, as the following example demonstrates, when $B > 1$, the algorithm is not the same as $B$ independent executions of \textbf{Algorithm T}.

\begin{tcolorbox}[
    colback=SkyBlue!5,
    colframe=blue!20,
    boxrule=1pt,
    arc=2pt,
    width=\textwidth,
    breakable
]
\begin{example}
For the hairbrush tree from \cref{ex:hairbrush}, we notice that with a budget $B = 2$,
we have $H_0 = \set{G_\init}$ and for each $1\leq i \leq n$, we have that $H_i$ consists of two nodes: the leaf node and the internal node at depth $i$.
Hence \textbf{Algorithm S} has exactly one run on this tree.
The value of the run, which is also the expected value, is $n+1$ and the variance is $0$!
More generally, if the maximal number of leaves at any level of $\mathcal{D}(P)$ is bounded by $B$, then \textbf{Algorithm S} with budget $B$
has exactly one run and variance 0.
\qed
\end{example}

\end{tcolorbox}

Finally, while stochastic enumeration is a general technique for reducing the variance of unbiased estimators, we cannot apply it to \textbf{Algorithm K} on arbitrary programs $P$, since as we have seen, \textbf{Algorithm K} need not be an unbiased estimator for programs $P$ where $\mathcal{T}(P)$ is not a tree.
Furthermore, since the correctness of stochastic enumeration (as discussed in \cite{rubinstein,vaisman2017}) depends on the underlying search space being a tree, it cannot be applied to \textbf{Algorithm P}.

\subsection{Optimization: Tree Compression}

We implement \textbf{Algorithm S} by a simple modification of the iterative \trust algorithm, where we keep 
a frontier set of up to $B$ nodes.
In each step of the iteration, we compute the successor nodes of the frontier set as in \trust, update the estimates,
and then pick a random subset of size $B$ of the successor nodes.
(If the successor set has size at most $B$, we retain the entire set.)
This incurs a cost in cloning execution graphs and re-execution over the \trust algorithm, since we cannot simply use depth first traversal.

We perform a simple tree compression optimization.
During the execution of \trust, it is often the case that a node has only a single successor.
In such cases, rather than returning this successor immediately,  
the algorithm continues until it encounters either multiple successors or a leaf node.
This technique reduces the effective size of the exploration tree, and avoids redundant cloning and re-execution.
Importantly, tree compression does not eliminate any leaf from the original tree, preserving the correctness of the estimation.
\begin{tcolorbox}[
    colback=SkyBlue!5,
    colframe=blue!20,
    boxrule=1pt,
    arc=2pt,
    width=\textwidth,
    breakable
]
  \begin{minipage}[t]{0.48\textwidth}
    \centering
    \begin{tikzpicture}[
    scale=0.7,
    transform shape,
    >=stealth,
    node distance=2.2cm and 2.5cm,
    hypernode/.style={draw, rectangle, rounded corners=4pt, very thick, inner sep=3pt, font=\small\bfseries},
    event/.style={draw=none, font=\normalsize},
    rf/.style={->, thick, dashed, green!60!black},
    co/.style={->, thick, dotted, red!80!black},
    graph box/.style={draw, rectangle, rounded corners=4pt, thick, minimum width=1.4cm, minimum height=1cm, inner sep=0.1cm, minimum size=0cm},
    edge label/.style={midway, very thick, fill=white, font=\footnotesize, inner sep=1pt},
    graph box yellow/.style={graph box, dashed},
    graph box blue/.style={graph box},
    every node/.style={transform shape},
    edge-t1/.style={->, very thick, lime!60!black},
    edge-t2/.style={->, very thick, red!70!black},
    circ-label-t1/.style={text=black, inner sep=1pt, font=\normalsize, scale=1.2},
    circ-label-t2/.style={text=black, inner sep=1pt, font=\normalsize, scale=1.2}
]

\node[hypernode] (G0) at (0,0) {$G_0$};

\node[hypernode] (G1) at (-2,-1) {$G_1$};

\node[hypernode] (G2) at (2,-1) {$G_2$};

\node[hypernode] (G3) at (-2,-2) {$G_3$};

\node[hypernode] (G4) at (1,-2) {$G_4$};

\node[hypernode] (G5) at (2,-2) {$G_5$};

\node[hypernode] (G6) at (3,-2) {$G_6$};

\node[hypernode] (G7) at (-2,-3) {$G_7$};

\node[hypernode] (G8) at (1,-3) {$G_8$};

\node[hypernode] (G9) at (2,-3) {$G_9$};

\node[hypernode] (G10) at (-3,-4) {$G_{10}$};

\node[hypernode] (G11) at (-1,-4) {$G_{11}$};

\node[hypernode] (G12) at (1,-4) {$G_{12}$};

\node[hypernode] (G13) at (2,-4) {$G_{13}$};

\node[hypernode] (G14) at (-3,-5) {$G_{14}$};

\node[hypernode] (G15) at (0,-5) {$G_{15}$};

\node[hypernode] (G16) at (2,-5) {$G_{16}$};

\node[hypernode] (G17) at (-3,-6) {$G_{17}$};

\node[hypernode] (G18) at (0,-6) {$G_{18}$};

\node[hypernode] (G19) at (0,-7) {$G_{19}$};

%% Hyper Edges

\draw[edge-t1] (G0) -- (G1);

\draw[edge-t1] (G0) -- (G2);

\draw[edge-t2] (G1) -- (G3);

\draw[edge-t1] (G2) -- (G4);

\draw[edge-t1] (G2) -- (G5);

\draw[edge-t1] (G2) -- (G6);

\draw[edge-t1] (G3) -- (G7);

\draw[edge-t2] (G4) -- (G8);

\draw[edge-t2] (G5) -- (G9);

\draw[edge-t1] (G7) -- (G10);

\draw[edge-t1] (G7) -- (G11);

\draw[edge-t1] (G8) -- (G12);

\draw[edge-t1] (G9) -- (G13);

\draw[edge-t2] (G10) -- (G14);

\draw[edge-t1] (G12) -- (G15);

\draw[edge-t1] (G12) -- (G16);

\draw[edge-t1] (G14) -- (G17);

\draw[edge-t2] (G15) -- (G18);

\draw[edge-t1] (G18) -- (G19);

\end{tikzpicture}
\captionof{figure}{The transition system $\mathcal{D}(\ref{ex:lll})$}
\label{lll-full}
  \end{minipage}
  \hfill
  \begin{minipage}[t]{0.48\textwidth}
    \centering
    \begin{tikzpicture}[
    scale=0.7,
    transform shape,
    >=stealth,
    node distance=2.2cm and 2.5cm,
    hypernode/.style={draw, rectangle, rounded corners=4pt, very thick, inner sep=3pt, font=\small\bfseries},
    event/.style={draw=none, font=\normalsize},
    rf/.style={->, thick, dashed, green!60!black},
    co/.style={->, thick, dotted, red!80!black},
    graph box/.style={draw, rectangle, rounded corners=4pt, thick, minimum width=1.4cm, minimum height=1cm, inner sep=0.1cm, minimum size=0cm},
    edge label/.style={midway, very thick, fill=white, font=\footnotesize, inner sep=1pt},
    graph box yellow/.style={graph box, dashed},
    graph box blue/.style={graph box},
    every node/.style={transform shape},
    edge-t1/.style={->, very thick, lime!60!black},
    edge-t2/.style={->, very thick, red!70!black},
    circ-label-t1/.style={text=black, inner sep=1pt, font=\normalsize, scale=1.2},
    circ-label-t2/.style={text=black, inner sep=1pt, font=\normalsize, scale=1.2}
]

\node[hypernode] (G0) at (0,0) {$G_0$};

\node[hypernode] (G1) at (-2,-1) {$G_1$};

\node[hypernode] (G2) at (2,-1) {$G_2$};

\node[hypernode] (G3) at (-3,-2) {$G_3$};

\node[hypernode] (G4) at (-1,-2) {$G_4$};

\node[hypernode] (G5) at (1,-2) {$G_5$};

\node[hypernode] (G6) at (2,-2) {$G_6$};

\node[hypernode] (G7) at (3,-2) {$G_7$};

\node[hypernode] (G8) at (0,-3) {$G_8$};

\node[hypernode] (G9) at (2,-3) {$G_9$};

%% Hyper Edges

\draw[edge-t1] (G0) -- (G1);

\draw[edge-t1] (G0) -- (G2);

\draw[edge-t2] (G1) -- (G3);

\draw[edge-t1] (G1) -- (G4);

\draw[edge-t1] (G2) -- (G5);

\draw[edge-t1] (G2) -- (G6);

\draw[edge-t1] (G2) -- (G7);

\draw[edge-t2] (G5) -- (G8);

\draw[edge-t1] (G5) -- (G9);

\end{tikzpicture}
\captionof{figure}{The compressed transition system $\mathcal{D}(\ref{ex:lll})$}
\label{lll-comp}
  \end{minipage}
\begin{example}
Consider the following program.

\begin{equation*}
  \tag{\textsc{L+L+L}}
  \label{ex:lll}
  \inarr{ 
\lockInst{l} \\
\unlockInst{l} \\
}{ 
\lockInst{l} \\
\unlockInst{l} \\
}{
\lockInst{l} \\
\unlockInst{l} \\
}
\end{equation*}

This program consists of three threads, each of which attempts to acquire a lock and subsequently release it. It has $3!$ execution graphs. The complete $\mathcal{D}(\ref{ex:lll})$ tree of this program without tree compression is depicted in \cref{lll-full} and
with compression in \cref{lll-comp}.
With $B=5$, the uncompressed tree requires 20 clones and re-execution but the compressed tree only 10.
%When running \testor with $\mathcal{B} = 4$ over this tree, approximately 20 graph clonings and re-executions must take place per iteration. However, when running \testor with $\mathcal{B} = 5$ over the compressed $\mathcal{D}(\ref{ex:lll})$ tree, as depicted in \cref{lll-comp}, only approximately 10 graph clonings and re-executions are required, reflecting a twofold reduction in overhead.
\qed
\end{example}
\end{tcolorbox}

\subsection{Related Work: The GenMC Estimator}

The state estimation problem was also considered in \cite{gator}, where the authors give a randomized procedure based on the \trust algorithm.
However, as we show below, that estimator is not an unbiased estimator; thus, even with exponentially many samples,
it has no guarantees to converge to the correct answer.
Experimentally, the estimator converges to incorrect values quite frequently.

%\citet{gator} describe a biased probabilistic estimator for \trust. 
%
% However, the proposed estimator was biased.
%% likely because the authors did not recognize that the backtracking tree implicitly generated by the \trust algorithm already provides a natural structure for unbiased estimation.

%The method follows Knuth-style single-path sampling. 
The procedure is also based on the \trust algorithm, and proceeds as follows.
First, it preferentially schedules writes over reads.
For each read or write event, it explores one of the forward-revisits uniformly at random. 
Instead of computing the estimate on the fly, it records the number of values each read can read from. 
In case a new write is added and \trust detects that a backward revisit would be required, the estimator
does not perform the backward revisit but increments the map entry for any read that could read from this new value by one. 
At the end, it takes the product of the ``possible reads-from'' counts for every read and the ``possible modification order" 
for every write to produce the final estimate.

Intuitively, the estimator assumes one can generate all graphs in $\mathcal{T}(P)$ 
without any backward revisits, simply by scheduling write events before read events.
%% (when only read events are enabled, the scheduler proceed with threads uniformly at random). 
When this assumption holds, the estimate is unbiased. 
However, the assumption may not hold. 
Since the algorithm forbids backward revisits, it may omit graphs 
reachable only via backward revisits and produce a biased estimate.
%% The estimator can be very biased. The experiment from \citet{gator}  also showed that the algorithm is biased and can converge to wrong values quite frequently.
\begin{tcolorbox}[
    colback=SkyBlue!5,
    colframe=blue!20,
    boxrule=1pt,
    arc=2pt,
    width=\textwidth,
    boxsep=5pt,
    left=5pt,
    right=5pt,
    top=5pt,
    bottom=5pt
]

\noindent
\begin{minipage}[t]{0.7\textwidth}
\begin{example}
Consider the following program. It has 4 execution graphs: 
one where the line $a_2:=x$ reads from the $\init$, one where it reads from $x:=1$, and two where it reads from $x:=2$. 
The algorithm in \cite{gator} runs as follows. 
It first schedules $t_1$ and writes 1 to $x$. 
Then since both threads are enabled with a read, the scheduler picks one uniformly at random.
A case analysis of the two possibilities shows that the expected value of the estimator is
3.5.
\end{example}
\end{minipage}%
\hfill
\begin{minipage}[t]{0.3\textwidth}
\vspace{0pt}
\begin{equation*}
  \tag{\textsc{wrww+rr}}
  \label{ex:wrww+rr}
  \inarr{ 
\writeInst{x}{1} \\
\readInst{a_1}{x} \\
\writeInst{x}{2} \\
\writeInst{y}{1} 
}{ 
\readInst{a_2}{x} \\
\text{if } a_2 = 2 \\
\readInst{a_2}{y }}
\end{equation*}
\end{minipage}%

\begin{comment}
\vspace{0.3cm}

Case A: the scheduler proceeds with $t_1$. Since the remaining events of $t_1$ are writes, the scheduler will continue proceeding with $t_1$ until it finishes. Up to this point, the estimation is $1$ because no forward revisit is needed. It then schedules the first event of $t_2$. There are three possible read-from values (including the initial value), so the estimator will store the knowledge that the read $a_2:=x$ has three possible forward revisits. With probability $2/3$, the estimator explores the case that $a_2$ does not read $2$, then $t_2$ terminates and the estimator returns 3. With probability $1/3$, $a_2$ reads  $2$ and enters the if-branch and reads $y$, which has two possible values, so the estimator returns $6$.

Case B: the scheduler proceeds with $t_2$. At that moment, the read has two possible forward revisits (from $\init$ and from the write $x:=1$), so the estimator records $2$ for this read. Regardless of which value $a_2$ reads in this case, it does not enter the if-branch and terminates. $t_1$ is then scheduled, and the estimator detects a backward revisit and updates that the earlier read into $a_2$ actually had three possible values to read from; however, because the algorithm forbids backward revisits, it does not explore that alternative and still returns $3$.

Thus the expected output in Case A is $(2/3)\times3 + (1/3)\times6 = 4$, and in Case B it is $3$. 
Thus, the expected value of the estimator is $0.5\times 4 + 0.5\times 3 = 3.5$.
\end{comment}
\qed
\end{tcolorbox}

\section{A Sub-Exponential Time Approximate Counter}

In the previous subsections, we had presented various (randomized) polynomial-time algorithms in order to estimate the number of minimal execution graphs of the given concurrent program. Due to the $\sharpP$-hardness result of approximation (Theorem~\ref{thm:approximate-counting-problem-hardness}), all
of these algorithms, in the worst case, are going to have exponentially high variance.
In light of this, we can ask the following question: 
What if we relax the requirement that our estimation algorithm must work in polynomial time? 
%Of course, if we allow the algorithm to work in exponential time, then we can get zero variance, 
%by traversing the entire \trust{} tree completely deterministically. 
%Hence, in order for this question to make sense, we ask whether there is an
%$(r, 1/4)$-approximate counter for constant $r$ that runs in sub-exponential time.
%Surprisingly, we prove that such an estimation algorithm exists. 

More precisely, note that as mentioned in Section~\ref{trust_sec}, the depth and the branching factor of the \trust tree $\mathcal{D}(P)$ is quadratic in $|P|$. Hence, the total size of $\mathcal{D}(P)$ could be $M := O((|P|^2)^{|P|^2+1} - 1)$ in the worst case. In such cases, we would like to know whether we can avoid exhaustive exploration of $\mathcal{D}(P)$ while still getting a good estimate.
Surprisingly, we prove that such an estimation algorithm exists.

%\begin{theorem}
%\label{th:stockmeyer}
 %   For any constants $r$ and $\rho$, there is a randomized $\tilde{O}(\sqrt{C(P)})$-time $(r, \rho)$-approximate counter.
%\end{theorem}

\begin{theorem}
\label{th:stockmeyer}
    For any constants $r$ and $\rho$, there is a randomized $\tilde{O}(\sqrt{M})$-time $(r, \rho)$-approximate counter.
\end{theorem}

The proof can be found in \cref{sec:stockmeyer}.
Here $\tilde{O}$ suppresses terms that are a polynomial in $|P|$. Note that this theorem is quite surprising: 
%First, suppose we consider the case where $C(P)$ is a polynomial in $|P|$. Then the approximate counter promised by this theorem runs in polynomial-time and gives an approximation of $C(P)$ that is at most a constant factor away from the actual value of $C(P)$. This is in stark contrast to the other algorithms that we have seen in this paper, which can all have exponential variance, even when $C(P)$ is a linear function of $|P|$ (as evidenced by the examples \ref{ex:r+r+r}, \ref{ex:r+rr} and \ref{ex:hairbrush}, of which the latter can be easily generalized to produce instances of exponential variance for \textbf{Algorithm S} as well). 
It shows that even when $C(P)$ is quite close to $M$, we can avoid exhaustive exploration and still get a good estimate of $C(P)$. 
%For instance, as we shall see below, \RM{where do we see this?} our algorithm indicates that if $M$ is around 1 million, then with about 1000 samples, we can get a good estimate with high probability!
However, the catch behind this algorithm is that the search has a high polynomial cost, i.e., the polynomial term suppressed by the $\tilde{O}$ notation is roughly $|P|^6$ and so we did not implement it.

We now describe the main ideas behind this algorithm, which go back to Stockmeyer \cite{Stockmeyer85}.
To describe our algorithm, we assume that we are given a concurrent program $P$, 
which defines the tree $\mathcal{D}(P)$. 
As mentioned before, if $d$ is the size of the program $P$, 
this tree has height $d^2$ and the degree of each node is at most $d^2$. 
However, every node need not have degree $d^2$. In order to make this tree uniform, we add dummy nodes to each node of $\mathcal{D}(P)$ so that the degree of each node becomes $d^2$. 
Call the resulting tree $\mathcal{H}(P)$. Note that $\mathcal{H}(P)$ 
is the complete $d^2$-ary tree of height $d^2$. 
(We will never explicitly construct $\mathcal{D}(P)$ nor $\mathcal{H}(P)$ during our algorithm; 
these are just defined implicitly). 

The algorithm distinguishes two situations: either $C(P)$ is ``small'' (less than a parameter $\theta$ that we will select later)
or ``big'' (bigger than the parameter $\theta$).
If $C(P)$ is big, then we independently sample a number of random walks proportional to $M/\theta$ in $\mathcal{H}(P)$.
We then count the number of times we hit a leaf node of $\mathcal{D}(P)$ in these random walks and then return this fraction of samples as our estimate.
Using Chebyshev bounds, we can bound the error in the procedure.
On the other hand, if $C(P)$ is small, we can enumerate all elements by a binary search.
The binary search uses an inorder numbering of all the nodes of the tree.
The numbering is used to define intervals of leaves, and the search checks if there is a leaf of $\mathcal{D}(P)$ in a given interval.
This query becomes the same as checking if a given path in the tree $\mathcal{H}(P)$ can be extended to a maximal execution graph in $\mathcal{D}(P)$, for which we can use \trust.

Finally, we set the parameter $\theta$ to be $\sqrt{M}$ using which we can bound the entire run time by $\tilde{O}(\sqrt{M})$. 
%We provide the details of the construction in \cref{sec:stockmeyer}.

%\Cref{th:stockmeyer} is extremely surprising: we get an estimate for $C(P)$ by making (roughly) $\sqrt{C(P)}$ runs of the model checker!
%Consider a program with about 1M executions; with about 1K runs, we can get a good estimate with high probability!
%The catch is that the search has a high polynomial cost $O(|P|^6)$ and so we did not implement it.

%!TEX root=./main.tex

\section{Evaluation}\label{sec:experiment}

\begin{comment}
\paragraph*{Experimental Setup}
    We conducted all experiments on a Dell Latitude 5450 system running
a custom Debian-based distribution with an Intel Core Ultra 5 135H CPU
(18 cores @ 4.60 GHz) and 32GB of RAM.
%
We set a timeout of 60 minutes and a memory limit of 512MB
(denoted by \timeout\ and \oom). 
\MZ{do we need this? Experiments were run on my local Mac}
\RM{if you ran on your laptop, why not give the specs of your laptop?
Where is OOM and timeout used later?} \MZ{that was why I asked, we ran this simulator and we didn't report performance }

\end{comment}

\subsubsection*{Implementation and Setup}

We implemented \textbf{Algorithm S} in a tool called \testor. 
It is built on top of \jmc~\cite{jmc}, an implementation of \trust for Java programs~\cite{condpor}.
%%(details on the implementation of the \testor algorithm are provided in \cref{app:testor-impl}). 
For our evaluation, we also developed an ``offline'' version of \testor that runs \textbf{Algorithm S} 
on the tree ($\mathcal{D}(P)$) by running \jmc until completion and logging the events.
% Since tree logger in \jmc defers state transitions whenever a state has a unique $\fwdtikz$ successor and only materializes states at branching points and $\bwdtikz$, \MZ{@mohammad, is there anything to be addressed with the online implementation?} \MO{I've added a section in the appendix to discuss the implementation.}
% the resulting tree model is more compact than one in which every event is treated as a separate state transition.
%We apply offline \testor to this logged tree, exactly 
% simulating the behavior it would have when embedded in a stateless \dpor exploration of $\mathcal{D}(P)$ 
% while isolating its statistical properties from runtime overheads.
%
%For our evaluation, we additionally developed an ``offline'' variant of \testor that operates on the pre-collected $\mathcal{D}(P)$, obtained by running \jmc to completion.
The offline variant simulates exactly the behavior of \testor as it would operate when embedded in a stateless \dpor exploration of $\mathcal{D}(P)$, thereby isolating its statistical properties from runtime overheads.

As a baseline, we also implemented \textbf{Algorithm P} directly within \jmc.
However, we found that the variance of this algorithm is so large that, in all of our experiments, the estimates never converge into a meaningful value.
All experiments were conducted on a MacBook Pro equipped with 34GB of RAM and an M3 chip.

\subsubsection*{Benchmarks}
Our benchmark set consists of handcrafted (\textsc{IncNTest}, \textsc{FineCounter}) and SV-COMP-\texttt{pthread} \cite{svcomp} benchmarks (\textsc{Big0}, \textsc{Fib1}, \textsc{Sigma}, \textsc{SVQueue2}, \textsc{SVQueue3}, and \textsc{SVStack2}),
  textbook concurrent lists/stacks/queues (\textsc{CoarseList} \cite{book:herlihy-shavit}, \textsc{FineList} \cite{finelist}, \textsc{OptList} \cite{book:herlihy-shavit}, \textsc{UbQueue} \cite{book:herlihy-shavit}, 
  \textsc{LbQueue} \cite{book:herlihy-shavit}, and \textsc{AgmStack} \cite{agmstack}), and real-world data structures such as \textsc{TimestampStack} \cite{timestamped-stack}.

\subsection{Overview }
We consider the following two configurations in our evaluation.

\begin{enumerate}[label=\textbf{(\Alph*)}, leftmargin=2em]
\item \label{phase:online}
We evaluate the practical potential of \testor by running it on a representative set of benchmarks with varying workloads.
For each benchmark suite, we run online \testor for 1 minute and 5 minutes, assess the accuracy of the resulting estimates, and compare its running time against that of \jmc run to completion or until timeout.
% \item \label{phase:acp}\textbf{Approximate Counting Problem.}In this part, we ask how quickly do our estimators converge in practical benchmarks and how well the estimators scale as the workload (and consequently the state space) increases.
  %
\item \label{phase:acp} We then conduct a principled evaluation of the offline implementation. The offline setting enables a cleaner experimental setup, isolates the estimator from runtime-dependent effects, and improves reproducibility. We run the estimator on $\mathcal{D}(P)$ until it converges to a value close to the ground truth, and examine how its behavior varies across different budget settings.

%  We evaluate evaluating \testor and \pestor across synthetic benchmarks consists of handcrafted (IncNTest, FineCounter) and SV-COMP \texttt{pthread} category \cite{svcomp} benchmarks (Big0, Fib1, Sigma, SV Queue2, SV Queue3, and SV Stack2),
%  classic concurrent stacks/queues (CoarseList\cite{book:herlihy-shavit}, FineList\cite{finelist}, OptList\cite{book:herlihy-shavit}, UbQueue\cite{book:herlihy-shavit}, LbQueue\cite{book:herlihy-shavit}, and AgmStack\cite{agmstack}), and real-world data structures (timestamped stacks\cite{timestamped-stack}). 
%  We examine practical convergence on standard concurrent patterns.

  %\item \label{phase:trust-cost}\textbf{\jmc\ Cost Estimation.}In this part, we use \testor to evaluate the cost of model checking: is \testor effective in estimating the total cost of exploration?
\end{enumerate}  

\begin{table}[tbp]
\centering
\tiny

\begin{tabular}{llcccc}
\toprule
Benchmark & Param & Est.\ @1m & Est.\ @5m & \jmc count & \jmc status \\
\midrule
\multirow{3}{*}{CoarseList} & 8 & $(4.63 \pm 1.49) \times 10^{4}$ & $(4.13 \pm 0.77) \times 10^{4}$ & $4.03 \times 10^{4}$ & 1m 34s \\
 & 9 & $(6.72 \pm 5.65) \times 10^{5}$ & $(4.12 \pm 2.76) \times 10^{5}$ & $3.63 \times 10^{5}$ & 14m 11s \\
 & 10 & $(1.93 \pm 0.98) \times 10^{6}$ & $(2.12 \pm 0.54) \times 10^{6}$ & $>1.13 \times 10^{6}$ & TO at 1h \\
\midrule
\multirow{3}{*}{Fib1} & 10 & $(1.84 \pm 0.27) \times 10^{5}$ & $(1.87 \pm 0.06) \times 10^{5}$ & $1.85 \times 10^{5}$ & 6m 46s \\
 & 11 & $(6.93 \pm 1.10) \times 10^{5}$ & $(7.02 \pm 0.51) \times 10^{5}$ & $7.05 \times 10^{5}$ & 27m 15s \\
 & 12 & $(2.62 \pm 0.09) \times 10^{6}$ & $(2.61 \pm 0.16) \times 10^{6}$ & $>1.50 \times 10^{6}$ & TO at 1h  \\
\midrule
\multirow{3}{*}{FineCounter} & 11 & $(9.84 \pm 1.87) \times 10^{4}$ & $(1.21 \pm 0.11) \times 10^{5}$ & $8.64 \times 10^{4}$ & 2m 57s \\
 & 12 & $(7.83 \pm 3.12) \times 10^{5}$ & $(5.75 \pm 1.49) \times 10^{5}$ & $5.18 \times 10^{5}$ & 18m 3s \\
 & 13 & $(3.35 \pm 1.87) \times 10^{6}$ & $(3.78 \pm 1.48) \times 10^{6}$ & $>1.60 \times 10^{6}$ & TO at 1h  \\
\midrule
\multirow{3}{*}{FineList} & 8 & $(5.60 \pm 0.76) \times 10^{4}$ & $(4.37 \pm 0.37) \times 10^{4}$ & $4.03 \times 10^{4}$ & 2m 16s \\
 & 9 & $(3.25 \pm 2.55) \times 10^{5}$ & $(3.10 \pm 0.67) \times 10^{5}$ & $3.63 \times 10^{5}$ & 22m \\
 & 10 & $(3.13 \pm 0.67) \times 10^{5}$ & $(3.58 \pm 2.30) \times 10^{6}$ & $>8.85 \times 10^{5}$ & TO at 1h  \\
\midrule
\multirow{3}{*}{IncNTest} & 5 & $(1.43 \pm 0.06) \times 10^{4}$ & $(1.43 \pm 0.03) \times 10^{4}$ & $1.44 \times 10^{4}$ & 19s \\
 & 6 & $(4.96 \pm 0.88) \times 10^{5}$ & $(5.68 \pm 1.00) \times 10^{5}$ & $5.18 \times 10^{5}$ & 5m 15s \\
 & 7 & $(2.41 \pm 0.56) \times 10^{7}$ & $(2.42 \pm 0.51) \times 10^{7}$ & $>5.18 \times 10^{6}$ & TO at 1h  \\
\midrule
\multirow{3}{*}{OptList} & 4 & $(1.78 \pm 0.05) \times 10^{3}$ & $(1.77 \pm 0.05) \times 10^{3}$ & $1.69 \times 10^{3}$ & 15s \\
 & 5 & $(4.59 \pm 1.92) \times 10^{5}$ & $(4.83 \pm 1.11) \times 10^{5}$ & $4.57 \times 10^{5}$ & 16m 45s \\
 & 6 & $(1.94 \pm 3.20) \times 10^{8}$ & $(3.27 \pm 4.75) \times 10^{8}$ & $>1.16 \times 10^{6}$ & TO at 1h  \\
\midrule
\multirow{3}{*}{SVStack2} & 10 & $(6.73 \pm 1.18) \times 10^{4}$ & $(6.75 \pm 0.11) \times 10^{4}$ & $6.87 \times 10^{4}$ & 3m 15s \\
 & 11 & $(2.37 \pm 0.54) \times 10^{5}$ & $(2.53 \pm 0.30) \times 10^{5}$ & $2.42 \times 10^{5}$ & 11m 39s \\
 & 12 & $(7.60 \pm 0.99) \times 10^{5}$ & $(8.91 \pm 0.52) \times 10^{5}$ & $8.59 \times 10^{5}$ & 43m 39s \\
\midrule
\multirow{3}{*}{SVQueue2} & 9 & $(5.07 \pm 0.29) \times 10^{4}$ & $(4.73 \pm 0.19) \times 10^{4}$ & $4.86 \times 10^{4}$ & 4m 34s \\
 & 10 & $(1.88 \pm 0.22) \times 10^{5}$ & $(1.92 \pm 0.20) \times 10^{5}$ & $1.85 \times 10^{5}$ & 19m 38s \\
 & 11 & $(6.63 \pm 2.31) \times 10^{5}$ & $(7.38 \pm 0.88) \times 10^{5}$ & $>4.89 \times 10^{5}$ & TO at 1h  \\

\midrule
\multirow{3}{*}{SVQueue3}  & 15 & $(8.87 \pm 1.15) \times 10^{4}$ & $(1.13 \pm 0.11) \times 10^{5}$ & $1.09 \times 10^{5}$ & 8m 21s \\
 & 16 & $(2.05 \pm 0.85) \times 10^{5}$ & $(2.25 \pm 0.50) \times 10^{5}$ & $2.18 \times 10^{5}$ & 17m 5s \\
 & 17 & $(4.25 \pm 2.08) \times 10^{5}$ & $(5.43 \pm 2.74) \times 10^{5}$ & $4.37 \times 10^{5}$ & 36m 31s \\
\midrule
\multirow{3}{*}{UBQueue} & 9 & $(4.60 \pm 0.50) \times 10^{4}$ & $(4.61 \pm 0.10) \times 10^{4}$ & $4.61 \times 10^{4}$ & 1m 23s \\
 & 10 & $(5.01 \pm 1.86) \times 10^{5}$ & $(4.58 \pm 0.20) \times 10^{5}$ & $4.61 \times 10^{5}$ & 13m 16s \\
 & 11 & $(2.28 \pm 0.92) \times 10^{6}$ & $(2.98 \pm 0.30) \times 10^{6}$ & $>1.85 \times 10^{6}$ & TO at 1h  \\
 \midrule
 \multirow{3}{*}{LBQueue} & 13 & $(1.03 \pm 0.88) \times 10^{5}$ & $(6.23 \pm 0.76) \times 10^{4}$ & $6.91 \times 10^{4}$ & 3m 52s \\
 & 14 & $(2.48 \pm 1.12) \times 10^{5}$ & $(3.49 \pm 1.03) \times 10^{5}$ & $3.46 \times 10^{5}$ & 19m 45s \\
 & 15 & $(5.98 \pm 2.57) \times 10^{5}$ & $(1.65 \pm 1.45) \times 10^{6}$ & $>9.69 \times 10^{5}$ & TO at 1h  \\
 \midrule
\multirow{3}{*}{AgmStack} & 5 & $(6.25 \pm 0.36) \times 10^{3}$ & $(6.15 \pm 0.15) \times 10^{3}$ & $6.23 \times 10^{3}$ & 19s \\
 & 6 & $(5.14 \pm 1.25) \times 10^{5}$ & $(4.28 \pm 0.76) \times 10^{5}$ & $4.48 \times 10^{5}$ & 11m 41s \\
 & 7 & $(4.36 \pm 2.75) \times 10^{7}$ & $(5.07 \pm 1.52) \times 10^{7}$ & $>1.87 \times 10^{6}$ & TO at 1h  \\
 \midrule
\multirow{3}{*}{Sigma} & 4 & $(1.49 \pm 0.09) \times 10^{4}$ & $(1.40 \pm 0.05) \times 10^{4}$ & $1.38 \times 10^{4}$ & 20s \\
 & 5 & $(4.36 \pm 4.00) \times 10^{6}$ & $(1.64 \pm 0.41) \times 10^{6}$ & $1.73 \times 10^{6}$ & 21m 27s \\
 & 6 & $(2.61 \pm 2.57) \times 10^{8}$ & $(4.16 \pm 5.21) \times 10^{8}$ & $>3.85 \times 10^{6}$ & TO at 1h  \\
 \midrule
\multirow{3}{*}{TimeStampStack} & 3 & $(2.94 \pm 0.11) \times 10^{2}$ & $(2.90 \pm 0.05) \times 10^{2}$ & $2.91 \times 10^{2}$ & 12s \\
 & 4 & $(3.34 \pm 0.60) \times 10^{4}$ & $(4.71 \pm 0.60) \times 10^{4}$ & $3.26 \times 10^{4}$ & 1m 54s \\
 & 5 & $(1.41 \pm 2.29) \times 10^{7}$ & $(3.50 \pm 1.09) \times 10^{6}$ & $>7.76 \times 10^{5}$ & TO at 1h  \\
\bottomrule
\end{tabular}
\vspace{7pt}
\caption{Results of the online \testor with stochastic enumeration budget of 20 under fixed 1-minute and 5-minute time constraints, compared with the corresponding \jmc runs. For each benchmark and parameter, each estimator entry reports the mean and standard deviation over 5 independent runs. The \jmc \textit{count} column reports the number of executions explored by \jmc, which is exact when the run finishes and partial when the run times out. The \jmc \textit{status} column reports the \jmc execution time, or indicates a timeout when the run does not finish within the time limit.}
\label{tab:online-estimator-results}
\end{table}

\subsection{Configuration \ref{phase:online}: Effectiveness of \testor}

In this configuration, we evaluate \testor as a tool for estimating the state space of a program.

\subsubsection*{Setup}

A practical use case for \testor is as a lightweight budgeting tool: one first runs the estimator for a short period, and then uses the resulting estimate to inform how long the model checker should continue running.
Accordingly, in this configuration, we run \testor for 1 minute and 5 minutes, and run \jmc with a 1-hour time limit.

We consider a set of standard concurrent benchmarks whose workload scales with input parameters, such as the number of workers or the number of operations.
For each benchmark, we select three workload configurations: one for which the model checker finishes comfortably within the time limit, one expected to be close to the largest workload the model checker can still handle within the time limit, and one for which the model checker is unlikely to terminate within the time limit.
For each benchmark–workload configuration, we run \testor five times and record the mean estimate, as well as the standard deviation.
We adopt a stochastic enumeration budget of $\mathcal{B} = 20$, informed by the findings of our study in Configuration~\ref{phase:acp}.

\subsubsection*{Observations}
Table~\ref{tab:online-estimator-results} reports the performance of \testor under fixed one-minute and five-minute budgets, alongside the corresponding model-checking runs. For each benchmark and parameter setting, each estimator entry reports the mean and standard deviation over 5 independent runs. The \jmc \emph{count} column gives the number of executions explored by \jmc: this is exact when \jmc terminates within the time limit, and partial when the run times out. The \jmc \emph{status} column reports either the wall-clock time of the completed model-checking run or the timeout status.
Overall, \testor yields encouraging results. On configurations for which \jmc terminates and the exact count is known, 99 out of 150 1-minute runs and 113 out of 150 5-minute runs fall within 80\%–120\% of the exact count.
On the subset of non-timeout configurations whose exact \jmc runs take more than 5 minutes, 47 out of 85 1-minute runs and 61 out of 85 5-minute runs fall within 80\%–120\% of the exact count. This suggests that even a short fixed budget often suffices to obtain a useful estimate, while a 5-minute budget yields a meaningful improvement in fine-grained accuracy. As expected, 1-minute estimates exhibit higher standard deviations than their 5-minute counterparts.
For the smallest configurations, exhaustive exploration finishes before the one minute and five minute deadline; however, there is no a priori way to predict that this will happen other than by computing estimates.

The table further illustrates that the practical value of  \testor extends beyond predicting the final count: it also reveals how difficulty scales across neighboring workload configurations. For instance, on \textsc{CoarseList}, the estimated count increases by only a factor of approximately five from parameter~9 to parameter~10. Although the parameter-10 model-checking run times out, this estimate suggests that the exact run may still be feasible within a few hours. Moreover, the partial count at timeout already covers roughly half of the estimated total, providing a meaningful measure of progress. In contrast, for \textsc{OptList}, the estimator predicts a jump from only hundreds of thousands of executions to hundreds of millions. This indicates a much sharper scalability cliff: even substantially extending the execution time is unlikely to make the larger instance tractable, and the roughly one million executions explored before timeout may represent only a negligible fraction of the total space.

We note that the issue of variance is not fully mitigated.
For example, one configuration of \textsc{CoarseList} exhibits an unusually large standard deviation. 
Upon inspection, this is caused not by pervasive instability across all repeated runs, but by a single outlying run among the five repetitions. 
%This suggests that, for certain program structures,  
%\testor may occasionally produce rare but substantial deviations even under identical time budgets.
We also observe that when the estimator returns a large estimate, the standard deviation tends to be larger as well, which is expected. 
A short runtime is unlikely to suffice for accurate estimation when the underlying search space is inherently enormous. Nevertheless, in such cases, even obtaining the correct order of magnitude is already very useful.

%!TEX root=./main.tex

\begin{table}[t]
  \centering
  \tiny
  \begin{tabular}{@{} l l S[scientific-notation=true] S[table-format=1.2] l @{}}
    \toprule
    \textbf{Concurrent Structures} & \textbf{Estimator} & \textbf{Relative Error} & \textbf{Success Ratio} & \textbf{Trials} \\
    \midrule

   \multirow{4}{*}{\makecell[l]{CoarseList\\ Insertion threads=9\\ Executions = \num{362880}}} &  B=1 & \num{1.42} & \num{0} & (failed) \\
    &  B=5      & \num{0.25} & \num{0.20} & (failed) \\
    &  B=20     & \num{0.0647} & {1 \ding{51}} & \num{187} \\
    &  B=50     & \num{0.0271} & {1\ding{51}} & \num{125} \\
    \midrule
    \multirow{4}{*}{\makecell[l]{FineList\\ Insertion threads=9\\ Executions = \num{362880}}} &  B=1 & \num{1.42} & \num{0} & (failed) \\
    &  B=5      & \num{0.25} & \num{0.20} & (failed) \\
    &  B=20     & \num{0.0647} & {1\ding{51}} & \num{187} \\
    &  B=50     & \num{0.0271} & {1\ding{51}} & \num{125} \\
    \midrule
     \multirow{4}{*}{\makecell[l]{OptList\\ Insertion threads=3\\ Deletion threads=2\\ Executions = \num{456995}}} & B=1 & \num{0.784} & \num{0} & (failed) \\
    &  B=5      & \num{0.32} & {\num{0.60}\,\ding{51}} & \num{1156} \\
    &  B=20     & \num{0.371} & {\num{0.80}\,\ding{51}} & \num{794} \\
    &  B=50     & \num{0.0851} & {1\ding{51}} & \num{210} \\
    \midrule
    \multirow{4}{*}{\makecell[l]{UbQueue\\ Insertion threads=5\\ Deletion threads=5\\Executions = \num{460800}}} &  B=1             & \num{0.217} & \num{0.40} & (failed) \\
    &  B=5            & \num{0.0654} & {1\ding{51}} & \num{350} \\
    &  B=20           & \num{0.0303} & {1\ding{51}} & \num{311} \\
    &  B=50           & \num{0.00845} & {1\ding{51}} & \num{126} \\ 
    \midrule
        \multirow{4}{*}{\makecell[l]{LbQueue\\  Enqueue threads=14\\Executions = \num{345600}}} &  B=1              & \num{0.787} & \num{0} & (failed) \\
    &  B=5             & \num{0.191} & {\num{0.60}\,\ding{51}}  & \num{1097} \\
    &  B=20            & \num{0.0604} & {1\ding{51}} & \num{165} \\
    &  B=50            & \num{0.0408} & {1\ding{51}} & \num{197}  \\ 
    \midrule
\multirow{4}{*}{\makecell[l]{AmgStack\\ Pusher threads=3\\ Poper thread threads=3\\ Exeuctions = \num{447576}}} & B=1 & \num{0.509} & \num{0.00} & (failed) \\
    &B=5      & \num{0.191} & {\num{0.60}\,\ding{51}}  & \num{1005} \\
    &  B=20     & \num{0.0469} & {1\ding{51}} & \num{456} \\
    & B=50     & \num{0.0429} & {1\ding{51}} & \num{357} \\
    \midrule
    
\multirow{4}{*}{\makecell[l]{TimeStampStack\\ Pusher threads=2\\ Poper thread threads=2\\ Executions = \num{32648}}} &   B=1 & \num{0.395} & \num{0.00} & (failed) \\
    &  B=5      & \num{0.21} & {\num{0.60}\,\ding{51}} & \num{1100} \\
    &  B=20     & \num{0.0376} & {1\ding{51}} & \num{479} \\
    &  B=50     & \num{0.0362} & {1\ding{51}} & \num{286} \\
    \bottomrule

  \end{tabular}

  \vspace{0.1in}
  \caption{Concurrent data structure benchmarks for estimator evaluation on  \textbf{maximal executions}. Columns (left→right): benchmark info, which contains input detail and total complete execution number; estimator configuration; average relative error after 2000 trials; success ratio (fraction of 5 seeds that converge within the 80–120\% band); and average trials to converge across the successful seeds (reported only if success ratio$\geq 0.6$; otherwise shown as “failed”). }
    \label{tab:convergence}

\end{table}

%!TEX root=./main.tex

\begin{table}[t]
  \centering
  \tiny
  \begin{tabular}{@{} l l S[scientific-notation=true] S[table-format=1.2] l @{}}
    \toprule
    \textbf{Synthetic Benchmarks} & \textbf{Estimator} & \textbf{Relative Error} & \textbf{Success Ratio} & \textbf{Trials} \\
    \midrule
\multirow{5}{*}{\makecell[l]{IncNTest\\  Incrementor threads=6\\ Executions = \num{518400}}} &   B=1 & \num{0.408} & \num{0.20} & (failed) \\
    & B=5      & \num{0.206} & {\num{0.60}\,\ding{51}}  & \num{575} \\
    &  B=20     & \num{0.256} & {1\ding{51}} & \num{909} \\
    &  B=50     & \num{0.0535} & {1\ding{51}} & \num{560} \\
    \midrule
\multirow{4}{*}{\makecell[l]{Big0\\ Executions = \num{69112}}} &   B=1              & \num{0.0333} & {1\ding{51}}& \num{159} \\
    &  B=5             & \num{0.0068} & {1\ding{51}} & \num{151} \\
    &  B=20            & \num{0.00767} & {1\ding{51}} & \num{112} \\
    &  B=50            & \num{0.00231} & {1\ding{51}} & \num{106}  \\
    \midrule
\multirow{4}{*}{\makecell[l]{FineCounter\\   Threads=12\\ Executions = \num{518400}}} &   B=1 & \num{0.651} & \num{0} & (failed) \\
    &  B=5      & \num{0.481} & \num{0.20} & (failed) \\
    &  B=20     & \num{0.0923} & {1\ding{51}} & \num{192} \\
    &  B=50     & \num{0.0454} & {1\ding{51}} & \num{239} \\
  
    \midrule
    \multirow{4}{*}{\makecell[l]{Fib1\\   operations=11\\ Executions = \num{705432}}} &    B=1              & \num{2.9} & \num{0.20} & (failed) \\
    &  B=5             & \num{0.0951} &  {1\ding{51}} & \num{570} \\
    &  B=20            & \num{0.0402} &  {1\ding{51}} & \num{184} \\
    &    B=50            & \num{0.0204} &  {1\ding{51}} & \num{136} \\ 
    \midrule
    \multirow{4}{*}{\makecell[l]{Sigma\\   operations=5\\ Executions = \num{1728000}}} &   B=1              & \num{0.53} & \num{0} & (failed) \\
    &  B=5             & \num{2.68} & {\num{0.60}\,\ding{51}} & \num{1342} \\
    &  B=20            & \num{0.616} &  {1\ding{51}} & \num{1022} \\
    &  B=50            & \num{0.108} & {\num{0.80}\,\ding{51}} & \num{552} \\ 
  
    \midrule
      \multirow{4}{*}{\makecell[l]{SV Queue2\\   Operations=10\\ Executions = \num{184756}}} &  B=1 & \num{0.651} & \num{0} & (failed) \\
    &  B=5      & \num{0.105} &  {1\ding{51}} & \num{395} \\
    &  B=20     & \num{0.0319} &  {1\ding{51}} & \num{174} \\
    &  B=50     & \num{0.0104} &  {1\ding{51}} & \num{127} \\
     \midrule
        \multirow{4}{*}{\makecell[l]{SV Queue3\\   Operations=17\\ Executions = \num{436728}}} & B=1      & \num{0.956} & \num{0.00} & (failed) \\
    & B=5      & \num{0.396} & \num{0.40} & (failed) \\
    & B=20     & \num{0.0696} &  {1\ding{51}} & \num{715} \\
    & B=50     & \num{0.0162} &  {1\ding{51}} & \num{140} \\
     \midrule
      \multirow{4}{*}{\makecell[l]{SV Stack2\\   Operations=11\\ Executions = \num{241616}}} &  B=1 & \num{0.408} & \num{0} & (failed) \\
    &  B=5      & \num{0.266} &  {1\ding{51}} & \num{791} \\
    &  B=20     & \num{0.0314} &  {1\ding{51}} & \num{259} \\
    &  B=50     & \num{0.0306} &  {1\ding{51}} & \num{135} \\
    \bottomrule

  \end{tabular}

  \vspace{0.1in}
  \caption{Synthetic benchmarks for estimator evaluation on \textbf{maximal executions}. Columns (left→right): benchmark info; estimator configuration; average relative error after 2000 trials; success ratio (fraction of 5 seeds that converge within the 80–120\% band); and average trials to converge across the successful seeds (reported only if success ratio$\geq 0.6$; otherwise shown as “failed”). }
    \label{tab:synthtic_converge}

\end{table}

\subsection{Configuration \ref{phase:acp}: Convergence of Estimation
}

%Although unbiased estimators converge to the expected value, this guarantee is often meaningless in practice: high variance can lead to extreme instability or even non-convergence under any practical finite budgets.
% Thus, our first question is whether our estimators exhibit practical convergence within a modest sampling budget when estimating program behavior spaces.

We use the offline implementation to  explore the convergence behavior of \testor's estimation in a principled manner.

\subsubsection*{Setup}
In this configuration, we consider offline-\testor has succeeded if its estimate converges to a range around the real count. 
Specifically, we declare convergence when the running mean of offline-\testor enters the $\pm20\%$ relative error band and remains within it for a stability window of 
$50$ iterations, while also satisfying a flatness condition: the moving average changes by less than a $2\%$ relative threshold over a stability window of 100 iterations.
We evaluate four configurations of offline-\testor with stochastic enumeration budgets $\mathcal{B} \in \{1, 5, 20, 50\}$. 
Each configuration is executed for a sequence of $2000$ runs, and to account for randomness, each sequence is repeated with five independent random seeds. 
A configuration is deemed to successfully converge if it converges under at least three of the five seeds. 
To assess the effectiveness of our estimators, we focus on two metrics:
\begin{enumerate}
    \item \emph{Relative error after 2000 runs}, computed as $\text{estimated count} / \text{exact count}$; and
    \item \emph{Trials to convergence}, reported only when a configuration converges on a given benchmark and averaged over the seeds for which convergence was achieved.
\end{enumerate}

\subsubsection*{Benchmarks}

For this evaluation configuration, we study \testor on 15 benchmarks: the same benchmark suite used in Configuration~\ref{phase:online}, plus one additional benchmark (Big0) that is small and does not have a tunable parameter, and was therefore not included in Configuration~\ref{phase:online}.
 
The first seven benchmarks (presented in \cref{tab:convergence}) correspond to typical uses of concurrent data structures, 
including \textsc{CoarseList}, \textsc{FineList}, \textsc{OptList}, \textsc{UbQueue}, \textsc{LbQueue}, \textsc{AMGStack}, and \textsc{TimestampStack}. 
These programs exhibit high thread contention, yielding large and complex behavior spaces. 
The remaining eight (presented in \cref{tab:synthtic_converge}) are synthetic benchmarks designed to isolate specific concurrency 
patterns and stress particular aspects of the estimators’ convergence behavior.
%We omitted the others because they are dominated by blocked executions---we will return to them in Configuration (B) \MO{I think this citation is wrong}.

\subsubsection*{Results Summary}

\Cref{tab:convergence} and \ref{tab:synthtic_converge} present the results of the evaluation. 
%For each benchmark, we first report the details of the program. 
For each estimator configuration, we report (i) the average relative error after 2,000 trials, averaged over all 5 seeds, 
(ii) the success ratio (i.e., the fraction of seeds that satisfy the convergence criterion), and 
(iii) when this ratio exceeds 0.6 (at least 3 out of 5 seeds), the average number of trials-to-convergence over the converged seeds (rounded to int).

We note that the baseline \textbf{Algorithm P} failed to converge on every benchmark.

% RM: This is not unexpected, right? THe number depends on the concurrency structure of the test
%We also note that different benchmarks can induce the same underlying tree. 
%For instance, our \textsc{CoarseList} and \textsc{FineList} benchmarks, although implemented as different data structures, yield exactly the same estimates when we run them under the same random seeds.

\begin{description}
    \item[Convergence:] Positively, we observe offline-\testor with budgets of 20 and 50 converges on all benchmarks within a very small 
    number of trials--—typically a few hundred--—even when the underlying behavior space is much larger, 
    ranging from 30 thousand up to 1.7 million. We therefore believe that our evaluation demonstrates the practicality of our estimator.
    On the other hand, offline-\testor with a budget of 1 (\textbf{Algorithm T}) failed to converge on all benchmarks except Big0.
    \item[Relative Error:] Offline-\testor with budgets of 20 and 50 generally yields very small relative errors after 2000 runs. While a budget of 50 does not show a clear advantage in achieving convergence compared to a budget of 20, it does provide a noticeable improvement in the final relative error. We also observe that, for most benchmarks, whenever a configuration reported convergence, the relative error induced by the configuration remains well within the $20\%$ band.
    %%indicating that our convergence criterion is appropriately chosen. 
\end{description}

\subsection{Configuration~\ref{phase:acp}: Scaling with Budget}
%\subsubsection*{Scaling with Workload}Based on the previous result, a budget of 20 appears sufficient for our purpose. It offers comparable accuracy (converged for all benchmarks with success ratio  while requiring fewer computational resources compared to budget 50, so we use it as the default setting in the subsequent analysis.In this section, we first study its expected scaling behavior by examining how convergence evolves as the workload of the same benchmark grows.We examine several representative benchmarks under different workloads, varying the number of threads and operations to observe how the estimator’s convergence time scales with the growth of the space. We continue running with 2,000 runs and report the average number of trials required for convergence. As shown in Figure \ref{fig:workload}, the estimator converges quickly on small search spaces and continues to scale gracefully as the search space grows. Even for large spaces, the increase in convergence time remains sublinear, indicating that the estimator maintains efficiency across several orders of magnitude. More importantly, we find that for large workload, the number of trials required for our estimator to converge is typically at least thousands of times smaller than the total number of behaviors, demonstrating that it achieves stable estimates after exploring only a minute fraction of the search space. %

\begin{figure}[t]
  \centering
  \includegraphics[width=0.8\linewidth]{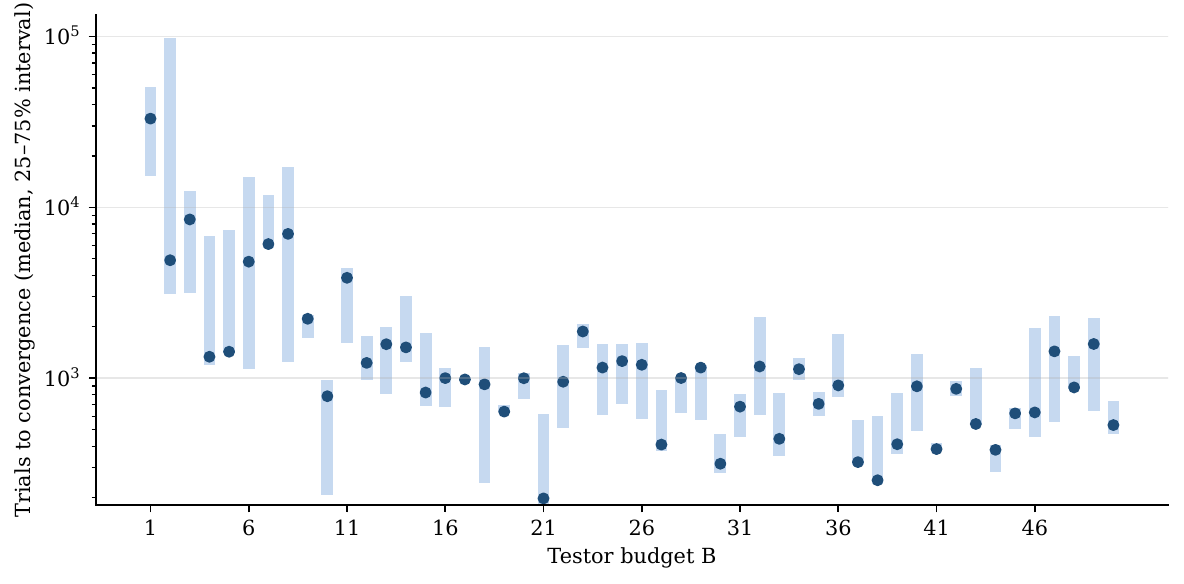}
  \caption{Median trials to convergence across budgets, with the 25th--75th percentile (IQR) band shown; the y-axis is log-scaled.}
  \label{fig:trial}
\end{figure}

%\subsubsection*{Scaling with Budget}

Based on the previous results, a budget of 20 appears sufficient for most of our purposes, offering comparable accuracy while requiring fewer computational resources than a budget of 50. To better understand the effect of the budget on convergence, we now revisit the \textsc{Sigma} benchmark, the only benchmark in our suite for which offline \testor with budget~20 requires more than 1,000 trials to converge. \textsc{Sigma} also has the largest exact count in the suite, at 1.7 million executions. 
Using this benchmark as a stress test, we study how the number of trials required for convergence varies with the budget. 
For each budget, we run \testor with five seeds and increase the trial cap until all seeds satisfy our convergence criterion or the cap reaches one million trials.

In Figure~\ref{fig:trial}, we plot the 25–75\% (IQR) band and mark the median explicitly. We observe that, as expected, larger budgets generally lead to more stable convergence. Some smaller budgets may still struggle to converge even after $10^5$ trials while larger budgets over 20 typically converge for all seeds within only a few hundreds trials. 
% This observation also supports the view that stochastic enumeration is effective at variance reduction: 
% for a fixed total amount of work, running \testor with budget~$k$ produces significantly more accurate estimates than running \testor with budget~$1$ for $k$ separate runs.
%
However, the pattern shown in the figure is not strictly monotone:  these irregular spikes are due to randomness under a finite number of seeds—at certain budgets we are simply ``unlucky.” In expectation over the randomization, increasing the budget can only improve, or at worst leave unchanged, the estimator’s performance.

\section{Related Work}

In this paper, we have studied the problem of estimating the number of Mazurkiewicz traces in a program. 
% \paragraph{Approximate Counting}{
(Approximate) Counting of combinatorial objects is a classical topic in complexity theory \cite{Valiant79,JerrumVV86}.
A central problem here is \emph{propositional model counting} ($\#SAT$).
Our proof of hardness of the estimation problem reduces from $\#SAT$.
There are approximate counters for $\#SAT$, such as ApproxMC \cite{chakraborty2013scalableapproximatemodelcounter}, that run in randomized polynomial time with access to
an NP-oracle.
These techniques are not easily applied in our context because compiling all executions of a real-world concurrent program into a SAT instance is a formidable problem.
Instead, we take the route of unbiased estimators obtained from stateless search.

Our work is also close to approximate counting on automata, where one asks for the number of words of a given length accepted by an NFA. 
This problem admits a fully polynomial randomized approximation scheme (FPRAS) \cite{GoreJKSM97,10.1145/3422648.3422661}, 
with several algorithmic improvements \cite{meel2024fasterfprasnfa,meel2025practicalfprasnfaexploiting}.
A recent paper \cite{ColnetMM26} considers the problem of counting traces in a regular trace language and shows an FPRAS for the problem.
It might seem that our estimation problem can be reduced to these problems by taking the automaton for each thread, forming their partially synchronous product, 
intersecting with the regular language of lexicographic normal forms of trace equivalence, and finally asking how many words of some appropriate length the resulting automaton accepts. 
However, the product automaton is exponential in the number of threads and variables, so even FPRAS algorithms for approximate counting on automata become exponential in the number of threads. 
%% Moreover, even finding the automata for each thread of a concurrent program statically is a difficult problem.

%Similar problems have been studied in the literature, albeit with different approaches. 
Our work is also related to estimating coverage in a model checker.
Penix and Visser~\cite{penix} add standard code-coverage metrics to Java PathFinder \cite{jpf}, tracking which bytecode instructions and branch outcomes are exercised during state-space exploration.
Bron, Ziv, and Ur~\cite{bron05} introduce synchronization coverage for concurrent Java programs, defining a small set of concurrency-specific coverage tasks (for example, whether a synchronized block ever caused another thread to block) and treating full coverage as a practical stopping criterion.
Nevertheless, code coverage and its derivatives are known to be imprecise measures for the space of a program.
% and note that the resulting coverage curves often “converge’’ once further exploration rarely reveals new program elements. 
Chockler, Kupferman, and Vardi~\cite{10.1007/978-3-540-39724-3_11} propose mutation-based coverage metrics for model checking that ask which states, transitions, and subformulas actually matter for establishing
if a specification is satisfied.  
% An element is considered covered if mutating it falsifies the property or makes it hold vacuously, and they give symbolic algorithms instantiating this idea for several standard coverage notions. 
Taleghani and Atlee \cite{5431751} estimate state-space coverage in explicit-state model checking by sampling unexplored transitions using BFS and extrapolating the number of remaining states. 
However the estimate is heuristic and comes with no guarantees of being an unbiased estimator. 
% Path profiling proposed by Ball and Larus \cite{10.5555/243846.243857} is another dynamic coverage metric: instead of asking which blocks or edges are covered, it measures how often each control-flow path in a procedure is executed under a given workload. 
% The classic scheme of Ball and Larus also relies on a DAG view of executions, obtained by cutting loop back-edges in the control-flow graph, and uses this structure to profile path frequencies, which shares certain similar sprit of our DAG estimator.
These approaches describe which parts of the system and specification are exercised, or roughly how much of the explicit state graph has been visited, but not how far verification has progressed within the underlying execution space. 
In contrast, our contribution is complementary: we start with a complete set of behaviors as our coverage goal and we estimate the size of the set.
Moreover, our notion of behavior identifies two causally identical executions as one.
Closest to our work is the estimate of \cite{gator}, which inspired this work.
As discussed before, our paper provides the precise complexity bounds for the estimation problem and ensures that the estimator is unbiased.

Purely statistical techniques have also been proposed in recent  fuzzing literature.
For example, the use of Good-Turing estimator \cite{10.1145/3210309, 10.1145/3611019,10.1145/3468264.3468570} has been proposed to give statistical evidence to continue or to stop a fuzzing campaign; it is unclear if their estimate can
be converted to an estimate for the concurrent state space.
More heuristic notions of saturation or past coverage trends have also been proposed to predict the expected time to full coverage  \cite{10.1145/3597503.3639198,10.1145/3597926.3598043}. 
% As such, these problems are specific to purely memoryless testing techniques like greybox and blackbox fuzzing, that aim to discover as many bugs as possible.
% Several works in fuzzing propose estimators to reason about how much testing has achieved. The STADS framework and Entropic power schedules model fuzzing as a species-discovery process and adapt Good–Turing–style estimators to predict, for example, the probability that the next test input will exercise a new behaviour and the eventual number of behaviours a fuzzer might discover~\cite{10.1145/3210309, 10.1145/3611019}. Subsequent work extrapolates future coverage from past coverage curves or uses the saturation of predicted-vulnerable functions as a heuristic stopping rule. These approaches are primarily evaluated empirically and, in the greybox setting, do not come with provable unbiasedness or variance guarantees. 
%
% Our setting is different: we aim to count semantic execution equivalence classes of a concurrent program, and we design estimators with explicit mathematical guarantees for this counting problem, rather than heuristic progress indicators for fuzzing campaigns.

\section{Conclusions and Future Work}

In this paper, we looked at the problem of estimating the size of the space explored by a model checker on a bounded concurrent program. Even though there has been previous work on related questions, they were based on heuristic approaches without any theoretical results. 
We lay the theoretical foundations for this problem, both by exhibiting hardness results as well as by coming up with the first poly-time unbiased estimator for the problem. 
An implementation of our procedure on a wide range of shared-memory concurrency benchmarks shows that our estimator provides robust estimates whilst only using a modest number of samples.
This work  opens up avenues for interesting future work including, accounting 
for weaker equivalences such as reads-from equivalence~\cite{readsFrom2019},
reads-value-from equivalences~\cite{readsValue2021} and recent intermediate
equivalences such as grain equivalence~\cite{FarzanMathur2024} as well as weaker memory models.

\begin{acks}
We thank the reviewers for their helpful comments. We also thank Viktor Vafeiadis and Michalis Kokologiannakis for their suggestions, which greatly improved the presentation of the paper. This research was sponsored in part by the \grantsponsor{DFG}{Deutsche Forschungsgemeinschaft}{https://www.dfg.de/} project \href{https://gepris.dfg.de/gepris/projekt/389792660}{\grantnum{DFG}{389792660}} TRR 248–CPEC. Umang Mathur is partially supported by the National Research Foundation, Singapore, and Cyber Security Agency of Singapore under its National Cybersecurity R\&D Programme (Fuzz Testing <NRF-NCR25-Fuzz-0001>). Any opinions, findings and conclusions, or recommendations expressed in this material are those of the author(s) and do not reflect the views of National Research Foundation, Singapore, and Cyber Security Agency of Singapore.
\end{acks}
\section*{Data-Availability Statement}
The artifact for this paper is publicly available on Zenodo~\cite{zenodo_artifact}. The source code of \jmc, including the implementations mentioned in this work, is publicly available at~\cite{jmc}.

\bibliographystyle{ACM-Reference-Format}
\bibliography{ref.bib}

@article{GoreJKSM97,
  author       = {Vivek Gore and
                  Mark Jerrum and
                  Sampath Kannan and
                  Z. Sweedyk and
                  Stephen R. Mahaney},
  title        = {A Quasi-Polynomial-Time Algorithm for Sampling Words from a Context-Free
                  Language},
  journal      = {Inf. Comput.},
  volume       = {134},
  number       = {1},
  pages        = {59--74},
  year         = {1997},
  url          = {https://doi.org/10.1006/inco.1997.2621},
  doi          = {10.1006/INCO.1997.2621},
  bibsource    = {dblp computer science bibliography, https://dblp.org}
}

@article{Valiant79,
  author       = {Leslie G. Valiant},
  title        = {The Complexity of Enumeration and Reliability Problems},
  journal      = {{SIAM} J. Comput.},
  volume       = {8},
  number       = {3},
  pages        = {410--421},
  year         = {1979},
  url          = {https://doi.org/10.1137/0208032},
  doi          = {10.1137/0208032},
  bibsource    = {dblp computer science bibliography, https://dblp.org}
}

@article{JerrumVV86,
  author       = {Mark Jerrum and
                  Leslie G. Valiant and
                  Vijay V. Vazirani},
  title        = {Random Generation of Combinatorial Structures from a Uniform Distribution},
  journal      = {Theor. Comput. Sci.},
  volume       = {43},
  pages        = {169--188},
  year         = {1986},
  url          = {https://doi.org/10.1016/0304-3975(86)90174-X},
  doi          = {10.1016/0304-3975(86)90174-X},
  bibsource    = {dblp computer science bibliography, https://dblp.org}
}

@article{huberkim,
 title = "Weighted-Ensemble {B}rownian Dynamics Simulations for Protein Association Reactions",
 author = "Gary A. Huber and  Sangtae Kim",
 journal = "Biophysical Journal",
 volume = 70, 
 year = 1996, 
 pages = "97--110",
 doi = "10.1016/S0006-3495(96)79552-8"
}

@article{ColnetMM26,
  author       = {Alexis de Colnet and
                  Kuldeep S. Meel and
                  Umang Mathur},
  title        = {Counting and Sampling Traces in Regular Languages},
  journal      = {Proc. {ACM} Program. Lang.},
  volume       = {10},
  number       = {{POPL}},
  pages        = {2352--2379},
  year         = {2026},
  url          = {https://doi.org/10.1145/3776723},
  doi          = {10.1145/3776723},
  bibsource    = {dblp computer science bibliography, https://dblp.org}
}

@InProceedings{condpor,
  author =	{Khoshechin Jorshari, Mohammad Hossein and Kokologiannakis, Michalis and Majumdar, Rupak and Nagendra, Srinidhi},
  title =	{{Optimal Concolic Dynamic Partial Order Reduction}},
  booktitle =	{36th International Conference on Concurrency Theory (CONCUR 2025)},
  pages =	{26:1--26:22},
  year =	{2025},
  volume =	{348},
  publisher =	{Schloss Dagstuhl -- Leibniz-Zentrum f{\"u}r Informatik},
  address =	{Dagstuhl, Germany},
  URL =		{https://drops.dagstuhl.de/entities/document/10.4230/LIPIcs.CONCUR.2025.26},
  ids = {condpor}
}

@Article{	  trust,
  title		= {Truly stateless, optimal dynamic partial order reduction},
  author	= {Kokologiannakis, Michalis and Marmanis, Iason and
		  Gladstein, Vladimir and Vafeiadis, Viktor},
  journal	= {Proc. ACM Program. Lang.},
  year		= {2022},
  month		= jan,
  issue_date	= {January 2022},
  publisher	= {ACM},
  address	= {New York, NY, USA},
  volume	= {6},
  number	= {POPL},
  doi		= {10.1145/3498711},
  articleno	= {49},
  numpages	= {26},
  ids		= {trust}
}

@Article{	  must,
  title		= {Model Checking Distributed Protocols in Must},
  author	= {Enea, Constantin and Giannakopoulou, Dimitra and
		  Kokologiannakis, Michalis and Majumdar, Rupak},
  journal	= {Proc. ACM Program. Lang.},
  year		= {2024},
  month		= oct,
  issue_date	= {October 2024},
  publisher	= {ACM},
  address	= {New York, NY, USA},
  volume	= {8},
  number	= {OOPSLA2},
  doi		= {10.1145/3689778},
  articleno	= {338},
  numpages	= {28},
  ids		= {must}
}

@InProceedings{	  verisoft,
  title		= {Model checking for programming languages using
		  {VeriSoft}},
  author	= {Patrice Godefroid},
  booktitle	= {POPL 1997},
  year		= {1997},
  pages		= {174--186},
  location	= {Paris, France},
  publisher	= {ACM},
  address	= {New York, NY, USA},
  doi		= {10.1145/263699.263717},
  ids		= {verisoft}
}

@InProceedings{	  dpor,
  title		= {Dynamic partial-order reduction for model checking
		  software},
  author	= {Cormac Flanagan and Patrice Godefroid},
  booktitle	= {POPL 2005},
  year		= {2005},
  pages		= {110--121},
  publisher	= {{ACM}},
  address	= {New York, NY, USA},
  doi		= {10.1145/1040305.1040315},
  ids		= {dpor}
}

@book{aroraBarak,
place={Cambridge},
title={Computational Complexity: A Modern Approach},
publisher={Cambridge University Press},
author={Arora, Sanjeev and Barak, Boaz},
year={2009},
doi ={10.1017/CBO9780511804090},
ids= {aroraBarak}
}

@article{knuth75,
 author = {Donald E. Knuth},
 journal = {Mathematics of Computation},
 number = {129},
 pages = {121--136},
 publisher = {American Mathematical Society},
 title = {Estimating the Efficiency of Backtrack Programs},
 urldate = {2025-11-03},
 volume = {29},
 year = {1975},
 doi = {10.2307/2005469},
ids = {knuth75}
}

@article{pitt,
title = {A note on extending Knuth's tree estimator to directed acyclic graphs},
journal = {Information Processing Letters},
volume = {24},
number = {3},
pages = {203-206},
year = {1987},
issn = {0020-0190},
doi = {10.1016/0020-0190(87)90187-6},
author = {Leonard Pitt},
ids = {pitt}
}

@InProceedings{gator,
author="Kokologiannakis, Michalis
and Majumdar, Rupak
and Vafeiadis, Viktor",
title="Enhancing GenMC's Usability and Performance",
booktitle="Tools and Algorithms for the Construction and Analysis of Systems (TACAS)",
year="2024",
publisher="Springer Nature Switzerland",
address="Cham",
pages="66--84",
doi = "10.1007/978-3-031-57249-4_4",
ids = "gator"
}

@article{Stockmeyer85,
  author       = {Larry J. Stockmeyer},
  title        = {On Approximation Algorithms for {\#}P},
  journal      = {{SIAM} J. Comput.},
  volume       = {14},
  number       = {4},
  pages        = {849--861},
  year         = {1985},
  doi          = {10.1137/0214060},
  ids          = {Stockmeyer85}
}

@misc{jmc,
  url          = {https://jmc.mpi-sws.org/},
  author         = {JMC: Java Model Checker},
 year = {2026},
ids = {jmc}
}

@InProceedings{	  optimal-dpor,
  title		= {Optimal dynamic partial order reduction},
  author	= {Parosh Aziz Abdulla and Stavros Aronis and Bengt Jonsson
		  and Konstantinos Sagonas},
  booktitle = {Proceedings of the 41st ACM SIGPLAN-SIGACT Symposium on Principles of Programming Languages},
  series = {POPL '14},
  year		= {2014},
  pages		= {373--384},
  publisher	= {ACM},
  address	= {New York, NY, USA},
  url		= {http://doi.acm.org/10.1145/2535838.2535845},
  doi		= {10.1145/2535838.2535845},
  ids		= {optimal-dpor}
}

@Misc{		  svcomp,
  key		= {svcomp},
  title		= {Competition on Software Verification ({SV-COMP})},
  author	= {SV-COMP},
  year		= {2019},
  url		= {https://sv-comp.sosy-lab.org/2019/},
  urldate	= {2019-03-27},
  ids		= {svcomp}
}

@Book{		  book:herlihy-shavit,
  title		= {The Art of Multiprocessor Programming},
  edition = {2nd},
  author	= {Herlihy, Maurice and Shavit, Nir and Luchangco, Victor and Spear, Michael},
  year		= {2020},
  publisher	= {Morgan Kaufmann Publishers Inc.},
  doi         = {10.1016/C2011-0-06993-4},
  ids   = {book:herlihy-shavit} 
}

@InProceedings{	  timestamped-stack,
  title		= {A Scalable, Correct Time-Stamped Stack},
  author	= {Dodds, Mike and Haas, Andreas and Kirsch, Christoph M.},
  booktitle	= {42nd ACM SIGPLAN-SIGACT Symposium on Principles of Programming Languages},
  year		= {2015},
  publisher	= {ACM},
  address	= {New York, NY, USA},
  doi		= {10.1145/2676726.2676963},
  pages		= {233–246},
  numpages	= {14},
  location	= {Mumbai, India},
  series	= {POPL '15},
  id = {timestamped-stack}
}

@article{agmstack,
  author       = {Afek, Yehuda and
                  Gafni, Eli and
                  Morrison, Adam},
  title        = {Common2 extended to stacks and unbounded concurrency},
  journal      = {Distributed Comput.},
  volume       = {20},
  year         = {2007},
  doi          = {10.1007/S00446-007-0023-3},
id ={agmstack}
}

@inproceedings{bron05,
author = {Bron, Arkady and Farchi, Eitan and Magid, Yonit and Nir, Yarden and Ur, Shmuel},
title = {Applications of synchronization coverage},
year = {2005},
isbn = {1595930809},
publisher = {Association for Computing Machinery},
address = {New York, NY, USA},
doi = {10.1145/1065944.1065972},
booktitle = {Proceedings of the Tenth ACM SIGPLAN Symposium on Principles and Practice of Parallel Programming},
location = {Chicago, IL, USA},
series = {PPoPP '05},
ids = {bron05}
}

@article{vaisman2017,
  title={Stochastic enumeration method for counting trees},
  author={Vaisman, Radislav and Kroese, Dirk P},
  journal={Methodology and Computing in Applied Probability},
  volume={19},
  number={1},
  pages={31--73},
  year={2017},
  publisher={Springer},
  doi={10.1007/s11009-015-9457-4},
  ids = {vaisman2017}
}

@article{rubinstein,
  title={Stochastic enumeration method for counting NP-hard problems},
  author={Rubinstein, Reuven},
  journal={Methodology and Computing in Applied Probability},
  volume={15},
  number={2},
  pages={249--291},
  year={2013},
  publisher={Springer},
  doi={10.1007/s11009-011-9242-y},
  ids = {rubinstein}
}

@article{finelist,
  author       = {Rudolf Bayer and
                  Mario Schkolnick},
  title        = {Concurrency of Operations on B-Trees},
  journal      = {Acta Informatica},
  volume       = {9},
  pages        = {1--21},
  year         = {1977},
  doi          = {10.1007/BF00263762},
  ids = {finelist}
}

@article{10.1145/3422648.3422661,
author = {Arenas, Marcelo and Croquevielle, Luis Alberto and Jayaram, Rajesh and Riveros, Cristian},
title = {Efficient Logspace Classes for Enumeration, Counting, and Uniform Generation},
year = {2020},
issue_date = {March 2020},
publisher = {Association for Computing Machinery},
address = {New York, NY, USA},
volume = {49},
number = {1},
issn = {0163-5808},
url = {https://doi.org/10.1145/3422648.3422661},
doi = {10.1145/3422648.3422661},
journal = {SIGMOD Rec.},
month = sep,
pages = {52–59},
numpages = {8}
}

@article{meel2025practicalfprasnfaexploiting,
  title={Towards practical FPRAS for \#NFA: Exploiting the Power of Dependence},
  author={Meel, Kuldeep S and de Colnet, Alexis},
  journal={Proceedings of the ACM on Management of Data},
  volume={3},
  number={2},
  pages={1--23},
  year={2025},
  doi = {10.1145/3725253},
  publisher={ACM New York, NY, USA}
}

@article{meel2024fasterfprasnfa,
  title={A faster FPRAS for \#NFA},
  author={Meel, Kuldeep S and Chakraborty, Sourav and Mathur, Umang},
  journal={Proceedings of the ACM on Management of Data},
  volume={2},
  number={2},
  pages={1--22},
  year={2024},
  publisher={ACM New York, NY, USA},
  doi={10.1145/3651613}
}

@inproceedings{chakraborty2013scalableapproximatemodelcounter,
  title={A Scalable Approximate Model Counter},
  author={Chakraborty, Supratik and Meel, Kuldeep S and Vardi, Moshe Y},
  booktitle={International Conference on Principles and Practice of Constraint Programming},
  pages={200--216},
  year={2013},
  organization={Springer},
  doi={10.1007/978-3-642-40627-0_18}
}

@InProceedings{10.1007/978-3-540-39724-3_11,
author="Chockler, Hana
and Kupferman, Orna
and Vardi, Moshe Y.",
editor="Geist, Daniel
and Tronci, Enrico",
title="Coverage Metrics for Formal Verification",
booktitle="Correct Hardware Design and Verification Methods",
year="2003",
publisher="Springer Berlin Heidelberg",
address="Berlin, Heidelberg",
doi ="10.1007/978-3-540-39724-3_11",
pages="111--125",
}

@INPROCEEDINGS{5431751,
  author={Taleghani, Ali and Atlee, Joanne M.},
  booktitle={2009 IEEE/ACM International Conference on Automated Software Engineering}, 
  title={State-Space Coverage Estimation}, 
  year={2009},
  volume={},
  number={},
  pages={459-467},
  doi={10.1109/ASE.2009.24}}

@article{10.1145/3210309,
author = {B\"{o}hme, Marcel},
title = {STADS: Software Testing as Species Discovery},
year = {2018},
issue_date = {April 2018},
publisher = {Association for Computing Machinery},
address = {New York, NY, USA},
volume = {27},
number = {2},
issn = {1049-331X},
url = {https://doi.org/10.1145/3210309},
doi = {10.1145/3210309},

journal = {ACM Trans. Softw. Eng. Methodol.},
month = jun,
articleno = {7},
numpages = {52}
}

@article{10.1145/3611019,
author = {B\"{o}hme, Marcel and Man\`{e}s, Valentin J. M. and Cha, Sang Kil},
title = {Boosting Fuzzer Efficiency: An Information Theoretic Perspective},
year = {2023},
issue_date = {November 2023},
publisher = {Association for Computing Machinery},
address = {New York, NY, USA},
volume = {66},
number = {11},
issn = {0001-0782},
url = {https://doi.org/10.1145/3611019},
doi = {10.1145/3611019},
journal = {Commun. ACM},
month = oct,
pages = {89–97},
numpages = {9}
}

@inproceedings{10.1145/3597503.3639198,
author = {Liyanage, Danushka and Lee, Seongmin and Tantithamthavorn, Chakkrit and B\"{o}hme, Marcel},
title = {Extrapolating Coverage Rate in Greybox Fuzzing},
year = {2024},
isbn = {9798400702174},
publisher = {Association for Computing Machinery},
address = {New York, NY, USA},
url = {https://doi.org/10.1145/3597503.3639198},
doi = {10.1145/3597503.3639198},
booktitle = {Proceedings of the IEEE/ACM 46th International Conference on Software Engineering},
articleno = {132},
numpages = {12},
location = {Lisbon, Portugal},
series = {ICSE '24}
}

@inproceedings{10.1145/3597926.3598043,
author = {Lipp, Stephan and Elsner, Daniel and Kacianka, Severin and Pretschner, Alexander and B\"{o}hme, Marcel and Banescu, Sebastian},
title = {Green Fuzzing: A Saturation-Based Stopping Criterion using Vulnerability Prediction},
year = {2023},
isbn = {9798400702211},
publisher = {Association for Computing Machinery},
address = {New York, NY, USA},
url = {https://doi.org/10.1145/3597926.3598043},
doi = {10.1145/3597926.3598043},
booktitle = {Proceedings of the 32nd ACM SIGSOFT International Symposium on Software Testing and Analysis},
pages = {127–139},
numpages = {13},
location = {Seattle, WA, USA},
series = {ISSTA 2023}
}

@inproceedings{pop-dpor,
  title={Parsimonious optimal dynamic partial order reduction},
  author={Abdulla, Parosh Aziz and Atig, Mohamed Faouzi and Das, Sarbojit and Jonsson, Bengt and Sagonas, Konstantinos},
  booktitle={International Conference on Computer Aided Verification},
  pages={19--43},
  year={2024},
  organization={Springer},
series ={CAV'24},
doi={10.1007/978-3-031-65630-9_2},
ids = {pop-dpor}
}

@book{spin,
  title={The SPIN model checker: Primer and reference manual},
  author={Holzmann, Gerard J},
  volume={1003},
  year={2004},
  publisher={Addison-Wesley Reading},
    ids={spin}
}

@Article{	  jpf,
  title		= {Model Checking {JAVA} Programs using {JAVA} PathFinder},
  author	= {Klaus Havelund and Thomas Pressburger},
  journal	= {Int. J. Soft. Tool. Tech. Transf.},
  year		= {2000},
  volume	= {2},
  number	= {4},
  pages		= {366--381},
  doi		= {10.1007/S100090050043},
  ids		= {jpf,pathfinder}
}

@article{sharedMem,
  title={Testing shared memories},
  author={Gibbons, Phillip B and Korach, Ephraim},
  journal={SIAM Journal on Computing},
  volume={26},
  number={4},
  pages={1208--1244},
  year={1997},
  doi={10.1137/S0097539794279614},
  publisher={SIAM}
}

@inproceedings{nidhugg,
  title={Stateless Model Checking for TSO and PSO},
  author={Abdulla, Parosh Aziz and Aronis, Stavros and Atig, Mohamed Faouzi and Jonsson, Bengt and Leonardsson, Carl and Sagonas, Konstantinos},
  booktitle={Proceedings of the 21st International Conference on Tools and Algorithms for the Construction and Analysis of Systems-Volume 9035},
  series={TACAS'15},
  pages={353--367},
  year={2015},
  doi={10.1007/978-3-662-46681-0_28},
  ids={nidhugg}
}

@inproceedings{10.1145/3468264.3468570,
author = {B\"{o}hme, Marcel and Liyanage, Danushka and W\"{u}stholz, Valentin},
title = {Estimating residual risk in greybox fuzzing},
year = {2021},
isbn = {9781450385626},
publisher = {Association for Computing Machinery},
address = {New York, NY, USA},
url = {https://doi.org/10.1145/3468264.3468570},
doi = {10.1145/3468264.3468570},
booktitle = {Proceedings of the 29th ACM Joint Meeting on European Software Engineering Conference and Symposium on the Foundations of Software Engineering},
pages = {230–241},
numpages = {12},
location = {Athens, Greece},
series = {ESEC/FSE 2021}
}

@software{zenodo_artifact,
  author       = {Balasubramanian, A. R. and Khoshechin Jorshari, Mohammad Hossein and Majumdar, Rupak and Mathur, Umang and Zhang, Minjian},
  title        = {State Space Estimation for DPOR-based Model Checkers (Artifact)},
  year         = {2026},
  publisher    = {Zenodo},
  doi          = {10.5281/zenodo.19625358}
}

@article{penix,
  title={Coverage Metrics for Model Checking},
  author={Penix, John and Visser, Willem and Norvig, Peter},
  url={https://ntrs.nasa.gov/citations/20010106055},
  year={2001}
}

@book{GodefroidThesis,
  author       = {Patrice Godefroid},
  title        = {Partial-Order Methods for the Verification of Concurrent Systems -
                  An Approach to the State-Explosion Problem},
  series       = {Lecture Notes in Computer Science},
  volume       = {1032},
  publisher    = {Springer},
  year         = {1996},
  doi          = {10.1007/3-540-60761-7},
  isbn         = {3-540-60761-7},
}

@article{WangCGMK18,
  author       = {Kaiyuan Wang and
                  Hayes Converse and
                  Milos Gligoric and
                  Sasa Misailovic and
                  Sarfraz Khurshid},
  title        = {A Progress Bar for the {JPF} Search Using Program Executions},
  journal      = {{ACM} {SIGSOFT} Softw. Eng. Notes},
  volume       = {43},
  number       = {4},
  pages        = {55},
  year         = {2018},
  url          = {https://doi.org/10.1145/3282517.3282525},
  doi          = {10.1145/3282517.3282525},
  bibsource    = {dblp computer science bibliography, https://dblp.org}
}

@article{readsFrom2019,
author = {Abdulla, Parosh Aziz and Atig, Mohamed Faouzi and Jonsson, Bengt and L\r{a}ng, Magnus and Ngo, Tuan Phong and Sagonas, Konstantinos},
title = {Optimal stateless model checking for reads-from equivalence under sequential consistency},
year = {2019},
issue_date = {October 2019},
publisher = {Association for Computing Machinery},
address = {New York, NY, USA},
volume = {3},
number = {OOPSLA},
url = {https://doi.org/10.1145/3360576},
doi = {10.1145/3360576},
journal = {Proc. ACM Program. Lang.},
month = oct,
articleno = {150},
numpages = {29}
}

@InProceedings{readsValue2021,
author="Agarwal, Pratyush
and Chatterjee, Krishnendu
and Pathak, Shreya
and Pavlogiannis, Andreas
and Toman, Viktor",
editor="Silva, Alexandra
and Leino, K. Rustan M.",
title="Stateless Model Checking Under a Reads-Value-From Equivalence",
booktitle="Computer Aided Verification",
year="2021",
publisher="Springer International Publishing",
address="Cham",
pages="341--366",
isbn="978-3-030-81685-8"
}

@article{FarzanMathur2024,
author = {Farzan, Azadeh and Mathur, Umang},
title = {Coarser Equivalences for Causal Concurrency},
year = {2024},
issue_date = {January 2024},
publisher = {Association for Computing Machinery},
address = {New York, NY, USA},
volume = {8},
number = {POPL},
url = {https://doi.org/10.1145/3632873},
doi = {10.1145/3632873},
journal = {Proc. ACM Program. Lang.},
month = jan,
articleno = {31},
numpages = {31}
}

\newpage
\appendix
\section{Complexity of Counting} \label{theory}

In this section, we prove our theoretical hardness results on the (Approximate) Counting Problem. We begin by proving the hardness of the Counting Problem, i.e., Theorem~\ref{thm:counting-problem-hardness}.

\subsection{Proof of Theorem~\ref{thm:counting-problem-hardness}} \label{theory:A}

In order to prove that the Counting Problem is $\sharpP$-complete, we first recall
the definition of the complexity class $\sharpP$. Suppose we are given some relation $R(x,y)$ over some set. For each $x$, we can associate the number $\#R(x)$ as the number
of possible $y$ such that $R(x,y)$ is true. Intuitively, if we think of $x$ as some problem instance and $y$ as a possible solution to the problem instance $x$, then 
$\#R(x)$ simply counts the number of possible solutions for $x$.

\begin{example}
    Let CNF be the relation such that CNF$(x,y)$ is true iff $x$ denotes a formula in CNF
    and $y$ denotes a satisfying assignment for $x$. Then $\#\text{CNF}(x)$ is simply the number
    of satisfying assignments of $x$.
\end{example}

This definition of $\#R(x)$ now leads us to the definition of $\sharpP$.
For any relation $R$, we say that the function $\#R$ belongs to $\sharpP$ if there is a non-deterministic polynomial-time Turing machine that, for any input $x$, has exactly $\#R(x)$ many accepting paths.

\begin{example}
    Note that $\#\text{CNF} \in \sharpP$: We can have a non-deterministic Turing machine 
    that on input $x$, simply guesses an assignment $y$ of $x$ and accepts iff $y$ satisfies $x$. This example already illustrates the power of $\sharpP$ : If we can efficiently count $\#\text{CNF}$, then we can easily decide if a formula $x$ is satisfiable
    by simply checking if $\#\text{CNF}(x) > 0$. 
\end{example}

Now, similar to the notion of NP-hardness, we also have a notion of $\sharpP$-hardness.
For a relation $R$, the function $\#R$ is said to be $\sharpP$-hard if every 
function in $\sharpP$ can be reduced to $\#R$ by means of polynomial-time Turing reductions. $\#R$ is then said to be $\sharpP$-complete if it is both in $\sharpP$
and $\sharpP$-hard.

We will not formally define Turing reductions here, but for the purposes of the paper, only the following two notions are important. 
\begin{itemize}
    \item The function $\#\text{CNF}$ is $\sharpP$-complete.
\item Suppose we have two functions $\#R_1$ and $\#R_2$ such that
$\#R_1$ is $\sharpP$-hard. Further, suppose there exist two
polynomial-time computable functions $f$ and $g$ such that, for every
input $x$,
\[
    \#R_1(x)=\#R_2(f(x))-\#R_2(g(x)).
\]
Then $\#R_2$ is $\sharpP$-hard as well, since this equality gives a
two-query non-adaptive polynomial-time Turing reduction from $\#R_1$
to $\#R_2$.
\end{itemize}

Note that proving that counting a function $\#R$ is $\sharpP$-complete is evidence of intractability. Indeed, if we can efficiently count $\#R$ for some $\sharpP$-complete
function $\#R$, then we can also efficiently count $\#\text{CNF}$, and by the argument above,
we can also efficiently decide satisfiability of propositional formulas, which is NP-hard.

As our first result, we prove the $\sharpP$-completeness of the Counting Problem
that we consider in this paper. To this end, let us consider the relation $\mathcal{R}(P,G)$, which is true iff $P$ is a (bounded) program and $G$ is a reachable maximal execution graph of $P$. Then $\#\mathcal{R}$ is precisely the Counting Problem that
we defined in Problem {prob1}. We now show that

\begin{theorem}
    The Counting Problem $\#\mathcal{R}$ is $\sharpP$-complete. Moreover, the hardness holds even when the number of threads of the given program is 2.
\end{theorem}

\begin{proof}
    First, we will prove $\#\mathcal{R}$ is in $\sharpP$. Indeed, it is easy to see that there is a non-deterministic Turing machine, that on input a program $P$, simply simulates $P$ with some schedule among the threads and records the execution trace. 
    %Any such schedule will be of length at most $B \times T$ where $B$ is the bound of $P$ and $T$ is the number of threads of $P$. 
    Once this is done, it checks if this trace is canonical. 
    Since this can be achieved in polynomial time, we can conclude
    that the Counting Problem is in $\sharpP$.

    To prove $\sharpP$-hardness, we give a two-query non-adaptive
polynomial-time Turing reduction from $\#$CNF.
    To this end, let $\varphi$ be a formula in CNF over $n$ variables. We will construct
    two bounded concurrent programs $P_\varphi$ and $P'$ such that $\#$CNF$(\varphi) = \#\mathcal{R}(P') - \#\mathcal{R}(P_\varphi)$. This will then establish the required
    Turing reduction.

    The intuitive idea behind $P_\varphi$ and $P'$ will be the following.
    We will construct $P_\varphi$ in such a way so that  every satisfying
    assignment of $\varphi$ will correspond to exactly one complete execution graph,
    every unsatisfying assignment of $\varphi$ will correspond to exactly two
    complete execution graphs and then there will be some remaining complete execution
    graphs called garbage execution graphs. Then, we will construct $P'$ so that
    every assignment of $\varphi$ will correspond to exactly two complete execution graphs
    and then there will be some garbage execution graphs which will be of the same magnitude as the number of garbage execution graphs of $P_\varphi$. It would then follow that $\#$CNF$(\varphi) = \#\mathcal{R}(P') - \#\mathcal{R}(P_\varphi)$,
    thereby proving the claim.

    First, we show how to construct $P_\varphi$. $P_\varphi$ will have 2 threads $t_1, t_2$
    and $n+2$ locations $x_1,x_2,r_1,\dots,r_n$. The set of values of the locations will be $\{T,F,\$\}$. Initially, we assume that all the locations have the value $F$.

    The thread $t_1$ will work as follows: It will have $n$ iterations
    and in each iteration it will do three writes to $x_1$ - First it
    will write $T$, then $F$ and finally $\$$.
    After the $n$ iterations, it will terminate by writing $T$ into $x_2$. 
    
    The thread $t_2$ will work as follows: It will have $n$ iterations and in each $i^{th}$ iteration it will do a read from $x_1$ and it will store this value into $r_i$. If this value is $\$$, it will read from $x_2$ and terminate.
    Otherwise, it will read from $x_1$ again. If this new value is not $\$$, it will
    read from $x_2$ and terminate. Otherwise, it will continue with the $(i+1)^{th}$ iteration. 
    Finally, after all the $n$ iterations are over, it would have stored either $T$ or $F$
    in each $r_i$. It will then check if the formula $\varphi$ is satisfiable
    with the assignment where the $i^{th}$ variable of $\varphi$ is assigned the value of $r_i$. If so,
    it will terminate. Otherwise, it will read from $x_2$ and terminate.

    To fully understand how $P_\varphi$ conforms to the intuition mentioned before,
    we first define \emph{lockstep execution graphs}. A reachable maximal execution graph of $P_\varphi$ is a lockstep execution graph if, whenever $t_2$ reads for the first time from $x_1$ in its $i^{th}$ iteration, the corresponding write was $t_1$ writing $T$ or $F$ into $x_1$
    in its $i^{th}$ iteration. Furthermore, whenever $t_2$ reads for the second time from $x_1$ in its $i^{th}$ iteration, the corresponding write was $t_2$ writing $\$$ into $x_1$ in its $i^{th}$ iteration.
    These are precisely the reachable maximal execution graphs of $P_\varphi$ which do not allow $t_2$ to terminate within any of its $n$ iterations. 
    All the other type of reachable maximal execution graphs correspond precisely to those executions in which $t_2$ will terminate within any one of its $n$ iterations.

    Based on the definition of lockstep execution graphs, it is now easy to
    show how to relate assignments of $\varphi$ to lockstep execution graphs
    of $P_\varphi$.
    \begin{itemize}
        \item For every satisfying assignment $A$ of $\varphi$, we can construct 
        a lockstep execution graph as follows: If $A$ assigns the $i^{th}$ variable of $\varphi$ to $T$, then we let $t_2$ read for the first time from $x_1$ in its $i^{th}$ iteration after $t_1$ has written $T$ in its $i^{th}$ iteration; otherwise we let $t_2$ read for the first time from $x_1$ in its $i^{th}$ iteration
        after $t_1$ has written $F$ in its $i^{th}$ iteration. Note that, in this way, after the $n$ iterations of $t_2$, $t_2$ will check that the assignment stored in $r_1,\dots,r_n$ satisfies the formula $\varphi$ and then terminate without reading from $x_2$. Hence, it does not matter when we let $t_1$ write into $x_2$, since that will never be read by $t_2$. 

        \item For every unsatisfying assignment $A$ of $\varphi$, we can construct
        two lockstep execution graphs as follows: If $A$ assigns the $i^{th}$ variable of $\varphi$ to $T$, then we let $t_2$ read for the first time from $x_1$ in its $i^{th}$ iteration after $t_1$ has written $T$ in its $i^{th}$ iteration; otherwise we let $t_2$ read for the first time from $x_1$ in its $i^{th}$ iteration
        after $t_1$ has written $F$ in its $i^{th}$ iteration. Note that, in this way, after the $n$ iterations of $t_2$, $t_2$ will check that the assignment stored in $r_1,\dots,r_n$ does not satisfy the formula $\varphi$.
        Hence, it will terminate only after reading from $x_2$. Depending on whether
        we allow $t_1$ to write into $x_2$ before $t_2$ reads from $x_2$ or not, we will 
        get two lockstep execution graphs.

        \item Finally, every other way of scheduling $t_1$ and $t_2$ will result
        in a maximal execution graph that is not a lockstep execution graph. Such graphs
        will be called garbage execution graphs and will correspond to exactly
        those maximal execution graphs in which $t_2$ terminates during one of its $n$ iterations.
    \end{itemize}

    This completes the description behind the behavior of $P_\varphi$. Now, we show how to construct $P'$. $P'$ will work in exactly the same way as $P_\varphi$, except that
    it will always read from $x_2$ before terminating, i.e., after it checks
    whether the values in $r_1,\dots,r_n$ satisfy the formula $\varphi$, it will
    read from $x_2$ and then terminate. Based on exactly the same argument as before,
    we can show that $P'$ has exactly two lockstep execution graphs corresponding to 
    every assignment of $\varphi$ and every other maximal execution graph
    is a garbage execution graph that is present in $P_\varphi$ as well.
    Hence, the number of satisfying assignments of $\varphi$ is precisely $\#\mathcal{R}(P') - \#\mathcal{R}(P_\varphi)$, which completes the proof.
\end{proof}

We now present a different construction which shows that hardness holds when the number of locations is a constant, but the number of threads is unbounded. 

\begin{theorem}
    The Counting Problem is $\#P$-hard, even when the number of locations is a constant.
\end{theorem}

\begin{proof}
    We once again construct a Turing reduction from $\#$CNF. Our proof is inspired by the construction in 
    \cite{sharedMem}, which shows that testing sequential consistency is $\mathsf{NP}$-hard for execution traces with only a constant number of locations. However, our reduction introduces several nontrivial modifications.

    Let $\varphi$ be a CNF formula with $n$ variables
    and $m$ clauses. Similar to the previous proof, we will construct
    two programs $P_\varphi$ and $P'$ such that  $\#$CNF$(\varphi) = \#\mathcal{R}(P') - \#\mathcal{R}(P_\varphi)$. The intuition behind these programs will
    be the same as in the previous proof.

    First, we describe the program $P_\varphi$. $P_\varphi$ will have $m+4$ threads
    denoted by $\{T_1,T_2,T_3,T_4,C_1,C_2,\dots,C_m\}$ and four locations $d,e,f,g$. 
    Furthermore, the set of values will be $\{T_j : 1 \le j \le n+m\} \cup \{F_j : 1 \le j \le n+m\} \cup  \{\$\}$. All of the locations will initially have the value $\$$.
    
    Intuitively, the threads $T_1$ and $T_2$ will be responsible for choosing either $T$ or $F$ for each variable of $\varphi$. The thread $T_3$ will act as a tiebreaker between $T_1$ 
    and $T_2$, by choosing one of their values for each variable of $\varphi$.
    Each thread $C_i$ will be responsible for checking that the assignment
    chosen by $T_3$ satisfies the $i^{th}$ clause of $\varphi$.
    Finally thread $T_4$ collects all the votes from all of the $C_i$ threads and 
    ensures that all the clauses are satisfied.

    More precisely, the thread $T_1$ will behave as follows: It will have $n$ iterations
    and in each $i^{th}$ iteration it will first write $T_i$ to the location $d$. Then, it will read from $d$.
    If this value is not $F_i$, then it will read from $g$ and terminate.
    Otherwise, it will continue with the $(i+1)^{th}$ iteration.

    The thread $T_2$ will behave as follows: It will also have $n$ iterations
    and in each $i^{th}$ iteration it will first read from the location $d$.
    If this value is not $T_i$, it will read from $g$ and terminate.
    Otherwise, it will write $F_i$ to the location $d$ and continue with the $(i+1)^{th}$ iteration.

    The thread $T_3$ will behave as follows: It will also have $n$ iterations
    and in each $i^{th}$ iteration it will first read from the location $d$.
    If this value is neither $T_i$ or $F_i$ it will read from $g$ and terminate.
    Otherwise, it will write the value that it read into the location $e$
    and continue with the $(i+1)^{th}$ iteration.
    At the end of $n$ iterations, it will write $T_1$ onto the location $g$.

    Each thread $C_j$ will behave as follows: It will have a single Boolean local variable $flag$ which is initially set to false. Furthermore, it will have $n$ iterations
    and in each $i^{th}$ iteration it will first read from location $e$.
    If this value is neither $T_i$ or $F_i$, it will read from $g$ and terminate.
    Otherwise, 
    \begin{itemize}
        \item If the $i^{th}$ variable of $\varphi$ appears in the $j^{th}$ clause positively (resp. negatively) and it read $T_i$ (resp. $F_i$), then it will 
        set its $flag$ variable to true.
        \item If the $i^{th}$ variable of $\varphi$ appears in the $j^{th}$ clause positively (resp. negatively) and it read $F_i$ (resp. $T_i$), then it will do nothing.
    \end{itemize}
    At the end of the $n^{th}$ iteration, $C_j$ will read check if its $flag$ variable is true.
    If so, then it will write $T_{n+j}$ onto the location $f$.
    Otherwise, it will write $F_{n+j}$ onto the location $f$.

    Finally, the thread $T_4$ will behave as follows: It will also have a single Boolean local variable $flag$ which is initially set to true. Furthermore, it will have $m$ iterations
    and in each $i^{th}$ iteration it will read from the location $f$.
    If this value is not $T_{n+i}$, then it will set $flag$ to false; otherwise, it will not do anything. It will then continue
    with the $(i+1)^{th}$ iteration.
    Finally, at the end of the $m^{th}$ iteration, it will check if $flag$ is true. 
    If it is, it will terminate; otherwise it will read from $g$ and terminate.

    Similar to the previous reduction, we define \emph{lockstep execution graphs}
    as reachable maximal execution graphs obtained where each thread finishes only after executing
    all of its iterations. Again similar to the previous reduction, we can 
    map one lockstep execution graph for every satisfying assignment
    and two lockstep execution graphs for every unsatisfying assignment
    (and all the other remaining execution graphs will not be lockstep execution graphs).
    Now, if we define $P'$ to be the same as $P_\varphi$, except thread $T_4$ always
    reads from $g$ and terminates, then similar to the previous reduction
    we have $\#$CNF$(\varphi) = \#\mathcal{R}(P') - \#\mathcal{R}(P_\varphi)$.
    This completes the reduction.
\end{proof}

% Hence, we have shown that, even if any two among the number of threads, number of locations and number of local registers is a constant, then the Counting Problem
% remains $\sharpP$-hard. Note that when all of them are a constant, then the Counting Problem is in P. Indeed, in this case, simulating an execution of the given concurrent program $P$ (under sequential consistency) upto a given bound $B$, can be easily done in non-deterministic logarithmic space (in the size of $P$ and $B$): Indeed, to store an execution graph, we just need to remember
% which line of $P$ each thread is executing, what the contents of each of the local registers of each thread are and what the contents of each of the shared locations are.
% All of these contents can be stored in logarithmic space. This allows us to simulate the canonical executions of a $P$ in non-deterministic logarithmic space. Since non-deterministic logarithmic space machines can be simulated in polynomial time,
% the result follows.

Having proved our hardness results for the Counting Problem, we now move to the Approximate Counting problem.

\subsection{Proof of Theorem~\ref{thm:approximate-counting-problem-hardness}} \label{theory:B}

Suppose there exists a polynomial-time randomized algorithm and some $\epsilon > 0$
and $\rho > 1/2$
with the following property: For every program $P$ of size $n$,
the algorithm returns a number $N$ such that $N/2^{n^{1-\epsilon}} \le C(P) \le N \times 2^{n^{1-\epsilon}}$ with probability at least $\rho$. Pick $k$ such that $k\epsilon > 1$. By definition of $k\epsilon$, it follows that for any constant $D$ and for all sufficiently large $n$,
$n^k > D(n^{k(1-\epsilon)+1})$. Using this, we now show that SAT can be decided in randomized polynomial time, which would prove that RP = NP. 

Let $\varphi$ be a Boolean formula in CNF over $n$ variables and $m$ clauses.
We construct a program $P_\varphi$ with threads $t_0,t_1,\dots,t_{n+n^k}$ and locations $x_1,\dots,x_{n+n^k}, r_1,\dots,r_n$. 
The values of the locations
will be $\{T,F\}$, with $F$ being the default value. 

Each thread $t_i$ for $1 \le i \le n+n^k$, simply writes $T$ into the location $x_i$
and terminates. The thread $t_0$ has $n$ iterations and in each $i^{th}$ iteration
it will copy the value in the location $x_i$ to $r_i$. Note that
after $n$ iterations, $t_0$ will store an assignment of the formula $\varphi$
in $r_1,\dots,r_n$. It will then check if this assignment satisfies $\varphi$.
If the assignment satisfies $\varphi$, it will do $n^k$ more steps, where in each $i^{th}$ iteration it will read from the location $x_{n+i}$ and at the end of all of 
these iterations, it will terminate.
If the assignment does not satisfy $\varphi$, it will simply terminate at that point.
Note that this program can be easily encoded in a way so that its size is $O(n^k)$.

For each assignment $A$ of $\varphi$, we can construct reachable maximal execution graph(s) of $P_\varphi$ as follows: If $A$ assigns the $i^{th}$ variable to true, then we let $t_i$
execute before the $i^{th}$ iteration of $t_0$; otherwise, we let $t_i$ execute after
the $i^{th}$ iteration of $t_0$. This ensures that when $t_0$ finishes its $n$ iterations,
the assignment stored in $r_1,\dots,r_n$ is the same as $A$. Then, we let
$t_0$ check whether $A$ satisfies $\varphi$. If it does not then $t_0$ immediately terminates and we schedule all of the remaining threads $t_{n+1},\dots,t_{n+n^k}$
to get a maximal execution graph. Suppose $A$ satisfies $\varphi$. Then, we can generate $2^{n^k}$ different maximal execution graphs by scheduling $t_0,t_{n+1},\dots,t_{n+n^k}$ in all possible ways.
Hence, for each unsatisfying assignment, we get one reachable maximal execution graph
and for each satisfying assignment, we get $2^{n^k}$ reachable maximal execution graphs.
From the construction of $P_\varphi$, it can be easily inferred that these are the only possible reachable maximal execution graphs of $P_\varphi$.

It follows then that if $\varphi$ is unsatisfiable, then $C(P_\varphi) = 2^n$
and if $\varphi$ is satisfiable, then $C(P_\varphi) \ge 2^n - 1 + 2^{n^k}$.
Hence, if $\varphi$ is unsatisfiable, then our assumed algorithm, on input $P_\varphi$, returns a number $N$ such that $N/2^{|P_\varphi|^{1-\epsilon}} \le 2^n$ and so $N \le 2^n \times 2^{|P_\varphi|^{1-\epsilon}} \le
2^{O(n^{k(1-\epsilon)}) + 1}$. On the other hand, if $\varphi$ is satisfiable,
then our assumed algorithm, on input $P_\varphi$,
returns a number $N$ such that $2^{n^k} \le N \times 2^{|P_\varphi|^{1-\epsilon}}$
and so $2^{n^k}/2^{|P_\varphi|^{1-\epsilon}} \le N$ which implies that
$2^{n^k-O(n^{k(1-\epsilon)})} \le N$.
If we can show that $2^{O(n^{k(1-\epsilon)}) + 1} \le 2^{n^k-O(n^{k(1-\epsilon)})}$ (for all sufficiently large $n$), then we can clearly decide whether $\varphi$ is satisfiable
with probability $\rho$, by getting the number $N$ from our assumed algorithm and checking whether it is lower than the first quantity
or higher than the second quantity. Hence, the only thing left to show is this inequality.

To this end, it suffices to show that $O(n^{k(1-\epsilon)}) + 1 \le n^k-O(n^{k(1-\epsilon)})$. For showing this, it suffices to show that $O(n^{k(1-\epsilon)+1}) \le n^k - O(n^{k(1-\epsilon)+1})$ which is the same as showing
$2O(n^{k(1-\epsilon)+1}) \le n^k$. Since $k\epsilon > 1$,
it follows that $k(1-\epsilon) + 1 < k$. Hence, for all sufficiently large $n$,
$n^k$ will eventually dominate $2O(n^{k(1-\epsilon)+1})$, which completes the proof.

Hence, with probability $\rho$, we can correctly decide whether $\varphi$ is satisfiable or not. By using the standard machinery of self-reducibility~\cite{aroraBarak}, it follows
that we can convert such an algorithm into one that always rejects when $\varphi$ is unsatisfiable and accepts with probability $> 1/2$ when $\varphi$ is satisfiable.
It then follows that RP = NP. 

Furthermore, the same proof implies that, if a deterministic algorithm exists
which would always output an estimate $N$ such that $N/2^{n^{1-\epsilon}} \le C(P) \le N \times 2^{n^{1-\epsilon}}$, then P = NP. (This is because in that case $\rho$ would be 1 and so we would be able to decide with probability 1, i.e., deterministically, whether
$\varphi$ is satisfiable).

\section{Proof of Proposition~\ref{prop1}}\label{theory:C}

First, we will show that any predecessor of $G$ defines a unique 
sequentially maximal event of $G$. Indeed, for any predecessor $H$ of $G$,
let $e_H$ be the event added from $H$ to reach $G$.
It is immediate to see that $e_H \neq e_H'$ if $H \neq H'$.
Now, let $e = e_H$ for some such predecessor $H$.
Inspecting the definition of $\rightarrow$ in $\mathcal{T}(P)$, it is easy to see that there is no $e'$ in $G.E$ such that $(e,e') \in G.\po{} \cup G.\rf{} \cup G.\mo{}$. Finally, the only event $e'$ for which $(e,e') \in G.\rf{}^{-1}$ could possibly happen is $e' = H.\momax{\loc(e)}$. However,
by definition there cannot be any $e''$ such that $(e',e'') \in G.\mo{}$ and so there cannot be any $e'$ such that $(e,e') \in G.\fr{}$.
Hence $e$ is a sequentially maximal event.

Next, we will show that any sequentially maximal event gives rise to a unique predecessor. Indeed, for any sequentially maximal event $e$ of $G$, let $H_e = G \setminus \{e\}$. Since $e$ is sequentially maximal, removing 
$e$ preserves all relation properties, i.e. $H_e.\rf{} \subseteq G.\rf{}$ and $H_e.\mo{}\subseteq G.\mo{}$. Thus, $H_e$ is a consistent execution graph. Note that $H_e \neq H_{e'}$ if $e \neq e'$. Furthermore $G$ can be obtained from $H_e$ by simply adding $e$ back in. Hence, the only thing we need to show is that $H_e$ is reachable.

To this end, since $H_e$ is sequentially consistent, we can pick a sequentially maximal event $e_1$ of $H_e$ and remove it from $H_e$ to get $H_{e_1}$, which will also be sequentially consistent, by the same argument as before. From there, we can remove another sequentially maximal event and so on and so forth, till we finally reach the initial execution graph $S_0$. This allows us to trace a path from $H_e$ all the way to $S_0$,
which proves that $H_e$ is reachable.
\section{Proof of Theorem \ref{th:stockmeyer}: A Sub-Exponential Time Approximate Counter}
\label{sec:stockmeyer}

%\begin{theorem}
 %   For any constants $r$ and $\rho$, there is a randomized $\tilde{O}(\sqrt{C(P)})$-time $(r, \rho)$-approximate counter.
%\end{theorem}

To describe our algorithm, we assume that we are given a concurrent program $P$, 
which defines the tree $\mathcal{D}(P)$. 
As mentioned before, if $d$ is the size of the program $P$, 
this tree has height $d^2$ and an inspection of the construction shows that the degree of each node is at most $d^2$. 
However, every node need not have degree $d^2$. In order to make this tree uniform, we add dummy nodes to each node of $\mathcal{D}(P)$ so that the degree of each node becomes $d^2$. 
Call the resulting tree $\mathcal{H}(P)$. Note that $\mathcal{H}(P)$ 
is the complete $d^2$-ary tree of height $d^2$. 
(We will never explicitly construct $\mathcal{D}(P)$ nor $\mathcal{H}(P)$ during our algorithm; 
these are just defined implicitly).

We define an \emph{inorder numbering} of $\mathcal{H}(P)$, which associates a subset of numbers to each node of $\mathcal{H}(P)$ recursively in the following way: We start at the root with a variable called count initialized to 0. For each $i$ from 1 to $d^2$, we recursively compute the inorder numbering of the $i^{th}$ subtree of the root. At the end of this computation, let the value of count be $n_i$. We increment $n_i$ by 1,
add $n_i+1$ to the inorder numbering of the root and move to the next iteration of the loop (i.e., to $i+1$). Note that the base case of the recursion is when we arrive at a leaf, in which case we will simply increment count by 1 and add this incremented value to the inorder numbering of the leaf and go back to its parent. In this way, we associate a subset of numbers to each node of $\mathcal{H}(P)$. 

For any node $u$, let $M(u)$ be the subset of numbers associated to $u$ by the inorder numbering. Let $S(u)$ and $B(u)$ be the smallest and biggest numbers appearing in $M(u)$ respectively. It is easy to see that $M(u)$ is a single number for each leaf
and otherwise, consists of exactly $d^2$ numbers. Furthermore, $M(u) \cap M(v) = \emptyset$
for any two distinct nodes $u,v$. 
Hence, from the construction, it follows that each number between $1$ and $N = (d^2)^{d^2} + (d^2)[{(d^2)^{d^2+1}} - 1]$ is assigned to exactly one node.

Now, for any pair of numbers $1 \le a \le b \le N$, let $[a,b]$ denote the interval of natural numbers between $a$ and $b$. Notice that if $M(u) \cap [a,b] \neq \emptyset$
and $M(v) \cap [a,b] \neq \emptyset$, then the least common ancestor $w$ of $u,v$ also satisfies $M(w) \cap [a,b] \neq \emptyset$. Hence, for each interval $[a,b]$, there
is a unique node $w_{a,b}$ such that $M(w_{[a,b]}) \cap [a,b] \neq \emptyset$ and 
every node $u$ satisfying $M(u) \cap [a,b] \neq \emptyset$ belongs to the subtree of $w_{[a,b]}$. 

Based on this notion of inorder numbering, we now present our estimation algorithm. Our estimation algorithm works on a parameter $\theta$ and consists of two acts: one where we operate under the assumption that the number of leaves of $\mathcal{D}(P) \le \theta$ and another where we operate under the assumption that the number of leaves of $\mathcal{D}(P) > \theta$.

\paragraph{Act One: Number of leaves of $\mathcal{D}(P) \le \theta$. }
In order to describe the first act of the algorithm, we need some notation and some small observations. Recall that any number between 1 and $N$ is assigned to exactly one node. Therefore, we can order the nodes of $\mathcal{H}(P)$ totally as follows:
$u \preceq v$ iff $S(u) < S(v)$. We now make the following series of observations.

\subparagraph*{Observation 1: } We can compute the inorder numbering $M(w)$ of any node as follows. We start by setting $u$ to be the root and a counter $x$ set to 0. Intuitively $x+1$ will denote the smallest number that can appear in the inorder numbering of any descendant of $u$.
At each step, we find the unique child $v$ of $u$ that contains $w$ as its descendant.
Suppose $v = w$. If $v$ is a leaf, then we simply return $\{x+1\}$.
If $v$ is an internal node at level $d-j$ (for some $j$), we return $\{x + N_{j-1} + 1, x + 2 N_{j-1} + 2, \dots, x + d^2 N_{j-1} + d^2\}$ where $N_{j-1} = (d^2)^{j-1} + d^2 [(d^2)^j - 1]$. Otherwise, if $v$ is the $i^{th}$ child for some $i \ge 1$ and $v$ is at level $d-j$ (for some $j$), then we set $x = x + (i-1) N_{j-1} + (i-1)$ and set $u$ to $v$ and continue. It can be easily checked that this procedure correctly computes $M(w)$ and 
runs in time $d^2$.

% Given any node $u$ of $\mathcal{H}(P)$ with inorder numbering $M(u)$, we can compute the inorder numbering $M(v)$ of any child $v$ of $u$ as follows. Indeed, let $v$ be the $i^{th}$ child of $u$ and let $x$ be the $(i-1)^{th}$ inorder number in $M(u)$ if $i-1 \ge 1$ and let $x$ be 0 otherwise. If $v$ is a leaf, then $M(v)$ is simply $\{x+1\}$. Otherwise, suppose $v$ is an internal node at level $d-j$.
% Let $x$ be the $(i-1)^{th}$ inorder number in $M(u)$ if $i - 1 \ge 1$ and let $x$ be 0 otherwise. Then, $M(v)$ is simply $\{x + N_{j-1} + 1, x + 2 N_{j-1} + 2, \dots, x + d^2 N_{j-1} + d^2\}$ where $N_{j-1} = (d^2)^{j-1} + d^2 (d^j) - d^2$.

\subparagraph*{Observation 2: } Given any node $u$ of $\mathcal{H}(P)$, we can check if 
there is a leaf of $\mathcal{D}(P)$ at the subtree rooted at $u$. Indeed, if $u$
corresponds to a dummy node, i.e., not a node of $\mathcal{D}(P)$, then
no such node will exist. On the other hand, if $u$ is a node of $\mathcal{D}(P)$
then indeed there will be a leaf of $\mathcal{D}(P)$ at the subtree rooted at $u$.
Hence, this check can be simply reduced to checking if $u$ is a node of $\mathcal{D}(P)$,
which can be done quite efficiently.

\subparagraph*{Observation 3: } Given any interval $I = [a,b]$, we can find the node 
$w_{I}$ as follows: We initially begin with the node $u$ set to the root and an interval $I'$ set to $[1,N]$. In each step, we will maintain the invariant that $w_I$ will be a descendant of $u$, $u = w_{I'}$ and also that $I \subseteq I'$. Indeed, this is true at the beginning.
Suppose at some step, we have reached $u$ with interval $I'$ such that $u = w_{I'}$
and $I \subseteq I'$. If $u$ is a leaf, then $w_I$ must be $u$ and so we are done.
Otherwise let $M(u)$ be $\{x_1,\dots,x_{d^2}\}$ and let $I' = [a',b']$. If $I$ is completely contained
inside any one of $I'_1 = [a'+1,x_1-1], I'_2 = [x_1+1,x_2-1], [x_2+1,x_3-1],\dots, I'_{d^2} = [x_{d^2-1}+1,b'-1]$, then we pick the unique $j$ such that $I$ is completely contained inside $I'_j$. Then we update $u$ to be the $j^{th}$ child of $u$ and update $I'$ to be $I'_j$. (Notice that the invariant is satisfied in this case). Otherwise,
it must be the case that $u = w_I$ and so we just output $u$. In this way, by going down the tree, in $d^2$ steps, we can find the node $w_I$.

Now, let $v_1,v_2,\dots$ be the leaves of $\mathcal{D}(P)$ (the \trust tree, not the complete $d^2$-ary tree) ordered according to the $\preceq$ ordering. Using the above mentioned observations and notations, we will now describe an algorithm which will find the nodes $v_1,v_2,\dots,v_{\theta}$. It does so in $\theta$ many phases, where in the $i^{th}$ phase it will find the node $v_i$.

The $i^{th}$ phase will operate in $d^2$ many subphases. At the start of each subphase
we will maintain a node $z$ of $\mathcal{D}(P)$ and an interval $I = [a,b]$ such that
$z = w_I$. (For the first subphase of the first phase, $z$ is taken to be the root and $I$ is taken to be $[1,N]$). For each subphase, we will always maintain the invariant that the node $v_i$ we are looking for will be a 
descendant of $z$, $S(v_i) \in I$ and $v_{i-1}$ will not be a descendant of $z$. (Notice that this invariant is definitely satisfied at the beginning of the first subphase of the first phase).

In the $j^{th}$ subphase, given the node $z$ (with corresponding inorder numbering $M(z)$) and the interval $I = [a,b]$, we proceed as follows. If $a = b$, then by our invariants, it follows that $v_i = z$ and so we return $z$. Otherwise, we loop over all of the $d^2$ children of $z$ (say $u_1,u_2,\dots,u_{d^2}$ such that $u_1 \preceq u_2 \preceq \dots \preceq u_{d^2}$). 
When looping over $u_k$, we first compute $M(u_k)$ using Observation 1. If $M(u_k) \cap I = \emptyset$, then we move to $u_{k+1}$ in the loop. Further, using Observation 2, we check if there is some node of $\mathcal{D}(P)$ that is present in the subtree rooted at $u_k$. If not, then also we move to $u_{k+1}$ in the loop. 
Else, it immediately follows that $v_i$ must be a descendant of $u_k$. 
Hence, we move to the $(j+1)^{th}$ subphase with the node $u_k$ and compute the new interval $I'$ as follows: Let $x$ be the $(k-1)^{th}$ number appearing in $M(z)$ if $k - 1 \ge 1$, otherwise let $x$ be $a$. Let $y$ be the $(k+1)^{th}$ number appearing in $M(z)$ if $k+1 \le d^2$, otherwise let $y$ be $b$. We then set $I'$ to be $[x,y]$. It can be easily verified that $u_k = w_{I'}$. Furthermore because $v_i$ must be a descendant of $u_k$, it follows that $S(v_i) \in I'$. Hence the invariants for the next subphase are satisfied.

Note that in each subphase, we are going one level down the tree. Hence there can be at most $d^2$ subphases in each phase. Furthermore, due to our invariants, it follows
that when the $i^{th}$ phase ends, we have arrived at the node $v_i$ with an interval $I$ of the form $[a,a]$. We then start the $(i+1)^{th}$ phase by setting our interval to be $[a+1,N]$ and computing the corresponding node $w_{a+1,N}$ using Observation 3. 
It easily follows that our invariant is satisfied at the beginning of the $(i+1)^{th}$ phase. This proves that the algorithm finds the nodes $v_1,v_2,\dots,v_{\theta}$.

Let us analyse the running time of this algorithm. Notice that each subphase of each phase
takes at most $O(d^4)$ time, in order to go over all of the possible children of $z$ and compute each of their inorder numberings. Since there are $d^2$ subphases in total and $\theta$ phases overall, it follows that the total amount of time spent in this act is $O(d^6 \theta)$. Note that if the number of leaves of $\mathcal{D}(P)$ is indeed
at most $\theta$, then our algorithm will stop with this act and output all of them.
Since it uses no randomization, in this case, it will output the exact number of leaves
of $\mathcal{D}(P)$ with zero variance.

\paragraph{Second Act: Number of leaves of $\mathcal{D}(P) > \theta$. }

Suppose we executed the first act of our algorithm to completion. It follows that there
are at least $\theta$ many leaves of $\mathcal{D}(P)$. At this point, we initiate the second act of our algorithm. Let $M = ((d^2)^{d^2+1})-1$ be the total number of nodes
of $\mathcal{H}(P)$. We take $z$ independent random walks and count the number of times we hit a maximal execution graph (say $y$). Then we output $(M/z) \cdot y$. 

Let $X$ be the random variable corresponding to the value that we output.
It is easy to see that the expected value of $X$ is the precise number
of leaves of $\mathcal{D}(P)$. Let us now count the variance of this approach, noting that
we are guaranteed that the number of leaves of $\mathcal{D}(P)$ is at least $\theta$
and therefore the probability of hitting a maximal execution graph along a random walk
is at least $p = \theta/M$. 

Let $Y$ be the random variable which counts the number of successful samples from $z$ independent random walks. Since this is essentially a binomial distribution it follows that
the variance of $Y$ is $zp(1-p)$. Hence, the variance of $X$ is the variance of
$\frac{M}{z} \cdot Y$ which is $\frac{M^2}{z^2} \cdot z p (1-p) = \frac{M^2 p (1-p)}{z}$. 
Let $\sigma$ be the square root of the variance of $X$. 

Now, suppose we are given numbers $r, \rho$ and we have to choose a $z$ so that the random variable $X$ is at most a factor of $r$ away from the total number of leaves of $\mathcal{D}(P)$ with probability at least $1-\rho$, i.e., we want to choose a $z$ so that $\mathbb{P}(X \ge r \mu )$ is at most $\rho$. Taking $z = \frac{1}{\rho} \cdot \frac{1}{(r-1)^2} \cdot \left(\frac{M}{\theta} - 1 \right)$, by Chebyshev's inequality we get that $$\mathbb{P}(X \ge \frac{1}{\sqrt{\rho}} \sigma + \mu ) \le \rho$$ 
Unraveling the expression $\frac{1}{\sqrt{\rho}} \sigma + \mu$ with the chosen value of $z$,
we get that it is equal to $r \mu$. Hence, the probability that $X \ge r \mu$ is at most $\rho$, as desired.

Note that the running time of the second act of the algorithm primarily depends only on $z$. Further note that $z = \alpha M/\theta$ where $\alpha$ is a constant depending only on $r, \rho$.
It follows that if we set $\theta = M^{1/2}$, then the second act of the algorithm runs in time $O(M^{1/2})$.

\paragraph{Wrapping up: } Overall, we see that if we set $\theta = M^{1/2}$, then
for any numbers $r, \rho$, our algorithm runs in time $O(d^6 M^{1/2})$ and returns
a value that is at most a factor of $r$ away from the leaves of $\mathcal{D}(P)$ with probability at least $1 - \rho$. Hence, at the cost of a sub-exponential running time,
we obtain a $(r,\rho)$-approximate counter.

% Furthermore, if we are guaranteed that the total number of leaves of $\mathcal{D}(P)$
% is at most a polynomial of the size of $P$, then just setting $\theta$ to be this polynomial value and running just the first 
% act, gives the correct total value
% of $\mathcal{D}(P)$, i.e., we get an estimator with zero variance. The running time in this case is polynomial and is roughly 
% of the order of $O(d^6\mathit{poly}(|P|))$. 

%\input{se-trust}

\end{document}